\documentclass[11pt,a4paper]{article}

\usepackage[utf8]{inputenc}
\usepackage[french,english]{babel}
\usepackage[T1]{fontenc}
\usepackage[mono=false]{libertine}
\usepackage{microtype}

\usepackage[margin=2.cm]{geometry}
\usepackage{amsthm}
\usepackage[cmintegrals,libertine]{newtxmath}

\usepackage[cal=esstix,scr=boondoxo]{mathalfa}
\usepackage{booktabs}
\useosf

\usepackage{wrapfig,flafter,xcolor}
\usepackage{float}

\usepackage{tabu}
\usepackage[font=small]{caption}
\usepackage{subcaption}

\usepackage{numprint}

\usepackage{enumerate}
\usepackage[shortlabels]{enumitem}

\usepackage{graphicx}
\graphicspath{{./Figures/}}

\usepackage{hyperref}
\hypersetup{
	colorlinks=true,
	citecolor=blue!60!black,
	linkcolor=red,
	urlcolor=green!40!black,
	filecolor=yellow!50!black,
	breaklinks=true,
	pdfpagemode=UseNone,
	bookmarksopen=false,
}

\usepackage{./mathsymb-linxiao}
\usepackage{./thm-linxiao}

\usepackage{ifthen}

\newcommand*{\+}{\ensuremath{\text{\rm\texttt{+}}}}
\newcommand*{\<}{\ensuremath{\text{\rm\texttt{-}}}}
\newcommand*{\jj}{{\text{\normalsize\texttt{\textpm}}}}
\newcommand{\R}{\mathbb{R}}

\newcommand{\N}{\mathbb{N}}
\newcommand{\Z}{\mathbb{Z}}

\usepackage[normalem]{ulem}

\theoremstyle{plain}
\newcommand*{\Map}[1]{\mathfrak{#1}} 
\newcommand{\emapo}{\emap^\circ}
\newcommand{\frontier}{\partial\emap}

\newcommand*{\algo}{\mathcal{A}}
\newcommand*{\map}{\Map{m}}  
\newcommand*{\tmap}{\Map{t}} 
\newcommand*{\emap}{\Map{e}} 
\newcommand*{\umap}{\Map{u}} 
\newcommand*{\bmap}{\Map{b}} 

\newcommand{\Probh}{\hat{\Prob}}
\newcommand*{\law}{\mathcal{L}}

\newcommand*{\steps}{\mathcal{S}}
\newcommand*{\step}{\mathtt{s}}
\newcommand*{\Step}{\mathtt{S}}
\newcommand*{\CC}{\mathtt{C}}
\newcommand*{\LL}{\mathtt{L}}		\newcommand*{\RR}{\mathtt{R}}
\newcommand*{\cp}{\CC^\+}			\newcommand*{\cm}{\CC^\<}
\newcommand*{\lp}[1][k]{\LL^\+_{#1}}	\newcommand*{\lm}[1][k]{\LL^\<_{#1}}
\newcommand*{\rp}[1][k]{\RR^\+_{#1}}	\newcommand*{\rn}[1][k]{\RR^\<_{#1}}
\newcommand*{\iroot}{\mathcal I}
\newcommand*{\bt}[1][]{(\tmap#1,\sigma#1)}
\newcommand*{\btsq}[1][]{[\tmap#1,\sigma#1]}
\newcommand*{\bts}{\mathcal{BT}}

\newcommand*{\zz}[2][p,q+1]{\,\frac{z_{#2}}{z_{#1}}}
\newcommand*{\zzz}[3][p,q+1]{\,\frac{z_{#2}\,z_{#3}}{z_{#1}}}

\newcommand*{\zzh}[2][p+1,q]{\,\frac{z_{#2}}{z_{#1}}}
\newcommand*{\zzzh}[3][p+1,q]{\,\frac{z_{#2}\,z_{#3}}{z_{#1}}}

\newcommand{\nseq}[2][0]{(#2_n)_{n\ge #1}}
\newcommand{\limsupp}{\limsup_{p,q\to\infty}}
\newcommand{\pqy}{_{p,q}}
\newcommand{\yy}{_{\infty}}

\newcommand{\refp}[2]{\hyperref[#2]{\ref*{#2}(#1)}}

\newcommand{\RRP}{Rational parametrization}
\newcommand{\RP}{rational parametrization}

\usepackage{esint}
\newcommand{\mc}{\mathcal}

\newcommand*{\disk}{\mathbb D}
\newcommand*{\cdisk}{\overline{\mathbb D}}

\usepackage{relsize}
\newcommand{\Ddom }[1][\theta]{\boldsymbol \Delta\,\!_{\epsilon,#1}}
\newcommand{\cDdom}[1][\theta]{\overline{\boldsymbol \Delta}\,\!_{\epsilon,#1}}
\newcommand{\slit }[1][\epsilon]{\boldsymbol \Delta\,\!_{#1}}
\newcommand{\cslit}[1][\epsilon]{\overline{\boldsymbol \Delta}\,\!_{#1}}

\newcommand{\Hdom}[1][\epsilon]{\mc H\!\,_{#1}}
\newcommand{\cHdom}[1][\epsilon]{\overline{\mc H}\!\,_{#1}}

\newcommand{\code}{\texttt}
\usepackage{ifthen}

\newcommand{\adot}{\,\cdot\,}
\renewcommand{\Re}{\mathfrak{Re}}
\renewcommand{\Im}{\mathfrak{Im}}

\title{Ising model on random triangulations of the disk: phase transition}

\author{Linxiao Chen\footnote{Department of Mathematics, ETH Z\"urich, R\"amistr.\ 101, 8092 Z\"urich, Switzerland, linxiao.chen$0$@gmail.com}\ ,\
Joonas Turunen\footnote{ENSL, Univ Lyon, CNRS, Laboratoire de Physique, F-69342 Lyon, France, joonas.turunen@alumni.helsinki.fi}}
\date{}

\begin{document}

\maketitle
\begin{abstract}
In (\emph{Commun. Math. Phys.} 374(3):1577–1643, 2020), we have studied the Boltzmann random triangulation of the disk coupled to an Ising model on its faces with Dobrushin boundary condition at its critical temperature. In this paper, we investigate the phase transition of this model by extending our previous results to arbitrary temperature: We compute the partition function of the model at all temperatures, and derive several critical exponents associated with the infinite perimeter limit. We show that the model has a local limit at any temperature, whose properties depend drastically on the temperature. At high temperatures, the local limit is reminiscent of the uniform infinite half-planar triangulation (UIHPT) decorated with a subcritical percolation. At low temperatures, the local limit develops a bottleneck of finite width due to the energy cost of the main Ising interface between the two spin clusters imposed by the Dobrushin boundary condition. This change can be summarized by a novel order parameter with a nice geometric meaning.
In addition to the phase transition, we also generalize our construction of the local limit from the two-step asymptotic regime used in (\emph{Commun. Math. Phys.} 374(3):1577–1643, 2020)  to a more natural diagonal asymptotic regime. We obtain in this regime a scaling limit related to the length of the main Ising interface, which coincides with predictions from the continuum theory of \emph{quantum surfaces} (a.k.a.\ Liouville quantum gravity).
\end{abstract}

\newcommand*{\g}{\Map{g}}
\newcommand*{\B}{\Map{b}}

\section{Introduction}
The two-dimensional Ising model is one of the simplest statistical physics models to exhibit a phase transition. We refer to \cite{MW73} for a comprehensive introduction. The systematic study of the Ising model on random two-dimentional lattices dates back to the pioneer works of Boulatov and Kazakov \cite{Kaz86,BouKaz87}, where they discovered a third order phase transition in the free energy density of the model, and computed the associated critical exponents. In their work, the partition function of the model was computed in the thermodynamic limit using matrix integral methods applied to the so-called \emph{two-matrix model}, see \cite{LZ04} for a mathematical introduction. Since then, this approach has been pursued and further generalized to treat other statistical physics models on random lattices, see e.g.\ \cite{EynOra05,EynBon99}.

In this paper, we will follow a more combinatorial approach to the model originated from a series of works by Tutte (see \cite{Tut95} and the references therein) on the enumeration of various classes of embedded planar graphs known as \emph{planar maps}, which is essentially another name for the random lattices studied in physics. The approach of Tutte utilizes a type of recursive decomposition satisfied by these classes of planar maps to derive a functional equation that characterizes their generating function. This method was later generalized by Bernardi and Bousquet-M\'elou \cite{BBM11,BBM16} to treat bicolored planar maps with a weighting that is equivalent to the Ising model. Before that, Bousquet-M\'elou and Schaeffer already had studied the Ising model on planar maps using some general bijection between bipartite maps and blossoming trees \cite{BMS02}. Another work of Bouttier, Di Francesco and Guitter also studied Ising model on quadrangulations using bijections between Eulerian maps and mobiles \cite{BDG07}.

From a probabilistic point of view, the aforementioned recursive decomposition can be seen as the operation of removing one edge from an (Ising-decorated) random planar map with a boundary, and observing the resulting changes to the boundary condition. By iterating this operation, one obtains a random process, called the \emph{peeling process}, that explores the random map one face at a time. Ideas of such exploration processes have their roots in the physics literature \cite{Wat95}, and was revisited and popularized by Angel in \cite{Ang02}. The peeling process proves to be a valuable tool for understanding the geometry of random planar maps without Ising model, see \cite{CurPeccot} for a review of recent developments.

In our previous article \cite{CT20}, we extended some enumeration results of Bernardi and Bousquet-M\'elou \cite{BBM11} to study the Ising-decorated random triangulations with Dobrushin boundary condition \emph{at its critical temperature}. We used the peeling process to construct the local limit of the model, and to obtain several scaling limit results concerning the lengths of some Ising interfaces. In this paper, we extend similar results to the model \emph{at any temperature}, and show how the large scale geometry of Ising-decorated random triangulations changes qualitatively at the critical temperature. In particular, our results confirm the physical intuition that, at large scale, Ising-decorated random maps at non-critical temperatures behave like non-decorated random maps.

A similar model of Ising-decorated triangulations (more precisely, a model dual to ours) has been studied in a recent work of Albenque, M\'enard and Schaeffer \cite{AMS18}. They followed an approach reminiscent of Angel and Schramm in \cite{AS03} to show that the model has a local limit at any temperature, and obtained several properties of the limit such as one-endedness and recurrence for a range of temperatures. However, they studied the model without boundary, and hence did not encounter the geometric consequences of the phase transition in terms of the infinite Ising interface. In the recent preprint \cite{AM22}, the first two of the aforementioned authors proved several exact results on the perimeter and volume of the spin clusters, demonstrating the phase transition through several critical exponents and geometric behaviors of the cluster in different phases. The model with spins on the vertices can also be studied with a boundary, and the methods introduced in \cite{CT20} and this article were recently applied to that model in \cite{T20} by the second author of this work.

\medskip

We start by recalling some essential definitions from \cite{CT20}.

\begin{figure}[t!]
\centering
\includegraphics[scale=0.9]{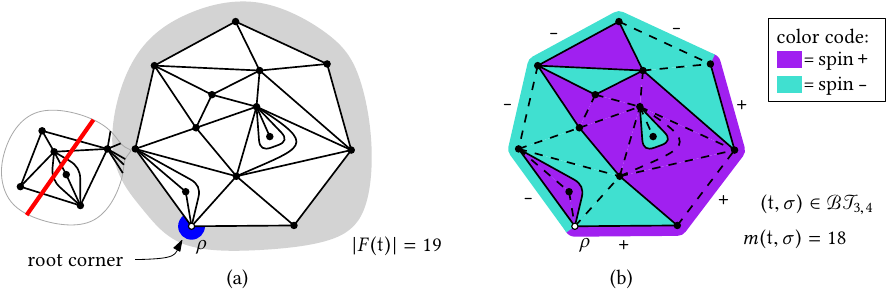}
\caption{
(a) A triangulation $\tmap$ of the 7-gon with 19 internal faces. The boundary will no longer be simple if one attaches to $\tmap$ the map inside the bubble to its left.
(b) an Ising-triangulation of the $(3,4)$-gon with 18 monochromatic edges (dashed lines).
}
\label{fig:def-map}
\end{figure}

\paragraph{Planar maps.}
Recall that a (finite) \emph{planar map} is a proper embedding of a finite connected graph into  the sphere $\mathbb{S}^2$, viewed up to orientation-preserving homeomorphisms of $\mathbb{S}^2$. Loops and multiple edges are allowed in the graph.
A \emph{rooted map} is a map equipped with a distinguished corner, called the \emph{root corner}.
\begin{center} \vspace{-1.ex}
	\emph{All maps in this paper are assumed to be planar and rooted.}
\end{center}   \vspace{-1.ex}
In a (rooted planar) map $\map$, the vertex incident to the root corner is called the \emph{root vertex} and denoted by $\rho$. The face incident to the root corner is called the \emph{external face}, and all other faces are \emph{internal faces}. We denote by $F(\map)$ the set of internal faces of $\map$.

A map is a \emph{triangulation of the $\ell$-gon} ($\ell\ge 1$) if its internal faces all have degree three, and the boundary of its external face is a simple closed path (i.e.\ it visits each vertex at most once) of length $\ell$. The number $\ell$ is called the \emph{perimeter} of the triangulation. Figure~\refp{a}{fig:def-map} gives an example of a triangulation of the $7$-gon.

\paragraph{Ising-triangulations with Dobrushin boundary conditions.}
We consider the Ising model with spins on the internal \emph{faces} of a triangulation of a polygon. A triangulation together with an Ising spin configuration on it is written as a pair $\bt$, where $\sigma \in \{\+,\<\}^{F(\tmap)}$. Observe that $\sigma$ can also be viewed as a coloring, and by combinatorial convention, we sometimes refer to it as such.
An edge $e$ of $\tmap$ is said to be \emph{monochromatic} if the spins on both sides of $e$ are the same. When $e$ is a boundary edge, this definition requires a boundary condition which specifies a spin outside each boundary edge. By an abuse of notation, we consider the information about the boundary condition to be contained in the coloring $\sigma$, and denote by $m\bt$ the number of monochromatic edges in $\bt$.

In this work we consider the \emph{Dobrushin boundary conditions} under which the spins \emph{outside} the boundary edges are given by a sequence of the form $\+^p \<^q$ ($p$ $\+$'s followed by $q$ $\<$'s, where $p,q\ge 0$ are integers and $p+q\ge 1$ is the perimeter of the triangulation) in the clockwise order from the root edge.
We call a pair $\bt$ with this boundary condition an \emph{Ising-triangulation of the $(p,q)$-gon}, or a \emph{bicolored triangulation of the $(p,q)$-gon}. Figure~\refp{b}{fig:def-map} gives an example in the case $p=3$ and $q=4$. We denote by $\bts_{p,q}$ the set of all Ising-triangulations of the $(p,q)$-gon.
For $\nu>0$, let
\begin{equation*}
	z_{p,q}(t,\nu) = \sum_{(\tmap,\sigma)\in \mathcal{BT}_{p,q}} \nu^{m(\tmap,\sigma)} t^{\abs{F(\tmap)}}
\end{equation*}
When $z_{p,q}(t,\nu)<\infty$, we can
define a probability distribution $\prob_{p,q}^{t,\nu}$ on $\mathcal{BT}_{p,q}$ by
\begin{equation*}
\prob_{p,q}^{t,\nu}(\tmap,\sigma) = \frac{ t^{\abs{F(\tmap)}} \nu^{m(\tmap,\sigma)} }{z_{p,q}(t,\nu)}
\end{equation*}
for all $(\tmap,\sigma)\in \mathcal{BT}_{p,q}$. A random variable of law $\prob_{p,q}^{t,\nu}$ will be called a \emph{Boltzmann Ising-triangulation of the $(p,q)$-gon}.
We collect the partition functions $( z_{p,q}(t,\nu) )_{p,q\ge 0}$ into the following generating series:
\begin{equation*}
Z_q(u,t,\nu) = \sum_{p=0}^\infty z_{p,q}(t,\nu)\, u^p
\qtq{and}
Z(u,v,t,\nu) = \sum_{p,q\ge 0} z_{p,q}(t,\nu)\, u^p v^q = \sum_{q=0}^\infty Z_q(u,t,\nu) v^q \,,
\end{equation*}
where by convention $z_{0,0} = 1$.

\paragraph{Partition functions and the phase diagram.}
The condition $z_{p,q}(t,\nu)<\infty$ does not depend on $(p,q)$:
For any pairs $(p,q),(p',q')\ne (0,0)$, one can construct an annulus of triangles which, when glued around \emph{any} bicolored triangulation of the $(p,q)$-gon, gives a bicolored triangulation of the $(p',q')$-gon. Thus $z_{p,q}(t,\nu)\le C\cdot z_{p',q'}(t,\nu)$, where $C$ is the weight of the annulus.
It has been shown in \cite[Section 12.2]{BBM11} that for all $\nu>1$, the series $t\mapsto z_{1,0}(t,\nu)$ converges at its radius of convergence $t_c(\nu)$. Then the above argument implies that $t_c(\nu)$ is the radius of convergence of $t\mapsto z_{p,q}(t,\nu)$ and we have $z_{p,q}(t_c(\nu),\nu)<\infty$, for all $(p,q)\ne (0,0)$ and $\nu>1$.
In this paper we always restrict ourselves to the case $\nu>1$. (This is called the \emph{ferromagnetic} case since in this case the weight $\nu^{m(\tmap,\sigma)}$ favors neighboring spins to have the same sign.)

We shall call $t_c(\nu)$ the critical line of the Boltzmann Ising-triangulation. It separates the inadmissible region $t>t_c(\nu)$, where the probabilistic model is not well-defined, from the subcritical region $t<t_c(\nu)$, where the probability for a Boltzmann Ising-triangulation to have size $n$ decays exponentially with $n$. (Here the size of an Ising-triangulation is defined as its number of internal faces.) It has also been shown in \cite{BBM11} that the function $t_c(\nu)$ is analytic everywhere on $(1,\infty)$ except at $\nu_c=1+2\sqrt 7$. This further divides the critical line into three phases: the high temperature phase $1<\nu<\nu_c$, the critical temperature $\nu=\nu_c$, and the low temperature phase $\nu>\nu_c$ (Figure~\ref{fig:tc_nu_plot}).

\begin{figure}[h!]
\centering
\includegraphics[scale=1.2]{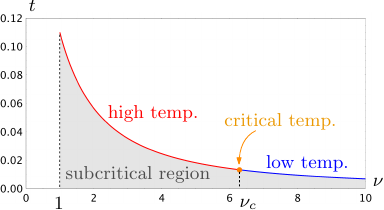}
\caption{Phase diagram of the Boltzmann Ising-triangulation for $\nu>1$. The critical line $t_c(\nu)$ is divided by $\nu_c=1+2\sqrt7$ into the high temperature, low temperature and critical temperature phases. Although hardly visible in the graph, the third derivative of $t_c(\nu)$ has a discontinuity at $\nu=\nu_c$.}
\label{fig:tc_nu_plot}
\end{figure}

In our previous paper \cite{CT20}, we studied the model at the critical point $(\nu,t)=(\nu_c,t_c(\nu_c))$. Results in \cite{CT20} include an explicit parametrization of $Z(u,v,t_c(\nu_c),\nu_c)$, the asymptotics of $z_{p,q}(t_c(\nu_c),\nu_c)$ when $q\to \infty$ and then $p\to \infty$, a scaling limit result closely related to the main interface length, and the local limit of the whole triangulation in that asymptotic regime. In this paper, we will extend this study to the critical line $t=t_c(\nu)$ in order to shed more light on the nature of the phase transition at $\nu=\nu_c$. For this reason we will write throughout this paper
\begin{equation*}
z_{p,q}(\nu)=z_{p,q}(t_c(\nu),\nu) \,, \qquad
Z_q(u,\nu)=Z_q(u,t_c(\nu),\nu)
\qtq{and}
Z(u,v,\nu)=Z(u,v,t_c(\nu),\nu) \,.
\end{equation*}

In \cite{CT20}, we have characterized $Z(u,v,t,\nu)$ as the solution of a functional equation, and solved it in the case of $(\nu,t)=(\nu_c,t_c(\nu))$. In this paper we solve the equation for general $(\nu,t)$ and give the solution in terms of a multivariate \RP:

\begin{theorem}[\RRP\ of $Z(u,v,t,\nu)$]\label{thm:Zparam}
For $\nu>1$, $Z(u,v,t,\nu)$ satisfies the parametric equation
\begin{equation}\label{eq:RP:Z General}
t^2      = \hat T(S,\nu), \qquad
t\cdot u = \hat U(H,S,\nu),\qquad
t\cdot v = \hat U(K,S,\nu) \qtq{and}
Z(u,v,t,\nu) = \hat Z(H,K,S,\nu) \,,
\end{equation}
where $\hat T$, $\hat U$ and $\hat Z$ are rational functions whose explicit expressions are given in Lemma~\ref{lem:RP Z0} and in \cite{CAS2}.
\end{theorem}

To specialize the above \RP\ of $Z(u,v,t,\nu)$ to the critical line $t=t_c(\nu)$, one needs to replace the parameter $S$ by its value $S_c(\nu)$ that parametrizes $t=t_c(\nu)$. It turns out that the function $S_c(\nu)$ itself has \RP s on $(1,\nu_c)$ and $(\nu_c,\infty)$, respectively. More precisely, $S_c(\nu)$ satisfies a parametric equation of the form
\begin{equation*}
\nu=\check \nu(R) \qtq{and} S_c(\nu)=\check S(R) \,,
\end{equation*}
where $\check \nu(R)$ and $\check S(R)$ are piecewise rational functions on the intervals $(R_1,R_c]$ and $[R_c,\infty)$, where the values $R_1,R_c,R_\infty$ correspond to $\nu=1,\nu_c,\infty$ in the sense that $\check \nu(R_1)=1$, $\check \nu(R_c)=\nu_c$ and $\check \nu(R_\infty)=\infty$. The expressions of $\check \nu(R)$, $\check S(R)$ and of $R_1,R_c,R_\infty$ are given in Section~\ref{sec:specialization of RP}. By making the substitution $\nu=\check \nu(R)$ and $S=\check S(R)$ in \eqref{eq:RP:Z General}, we obtain a piecewise \RP\ of $t_c(\nu)$ and $Z(u,v,\nu)$ of the form
\begin{equation*}
t_c(\nu)^2 = \check T(R), \qquad
t_c(\nu)\cdot u = \check U(H,R),\qquad
t_c(\nu)\cdot v = \check U(K,R),\qtq{and}
Z(u,v,\nu) = \check Z(H,K,R) .
\end{equation*}
See Section~\ref{sec:specialization of RP} for more details.

In \cite{CT20}, we computed the asymptotics of $z_{p,q}(t,\nu)$ when $(\nu,t)=(\nu_c,t_c(\nu_c))$ in the limit where $p\to \infty$ after $q\to \infty$. The following theorem extends this result to the whole critical line $t=t_c(\nu)$, and also to the limit where $p,q\to \infty$ at comparable speeds. 
These results are obtained by a close examination of the singular expansion of the multivariate generating function $Z(u,v,t_c(\nu),\nu)$ (in particular, by proving that $(u,v)\mapsto Z(u,v,t_c(\nu),\nu)$ is analytic in a product of two $\Delta$-domains), see Sections~\ref{sec:singularity structure}--\ref{sec:coeff asymp}. Similar methods have been applied to more complicated generating functions and made partly systematic in two recent works \cite{Chen2021,Chen2022} of the first author.

\begin{theorem}[Asymptotics of $z_{p,q}(\nu)$]\label{thm:asympt}
For any fixed $\nu>1$ and $0<\lambda_{\min}<\lambda_{\max}<\infty$, we have
\begin{align*}
u_c(\nu)^q \cdot z_{p,q}(\nu) & = \frac{a_p(\nu)}{\Gamma(-\alpha_0)} \cdot q^{-(\alpha_0+1)}  +  O\m({ q^{-(\alpha_0+1+\delta)} }
          &&\text{as } q\to\infty \text{ for each fixed }p\ge 0,   \\
u_c(\nu)^p \cdot~~\, a_p(\nu) & = \frac{b  (\nu)}{\Gamma(-\alpha_1)} \cdot p^{-(\alpha_1+1)}  +  O\m({ p^{-(\alpha_1+1+\delta)} }
          &&\text{as } p\to\infty,    \\
 u_c(\nu)^{p+q} \cdot z_{p,q}(\nu) & = \frac{b(\nu)\cdot c(q/p)}{\Gamma(-\alpha_0) \Gamma(-\alpha_1)} \cdot p^{-(\alpha_2+2)}  +  O\m({ p^{-(\alpha_2+2+\delta)} }
          &&\text{as } p,q\to\infty \text{ while }q/p \in [\lambda_{\min}, \lambda_{\max}],
\end{align*}
where the exponents $\alpha_i$, $\delta$ and the scaling function $c(\lambda)$ only depend on the phase of the model, and are given by
\begin{equation*}
\begin{tabular}{|c|c|c|c|c|}
\hline        & $\alpha_0$ & $\alpha_1$ & $\alpha_2$ & $\delta$ \\\hline
$\nu>\nu_c$        & $3/2$ & $3/2$      & $3$        & $1/2$ \\\hline
$\nu=\nu_c$        & $4/3$ & $1/3$      & $5/3$      & $1/3$ \\\hline
$\nu\in (1,\nu_c)$ & $3/2$ & $ -1$      & $1/2$      & $1/2$ \\\hline
\end{tabular}
\qquad
c(\lambda) = \begin{cases}
\lambda^{-5/2}   & \text{when }\nu>\nu_c \\
\frac43 \int_0^\infty (1+r)^{-7/3}(\lambda+r)^{-7/3} \dd r  & \text{when }\nu=\nu_c \\
(1+\lambda)^{-5/2} & \text{when }\nu\in (1,\nu_c)\,.
\end{cases}
\end{equation*}
On the other hand, $u_c(\nu)$, $a_p(\nu)$ (for $p\ge 0$) and $b(\nu)$ are analytic functions of $\nu$ on $(1,\nu_c)$ and $(\nu_c,\infty)$, respectively. And $u_c(\nu)$ is continuous at $\nu=\nu_c$. An explicit parametrization of $u_c(\nu)$ is given in Section~\ref{sec:dom of cvg}. Parametrizations of $b(\nu)$ and of the generating function $A(u,\nu) := \sum_{p} a_p(\nu) u^p$ are explained in Section~\ref{sec:local expansion} and given in \cite{CAS2}.
\end{theorem}

\begin{remark}
The exponents $\alpha_i$ and the scaling function $c(\lambda)$ satisfy a number of consistency relations.

First, one can exchange the roles of $p$ and $q$ in the last asymptotics of Theorem~\ref{thm:asympt}. Since we have $z_{p,q}=z_{q,p}$ for all $p,q$, this implies that $c(\lambda)\lambda^{\alpha_2+2}=c(\lambda^{-1})$ or, in a more symmetric form, $c(\lambda)\lambda^{(\alpha_2+2)/2}=c(\lambda^{-1}) \lambda^{-(\alpha_2+2)/2}$.

By replacing the factor $a_p(\nu)$ in the first asymptotics of Theorem~\ref{thm:asympt} with the dominant term in the second asymptotics, we obtain \emph{heuristically} that
\begin{equation*}
 u_c(\nu)^{p+q} \cdot z_{p,q}(\nu) \sim \frac{b(\nu)\cdot (q/p)^{-(\alpha_0+1)}}{\Gamma(-\alpha_0) \Gamma(-\alpha_1)} \cdot p^{-(\alpha_0+\alpha_1+2)}
          \qt{when } p,q\to\infty \text{ and }q\gg p \,.
\end{equation*}
This suggests that $\alpha_0+\alpha_1=\alpha_2$ and $c(\lambda)\sim \lambda^{-(\alpha_0+1)}$ when $\lambda \to \infty$. One can verify that both relations are indeed satisfied by $\alpha_i$ and $c(\lambda)$ in all three phases. Notice that thanks to the equation $c(\lambda)\lambda^{\alpha_2+2}=c(\lambda^{-1})$, the asymptotics $c(\lambda)\eqv\lambda \lambda^{-(\alpha_0+1)}$ is equivalent to $c(\lambda)\eqv[0]{\lambda}\lambda^{-(\alpha_1+1)}$.
%
\end{remark}

\paragraph{Infinite Ising-triangulations and local limits.}
Infinite bicolored triangulations are defined as the local limits of finite bicolored triangulations. Formally, the \emph{local distance} between two bicolored triangulations $\bt$ and $\bt[']$ is defined by
\begin{equation*}
d\1{loc}(\bt,\bt[']) = 2^{-R}\qtq{where}
	R = \sup\Set{r\geq 0}{ \btsq_r=\btsq[']_r }
\end{equation*}
and $\btsq_r$ denotes the ball of radius $r$ around the origin in $\bt$ which takes into account the colors of the faces. The set $\bts$ of (finite) bicolored triangulations of a polygon is a metric space under $d\1{loc}$. We denote its Cauchy completion by $\overline{\bts}$ and define the set of infinite bicolored triangulations as $\overline{\bts}\setminus\bts$.
We recall from graph theory that an infinite graph is \emph{$k$-ended} if the complement of any finite subgraph has at most $k$ infinite connected components \cite[14.2]{topograph}, and the same notion naturally extends to maps by considering their underlying graphs. 
We denote by $\bts_\infty^{(1)}$ the set of one-ended (infinite) bicolored triangulations with an external face of infinite degree. The elements of $\bts_\infty^{(1)}$ are called \emph{bicolored triangulations of the half plane}, since they have a proper embedding without accumulation points in the upper half plane such that the boundary coincides with the real axis.
Moreover, let $\bts_\infty^{(2)}$ be the set of two-ended  bicolored triangulations with an external face of infinite degree.

\begin{theorem}[Local limits of Ising-triangulations]~\label{thm:cv}\\
For every $\nu>1,$ there exist probability distributions $\prob_p^\nu$ and $\prob_\infty^\nu$, such that
\begin{equation}\label{eq:2-step cv}
\prob_{p,q}^\nu\ \cv[] q\ \prob_p^\nu\ \cv[] p\ \prob_\infty^\nu
\end{equation}
locally in distribution. Moreover, $\prob_p^\nu$ is supported on $\bts_\infty^{(1)}$ for all $\nu>1$, whereas $\prob_\infty^\nu$ is supported on $\bts_\infty^{(1)}$ when $1<\nu\leq\nu_c$ and on $\bts_\infty^{(2)}$ when $\nu>\nu_c$. In addition, for any $0<\lambda'\leq 1\leq\lambda<\infty$, when $\frac{q}{p}\in [\lambda',\lambda]$, we have
\begin{equation}\label{eq:diag cv}
\prob_{p,q}^\nu\ \cv[]{p,q}\ \prob_\infty^\nu
\end{equation}
locally in distribution.
\end{theorem}

This theorem generalizes our previous result \cite[Theorem~4]{CT20}, which contained only the convergence \eqref{eq:2-step cv} at $\nu=\nu_c$.
It also partially confirms a conjecture in \cite{CT20}, which states that $\prob_{p,q}^{\nu_c} \to \prob_\infty^{\nu_c}$ locally in distribution whenever $p,q\to \infty$.

\paragraph{Peeling process and perimeter processes.}
Recall that we consider bicolored triangulations $\bt$ with a Dobrushin boundary condition. We denote by $\rho$ the root vertex of $\bt$, and by $\rho^\dagger$ the other boundary vertex where the boundary condition changes sign.

An \emph{interface} in $\bt$ is a path on $\tmap$ formed by non-monochromatic edges. Due to the Dobrushin boundary condition, the vertices $\rho$ and $\rho^\dagger$ are always connected by an interface. However, because the spins are on the faces of the triangulation, this interface is in general not unique.
Similarly to \cite{CT20}, we will consider peeling processes that explore one such interface at a time. More precisely, when $\nu\ge\nu_c$, we will consider the peeling process that explores the \emph{left-most} interface $\iroot$ from $\rho$ to $\rho^\dagger$. (This is the same choice as in \cite{CT20}). When $1<\nu<\nu_c$, we will apply explorations along other interfaces, see Section~\ref{sec:hightemplimit} for details.
In all of the cases, the exploration reveals one triangle adjacent to the interface at each step, and swallows a finite number of other triangles if the revealed triangle separates the unexplored part into two pieces.

Formally, we define the \emph{peeling process} as an increasing sequence of \emph{explored maps} $\nseq \emap$. The precise definition of $\emap_n$ will be left to Section~\ref{sec:peeling}. The peeling process is also encoded by a sequence of \emph{peeling events} $\nseq[1] \Step$ taking values in a countable set of symbols, where $\Step_n$ indicates the position of the triangle revealed at time $n$ relative to the explored map $\emap_{n-1}$. Again, the detailed definition is left to Section~\ref{sec:peeling}.
The law of the sequence $\nseq[1] \Step$ can be written down fairly easily and one can perform explicit computations with it. We denote by $\Prob_{p,q}^\nu$ the law of the sequence $\nseq[1] \Step$ under $\prob_{p,q}^\nu$.

Let $(P_n,Q_n)$ be the boundary condition of the unexplored map at time $n$ and $(X_n,Y_n)$ its variation, that is, $X_n=P_n-P_0$ and $Y_n=Q_n-Q_0$. This definition makes sense when the initial condition $(P_0,Q_0)=(p,q)$ is finite. When $(p,q)$ is not finite, we need to define $(X_n,Y_n)$ differently: we will show that $(X_n,Y_n)$ is a deterministic function of the peeling events $(\Step_k)_{1\le k\le n}$, whose law has a well-defined limit when $p,q\to\infty$. This allows us to define the law of the process $\nseq{X_n,Y}$ under $\Prob^\nu\yy:=\lim_{p,q\to\infty}\Prob^\nu_{p,q}$.
We will see that $(X_n,Y_n)_{n\ge 0}$ is a random walk on $\integer^2$ under $\Prob^\nu\yy$. It was proven in \cite{CT20} for the corresponding expectations of the increments that
\begin{equation}\label{eq:mu}
\EE^\nu\yy(X_1)=\EE^\nu\yy(Y_1)=\mu:=\frac{1}{4\sqrt 7}>0\qquad \text{when}\quad \nu=\nu_c,
\end{equation}
which implies that almost surely, the interface hits the boundary of the half-plane a finite number of times, and then escapes towards infinity. When viewed as a function of the temperature $\nu$, the drift of the random walk $(X_n,Y_n)_{n\ge 0}$ actually defines an order parameter:

\begin{proposition}[Order parameter]\label{prop:orderparam intro}
Let $\mathcal{O}(\nu):=\EE_\infty^\nu((X_1+Y_1)\id_{|X_1|\vee|Y_1|<\infty})$. Then
\begin{equation*}
\mathcal{O}(\nu)=\begin{cases}
0,\qquad\text{if}\quad 1<\nu<\nu_c \\
f(\nu)\qquad\text{if}\quad \nu\geq\nu_c,
\end{cases}
\end{equation*}
where $f:[\nu_c,\infty)\to\R$ is a continuous, strictly increasing function such that $f(\nu_c)=2\mu>0$ and $\lim_{\nu\nearrow\infty}f(\nu)<\infty$ exists. Moreover, for $1<\nu<\nu_c$, we have the drift condition $\EE_\infty^\nu(X_1)=-\EE_\infty^\nu(Y_1)>0$.
\end{proposition}

Notice that there is an asymmetry between the two components of the drift of the random walk $(X_n,Y_n)_{n\ge 0}$ under $\EE_\infty^\nu$. This is a consequence of the following asymmetry in the definition of the perimeter process: In Section~\ref{sec:peeling}, we define a peeling process that explores the \emph{left-most} interface $\mathcal I$ from the vertex $\rho$. The perimeter process $(P_n,Q_n)_{n\ge 0}$ and its variation $(X_n,Y_n)_{n\ge 0}$ are defined relative to this peeling process. Therefore it is not surprising that the two components of $(X_n,Y_n)_{n\ge 0}$ have different drifts under $\EE_\infty^\nu$.

The function $\mathcal{O}$ defines an order parameter for two reasons:
First, its behavior fits formally the definition of an order parameter in physics, namely: the value of $\mathcal O(\nu)$ is zero on one side of the critical temperature, and positive on the other side. (A classical example of such an order parameter is the magnetization of the Ising model on regular lattices.)
More importantly, the positivity of $\mathcal O(\nu)$ really distinguishes the ordered phase $\nu\ge \nu_c$ from the disordered phase $\nu<\nu_c$ via the behavior of the interface $\mathcal I$ in the local limit. We will explain this in the next paragraph.

\paragraph{Interface geometry.}
Recall that for a finite bicolored triangulation $\bt$ with Dobrushin boundary condition, $\mathcal I$ is defined as the left-most interface from $\rho$ to $\rho^\dagger$ imposed by the boundary condition. In the limit $p,q\to \infty$, the interface $\mathcal I$ becomes a (possibly infinite) path on the infinite triangulation of distribution $\prob^\nu_\infty$. Many geometric properties of $\mathcal I$ --- especially its visits to the boundary of the triangulation --- are encoded by the random walk $(X_n,Y_n)_{n\ge 0}$ of law $\EE_\infty^\nu$. The next proposition summarizes some almost sure properties of the interface $\mathcal I$ which follow from Proposition~\ref{prop:orderparam intro}. The geometric pictures behind these properties are discussed after the proposition.


\begin{proposition}[Geometry of the interface $\mathcal{I}$]\label{prop:interface geom}
In the local limit $\prob_\infty^\nu$,
the left-most interface $\mathcal{I}$ has the following properties almost surely
\begin{itemize}
\item When $\nu\in (1,\nu_c)$\,: $\mathcal{I}$ is infinite and touches the boundary of the triangulation infinitely many times. 
\item When $\nu=\nu_c$\,: $\mathcal{I}$ is infinite, but touches the boundary of the triangulation only finitely many times.
\item When $\nu\in (\nu_c,\infty)$\,: $\mathcal{I}$ is finite.
\end{itemize}
\end{proposition}

When $\nu\in (1,\nu_c)$, due to the fact that $\EE_\infty^\nu(X_1)=-\EE_\infty^\nu(Y_1)>0$, the peeling process starting from the \< edge on the left of $\rho$ drifts to the left. This exploration also follows the left-most interface starting from $\rho$, which stays near the infinite \< boundary segment hitting it almost surely infinitely many times. Similarly, the right-most interface starting from $\rho$ and explored via a peeling exploration starting from the edge on the right of $\rho$ drifts to the right following the \+ boundary. Since $\EE^\nu_\infty(X_1)+\EE^\nu_\infty(Y_1)=0$, these two interfaces have the same geometry up to reflection. Using this property, we will construct a peeling algorithm under which the peeling process explores the half-plane in layers, with a starting point alternating between \< and \+ edges. The new peeling exploration obtained in this way reveals that the local limit constructed via this peeling process has a percolation-like interface geometry. On the contrary, if $\nu\in [\nu_c,\infty)$, the peeling process explores an interface which drifts towards the infinity $\rho^\dagger$ after hitting the boundary only finitely many times. The fact that this drift is increasing in $\nu$ means that the lower the temperature is, the less the interface hits the boundary and the faster the interface tends to the infinity. In fact, it is also shown that if $\nu\in (\nu_c,\infty)$, the peeling process approaches a neighborhood of $\rho^\dagger$ in a finite time almost surely.

One should compare the statement of Proposition~\ref{prop:interface geom} to the geometry of the percolation interface on the UIHPT (see \cite{Ang02,Ang05,AC13}). In that case, the interface hits the boundary infinitely many times almost surely.
As Proposition~\ref{prop:interface geom} suggests, in the high temperature
phase ($1<\nu<\nu_c$), the Ising model in the local limit looks like a subcritical face percolation, whereas in the low temperature phase ($\nu>\nu_c$), the local limit contains almost surely a bottleneck separating the \+ and \< regions. In the latter case, the local limit is not almost surely one-ended, contrary to the usual case of local limits of random planar maps. This property reflects that our model in the low temperature phase is really a quantum gravity version of the Ising model on 2D regular lattices in the ferromagnetic low temperature phase: the energy minimizing property forces the bottleneck due to the coupling of matter with gravity. Both the high and the low temperature cases are predicted in physics literature, though not extensively studied (see \cite{Kaz86,ADJ97}). More about the geometric interpretations is found in Section~\ref{sec:orderparam}.

Now we consider again the law of a finite Boltzmann Ising triangulation $\prob_{p,q}^{\nu}$ and study how the interface length scales together with the perimeter of the disk as $p,q\to\infty$ simultaneously. Let $T_m:=\inf\{n\ge 0: \min\{P_n,Q_n\}\leq m\},$ which can be seen as the first jump time of the interface to a neighborhood of the infinity. By its definition, $T_m$ is also the first hitting time of the stochastic process $\left(\min\{P_n,Q_n\}\right)_{n\ge 0}$ to $[0,m]$, which is a stopping time with respect to the filtration generated by $\nseq{P_n,Q}$ or $\nseq{X_n,Y}$. In the most interesting regime $\nu=\nu_c$, we find an explicit scaling limit of $T_m$ under diagonal rescaling of $p,q$:

\begin{theorem}[Scaling limit of $T_m$]\label{thm:scaling}
Let $\nu=\nu_c$ and consider the limit where $p,q\to \infty$ and $q/p\to\lambda$ for some $\lambda\in (0,\infty)$. For all $m\in\natural$ and all $t\ge 0$, the jump time $T_m$ has the following scaling limit:
\begin{equation}\label{eq:Tm scaling}
\Prob^{\nu_c}_{p,q}\m({T_m/p>t} \ \cv[]{p,q}\
\frac1{C(\lambda)} \int_{\mu t}^\infty
(1+s)^{-7/3}(\lambda+s)^{-7/3}ds
\end{equation}
where $C(\lambda) = \int_0^\infty (1+s)^{-7/3} (\lambda+s)^{-7/3} \dd s$.
In particular, when $\lambda=1$, we have
\begin{equation*}
\Prob^{\nu_c}_{p,q}\m({T_m/p>t} \ \cv[]{p,q}\
(1+\mu t)^{-11/3}.
\end{equation*}
\end{theorem}

An analogous result without the diagonal rescaling (via an intermediate local limit) was obtained in \cite[Proposition~11]{CT20}. As explained in \cite[Section 6]{CT20}, $T_m$ is, in some sense, an approximation of the interface length of a finite Boltzmann Ising-triangulation, though some technical difficulties remain to show that its scaling limit gives the scaling limit of the interface length. Hence, we state a conjecture:

\begin{conjecture}[Scaling limit of the interface length]
Let $\eta$ be the length of the left-most interface in $\bt$. Then
\begin{equation*}
\prob_{p,q}(\eta/p>t) \ \cv[]{p,q}\
\frac1{C(\lambda)} \int_{\mu t/E}^\infty(1+s)^{-7/3}(\lambda+s)^{-7/3}ds\qquad\text{while}\quad \frac{q}{p}\to\lambda,
\end{equation*}
where $E$ is the expected number of interface edges swallowed in a single peeling step.
\end{conjecture}

The idea behind the above conjecture is explained in \cite[Section 6]{CT20} in a similar setting. The main obstacle of the proof for the conjecture is that we lack information of $E$ with our current approach. One could find an asymptotic estimate for the volume of a finite Boltzmann Ising-triangulation, which gives an upper bound for the length of a piece of interface swallowed by a peeling step, but it turns out not to be sufficient. However, an analog of the conjecture could be proven for the model with spins on vertices, or with spins on faces and a general boundary. The former is conducted in the preprint \cite{T20}. The conjecture is also supported by a prediction derived from the \emph{Liouville Quantum Gravity}, seen as a continuum model of \emph{quantum surfaces} studied eg. in \cite{matingoftrees}, which also inspired us to find the correct constant in the scaling limit of Theorem~\ref{thm:scaling}. More discussion about this is given in Section~\ref{sec:onejumpscaling}.

To understand the phase transition at the critical point in greater detail, one should also consider the so-called near-critical regime. In our context, this means
that we let $\nu\to \nu_c$ simultaneously with the perimeters tending to infinity. Intuitively, one expects that if $\nu\to\nu_c$ fast enough compared to the growth of the perimeters, observables of the model will have the same limit as when $\nu=\nu_c$. On the contrary, if the convergence $\nu\to\nu_c$ is slow, the observables should have limits similar to those obtained at off-critical temperatures. An interesting question is to determine whether there is a critical window between the critical and the off-critical regimes, where the limits exhibits a qualitatively different behavior. These problems are considered in a work in progress.

\paragraph{Outline.}
The paper is composed of two parts, which can be read independently of each other.

The first part, which spans Sections~\ref{sec:RP of GF}--\ref{sec:coeff asymp}, deals with the enumeration of Ising-decorated triangulations. We start by deriving explicit \RP s of the generating function $Z(u,v,t,\nu)$ and its specialization $Z(u,v,\nu)\equiv Z(u,v,t_c(\nu),\nu)$ on the critical line (Section~\ref{sec:RP of GF}). Using these \RP s, we show that for each $\nu>1$, the bivariate generating function $Z(u,v,\nu)$ has a unique dominant singularity and an analytic continuation on the product of two $\Delta$-domains (Section~\ref{sec:singularity structure}). We then compute the asymptotic expansion of $Z(u,v,\nu)$ at its unique dominant singularity (Section~\ref{sec:local expansion}). Finally, we prove the coefficient asymptotics in Theorem~\ref{thm:asympt} using a generalization of the classical transfer theorem based on double Cauchy integrals (Section~\ref{sec:coeff asymp}).

The second part, which comprises Sections~\ref{sec:peeling perim}--\ref{sec:locallimit c diag} and Appendix~\ref{sec:bigjumplemma proof}, tackles the probabilistic analysis of the Ising-triangulations at any fixed temperature $\nu\in(1,\infty)$. It uses the combinatorial results of the first part as an input, and leads to the proofs of Theorems~\ref{thm:cv} and \ref{thm:scaling}. First, we introduce the different versions of the peeling process adapted to the three phases (high/low/critical temperature) and the two limit regimes examined in Theorem~\ref{thm:cv}. Then, we study the associated perimeter processes, whose drifts in the limit $p,q\to\infty$ define the order parameter introduced in Proposition~\ref{prop:orderparam intro} (Section~\ref{sec:peeling perim}). After that, we provide a general framework for constructing local limits, which we then use to prove the local convergence of Theorem~\ref{thm:cv} when $\nu\neq\nu_c$ (Section~\ref{sec:locallimits}). Finally, we prove Theorem~\ref{thm:scaling} and complete the proof of Theorem~\ref{thm:cv} by extending the above convergence result to the regime where $\nu=\nu_c$ and $p,q\to\infty$ simultaneously (Section~\ref{sec:locallimit c diag}). A central tool in the proofs in this last section is an adaptation of the one-jump lemma for the perimeter process in the diagonal regime, whose proof we present separately in Appendix~\ref{sec:bigjumplemma proof} as an adaptation of \cite[Appendix B]{CT20}.

\section{Rational parametrizations of the generating functions.}\label{sec:RP of GF}

The functional equations satisfied by the generating functions $Z_0(u,t,\nu)$ and $Z(u,v,t,\nu)$ were derived in our previous work \cite{CT20}. The result were written in the form of
\begin{equation}\label{eq:E0 and E}
\mc E_0 \mb({ Z_0(u) ,u, t,\nu, z_{1,0}, z_{3,0} } = 0
\qtq{and}
Z(u,v,t,\nu) = \mc E \mb({ Z_0(u), Z_0(v), u,v, t,\nu, z_{1,0}, z_{3,0} }
\end{equation}
where $\mc E_0$ and $\mc E$ are explicit rational functions with coefficients in $\rational$. Let us briefly summarize their derivation:

\begin{enumerate}
\item
We start by expressing the fact that the probabilities of all peeling steps sum to one. This gives two equations (called loop equations or Tutte's equations) with two catalytic variables for $Z(u,v,t,\nu)$. These equations are linear in $Z(u,v,t,\nu)$.
\item
By extracting the coefficients of $[v^0]$ and of $[v^1]$ from these two equations, we obtain four algebraic equations relating the variable $u$ to the series $Z_p(u,t,\nu)$ for $p=0,1,2,3$, whose coefficients are polynomials in $t$, $\nu$ and $z_{1,0}(t,\nu)$, $z_{3,0}(t,\nu)$. These equations are linear in the three variable $Z_1$, $Z_2$ and $Z_3$. After eliminating these variables, we obtain the first equation of \eqref{eq:E0 and E}.
This procedure is essentially equivalent to the method used in \cite[Chapter~8]{EynardBook} to solve Ising model on more general maps.

\item
Using the four algebraic equations found in Step 2, one can also express $Z_1(u,t,\nu)$ as a rational function of $Z_0(u,t,\nu)$, $u$, $t$, $\nu$ and $z_{1,0}(t,\nu)$, $z_{3,0}(t,\nu)$. Then, plug this relation into one of the two loop equations, and we obtain the second equation of \eqref{eq:E0 and E}.
\end{enumerate}

In this section, we first solve the equation for $Z_0(u,t,\nu)$ with the help of known \RP s of $z_{1,0}(t,\nu)$ and $z_{3,0}(t,\nu)$. Then, the solution is propagated to $Z(u,v,t,\nu)$ using its rational expression in $Z_0(u,t,\nu)$ and its coefficients. Finally, we specialize the parametrization of $Z(u,v,t,\nu)$ to the critical line $t=t_c(\nu)$ by replacing two parameters ($S,\nu$) with a single parameter $R$.

\subsection{\RRP\ of $Z_0(u,t,\nu)$}

\begin{lemma}\label{lem:RP Z0}
$Z_0(u,t,\nu)$ has the following \RP:
\begin{align}
t^2 &= \hat T(S,\nu) \hspace{5.5mm} :=
    \frac{ (S-\nu) (S+\nu-2) (4S^3-S^2-2S+\nu^2-2\nu) }{ 32 (1-\nu^2)^3 S^2 }
    \label{eq:RP:T General}\\
tu &= \hat U(H,S,\nu) \, :=  \ H\cdot
    \frac{2(4S^3-S^2-2S +\nu^2-2\nu) -4(S+1)S^2 H +4S^2 H^2 -S H^3}{16(1-\nu^2)^2 S}
    \label{eq:RP:U General} \\
Z_0(u,t,\nu) &= \hat Z_0(H,S,\nu) := \ \frac{\hat U(H,S,\nu)}{\hat T(S,\nu)} \cdot
    \frac{ (S-\nu)(S+\nu-2) + 2(S-\nu)S H - 2S^2 H^2 + S H^3 }{ 4(1-\nu^2)S H } \ .
    \label{eq:RP:Z0 General}
\end{align}
\end{lemma}

\begin{proof}
The following \RP s of $z_{1,0}(t,\nu)$ and $z_{3,0}(t,\nu)$ were obtained in \cite{CT20} by translating a related result from \cite{BBM11}: $t^2 = \hat T(S,\nu)$ and
\begin{align}
t^3 \cdot z_{1,0}(t,\nu) =&\ \hat z_{1,0}(S,\nu) :=
    \frac{ (\nu-S)^2 (S+\nu-2) }{ 64 (\nu^2-1)^4 S^2 }
           (3S^3 -\nu S^2 -\nu S +\nu^2 -2\nu) \,,
    \label{eq:RP:z1 General}\\
t^9 \cdot z_{3,0}(t,\nu) =&\ \hat z_{3,0}(S,\nu) :=
    \frac{ (\nu-S)^5 (S+\nu-2)^5}{2^{22} (\nu^2-1)^{12} S^8}
		\cdot \big(\, 160S^{10} -128S^9 -16(2\nu^2-4\nu+3)S^8
	\label{eq:RP:z3 General}
\\& +\ 32(2\nu^2-4\nu+3)S^7 - 7(16\nu^2-32\nu+27)S^6 - 2(32\nu^2-64\nu+57)S^5 \notag
\\& +\ (32\nu^4-128\nu^3+183\nu^2-110\nu+20)S^4 - 4(7\nu^2-14\nu-2)S^3 \notag
\\& +\ \nu(\nu-2)(9\nu^2-18\nu-20)S^2 + 14\nu^2(\nu-2)^2 S - 3 \nu^3 (\nu-2)^3 \notag
		\,\big) \,.
\end{align}
Substituting $u$ by $U/t$, and then $t$, $z_{1,0}$, $z_{3,0}$ by their respective parametrizations in the first equation of \eqref{eq:E0 and E}, we obtain an algebraic equation of the form $\hat{\mc E}_0(Z_0,U,S,\nu) = 0$.
It is straightforward to check that \eqref{eq:RP:U General}--\eqref{eq:RP:Z0 General} cancel the equation, that is, $\hat{\mc E}_0(\hat Z_0(H,S,\nu), \hat U(H,S,\nu), S,\nu)=0$ for all $H$, $S$ and $\nu$. See \cite{CAS2} for the explicit computation.

On the other hand, we know that \eqref{eq:E0 and E} uniquely determines the formal power series $Z_0(u)$, see \cite[Section~3.1]{CT20}. When $H\to 0$, Equations~\eqref{eq:RP:U General}--\eqref{eq:RP:Z0 General} clearly parametrize an analytic function $Z_0(u)$ near $u=0$. Therefore they are indeed a \RP\ of $Z_0(u,t,\nu)$.
\end{proof}

\begin{remark}
The proof of Lemma~\ref{lem:RP Z0} followed a guess-and-check approach. To actually derive the parametrization \eqref{eq:RP:U General}--\eqref{eq:RP:Z0 General}, we first check that the plane curve defined by $\hat{\mc E}_0(Z_0,U,S,\nu)=0$ has zero genus using the command \code{algcurves[genus]} of Maple, so it does have a \RP\ with coefficients in $\rational(S,\nu)$.
Theoretically, one should be able to produce one such parametrization using the Maple command \code{algcurves[parametrization]}. However, the execution takes too much time, presumably due to the presence of two indeterminates $(S,\nu)$ in the coefficient ring. Instead, we followed the steps below to find \eqref{eq:RP:U General}--\eqref{eq:RP:Z0 General}:

\begin{enumerate}[\arabic*.]
\item
Choose a finite set of values $\mc N \subset \rational \cap (1,\infty)$ for $\nu$. In practice we used the integers $\mc N=\{2,3,\ldots,10\}$.

\item
For each $\nu_* \in \mc N$, apply \code{algcurves[parametrization]} to the algebraic curve $\hat{\mc E}_0(Z_0,U,S,\nu_*)=0$. Let $\bar U_{\nu_*}(H,S)$ and $\bar Z_{0,\nu_*}(H,S)$ denote the rational functions over the ring $\rational(S)$ returned by the command.

If $\bar U_{\nu_*}(H,S) = \hat U(H,S,\nu_*)$ and $\bar Z_{0,\nu_*}(H,S) = \hat Z_0(H,S,\nu_*)$ for all $\nu_*\in \mc N$, where $\hat U$ and $\hat Z_0$ are two trivariate rational functions, then we can apply interpolation techniques to recover the expressions of $\hat U$ and $\hat Z_0$ for general values of $\nu$.
However, since the \RP\ of a (genus zero) algebraic equation is not unique, the functions $\mn({\, \bar U_{\nu_*},\,\bar Z_{0,\nu_*}\, }_{\nu_*\in \mc N}$ are in general \emph{not} the specializations of the same functions $(\hat U,\hat Z_0)$ at different values of $\nu_*$. In order to recover the specializations $\hat U(H,S,\nu_*)$ and $\hat Z_0(H,S,\nu_*)$ from them, we need to ``preprocess'' the pairs
$\mn({ \bar U_{\nu_*},\,\bar Z_{0,\nu_*} }$ as in the two following steps.

\item
Maple guarantees that $(\bar U_{\nu_*}, \bar Z_{0,\nu_*})$ is a \emph{proper} \RP\ of the curve $\hat{\mc E}_0(Z_0,U,S,\nu_*)=0$. We know that all proper \RP{}s of the same curve are related to each other by Möbius transformations \cite[Lemma~4.17]{SWP08}.
Therefore, there exists a family of Möbius transformations $\mc m_{S,\nu_*}$ indexed by the formal variable $S$ and the numerical values $\nu_*\in \mc N$, such that
\begin{equation}\label{eq:RP preprocessing}
\bar U_{\nu_*} \m({ \mc m_{S,\nu_*}(H),S } = \hat U(H,S,\nu_*)
\qtq{and}
\bar Z_{0,\nu_*} \m({ \mc m_{S,\nu_*}(H),S } = \hat Z_0(H,S,\nu_*)
\end{equation}
for some trivariate rational functions $\hat U$ and $\hat Z_0$. To find such a family of Möbius transformations, we make the following observations (see \cite{CAS2} for explicit verification):
\begin{enumerate}[(\roman*)]
\item
For all $\nu_* \in \mc N$, there exists a rational function $\bar H_{\nu_*}(S) \in \rational(S)$ such that $H=\bar H_{\nu_*}(S)$ is the unique pole of both $H\mapsto \bar U_{\nu_*}(H,S)$ and $H\mapsto \bar Z_{0,\nu_*}(H,S)$.

\item
The algebraic curve $\hat{\mc E}_0(Z_0,U,S,\nu)=0$ has a unique analytic branch at the point $(U,Z_0)=(0,1)$.
\\
And for all $\nu_* \in \mc N$, we have $\bar U(H,S,\nu_*) \to 0$ and $\bar Z_0(H,S,\nu_*) \to 1$ as $H\to \infty$.
\end{enumerate}
These two observations suggest that we choose Möbius transformations which map $\infty$ to $\bar H_{\nu_*}(S)$, and map $0$ to $\infty$. (See below for the consequences of this choice.) Such Möbius transformations are of the form
\begin{equation*}
m_{S,\nu_*}(H) = \bar H_{\nu_*}(S) - \Lambda_{\nu_*}(S)/H
\end{equation*}
where $\Lambda_{\nu_*}(S) \ne 0$ is an arbitrary scaling factor to be chosen later.

\item
Plugging the above Möbius transformation into \eqref{eq:RP preprocessing} gives our candidates for $\hat U(H,S,\nu_*)$ and $\hat Z_0(H,S,\nu_*)$. Our choice of $m_{S,\nu_*}$ ensures that these two functions are polynomial in $H$ (i.e.\ their only pole is at $\infty$) and that $\mn({ \hat U(0,S,\nu_*), \hat Z_0(0,S,\nu_*) }=(0,1)$. We compute in \cite{CAS2} the explicit expressions of $\hat U(H,S,\nu_*)$ and $\hat Z_0(H,S,\nu_*)$ and check that they are polynomials of degrees 4 and 6 respectively in the variable $H/\Lambda \equiv H/\Lambda_{\nu_*}(S)$.

Now we ask Maple to display $\hat U(H,S,\nu_*)$ as a polynomial in $H/\Lambda$, and look for common factors among its coefficients (which are elements of $\rational(S)$). With some trial-and-error, we find that the choice
\begin{equation*}
\Lambda_{\nu_*}(S) = -8S \cdot \frac{ [(H/\Lambda)^4] \hat U(H,S,\nu_*) }{ [(H/\Lambda)^3] \hat U(H,S,\nu_*) }
\end{equation*}
cancels all those common factors.
This choice is also equivalent to the condition that $\frac{ [H^3] \hat U(H,S,\nu_*) }{ [H^4] \hat U(H,S,\nu_*) } = -8S$. The prefactor $8$ is not chosen for simplification reasons. Rather, it is chosen so that $(\hat U,\hat Z_0)$, the \RP\ that we get after interpolation in $\nu$, will specialize to the \RP\ given in our previous article \cite{CT20} when $(\nu,t)=(\nu_c,t_c)$.

\item
The above choice of $\Lambda_{\nu_*}(S)$ gives us the expressions of $\hat U(H,S,\nu_*)$ and $\hat Z_0(H,S,\nu_*)$ for all $\nu_*\in \mc N$. Then, we apply the Maple routine \code{CurveFitting[RationalInterpolation]} to find a pair $(\hat U,\hat Z_0) \in \rational(H,S,\nu)^2$ that interpolates between these values of $\nu_*$. This gives the expressions \eqref{eq:RP:U General}--\eqref{eq:RP:Z0 General}.

One can run the above procedure with a larger set $\mc N$, and check that the result does not change.
\end{enumerate}
\end{remark}

\subsection{\RRP\ of $Z(u,v,t,\nu)$.}

We plug the parametrizations \eqref{eq:RP:T General}--\eqref{eq:RP:z3 General} into the second equation of \eqref{eq:E0 and E} to obtain a \RP\ of $Z(u,v,t,\nu)$ of the form
\begin{equation*}
t^2=\hat T(S,\nu) \qquad
tu = \hat U(H,S,\nu) \qquad
tv = \hat U(K,S,\nu) \qtq{and}
Z(u,v,t,\nu) = \hat Z(H,K,S,\nu) \,,
\end{equation*}
where the rational functions $\hat T$ and $\hat U$ are defined in Lemma~\ref{lem:RP Z0}, and the expression of $\hat Z$ is given in \cite{CAS2}.

\subsection{Specialization of $Z(u,v,t,\nu)$ to the critical line $t=t_c(\nu)$.}\label{sec:specialization of RP}

\paragraph{\RRP\ of the critical line.}
Recall that $t_c(\nu)$ is defined as the radius of convergence of the series $z_{1,0}(\adot,\nu)$. The series have nonnegative coefficients, and have a real rational parametrization of the form $t^2=\hat T(S,\nu)$ and $t^3\cdot z_{1,0} = \hat z_{1,0}(S,\nu)$ given by \eqref{eq:RP:T General} and \eqref{eq:RP:z1 General}. As explained in \cite[Appendix~B]{CT20}, the value $S=S_c(\nu)$ that parametrizes the point $t=t_c(\nu)$ is either a zero of $\partial_S \hat T(\adot,\nu)$ or a pole of $\hat z_{1,0}(\adot,\nu)$. More precise calculation (see \cite{CAS2}) using the method of \cite[Appendix~B]{CT20} shows that $S_c(\nu)$ is the largest zero of $\partial_S \hat T(\adot,\nu)$ below $S=\nu$ (which parametrizes $t=0$). The equation $\partial_S \hat T(S,\nu)=0$ factorizes, and $S_c(\nu)$ satisfies
\begin{align}
2 S^3 -3 S^2 - \nu^2 +2 \nu = 0 & \qt{if }\nu \in (1,\nu_c] \,,
\label{eq:def S_c high}    \\
3 S^2 -\nu^2 +2 \nu = 0         & \qt{if }\nu \in [\nu_c,\infty) \,,
\label{eq:def S_c low}
\end{align}
where $\nu_c = 1+2\sqrt7$.
It is not hard to check that $S_c(\nu)$ has the following piecewise \RP:
\begin{equation}\label{eq:RP:S_c}
\nu = \check \nu(R) = \begin{cases}
\frac{1}{2}(2-3R+R^3) & \\
\frac{27}{13+2R-2R^2} &
\end{cases}
\quad\tq{and}
S_c(\nu) = \check S(R) = \begin{cases}
\frac{1}{2}(R^2-1)         & \qt{for }R \in (R_1,R_c]    \\
\frac{3(2R-1)}{13+2R-2R^2} & \qt{for }R \in [R_c,R_\infty)
\end{cases}
\end{equation}
where $R_1=\sqrt3$, $R_c=\sqrt7$, $R_\infty=\frac{1+3\sqrt3}2$ correspond respectively to the coupling constants $\nu=1$, $\nu=\nu_c$ and $\nu=\infty$.
Plugging \eqref{eq:RP:S_c} into $\hat T(S,\nu)$ gives the following piecewise \RP\ of $t_c(\nu)$:
\begin{equation*}
t_c(\nu)^2 = \check T(R) := \begin{cases}
\displaystyle
\frac{3R^2-1}{2 R^3 (4-3R+R^3)^3}    & \text{for }R \in (R_1,R_c]
\vspace{1.5ex} \\
\displaystyle
\frac{ (1+R)^2 (13+2R-2R^2)^3 (19-10R-2R^2)
    }{ 128 (R-5) (4+R)^3 (7-R+R^2)^3
    }    & \text{for }R \in [R_c,R_\infty)
\end{cases}
\end{equation*}

\paragraph{\RRP\ of $Z(u,v,t,\nu)$ on the critical line.}
Define $\check U(H,R) = \hat U(H,\hat S(R),\hat \nu(R))$ and $\check Z(H,K,R) = \hat Z(H,K,\hat S(R),\hat \nu(R))$.
Then $Z(u,v,\nu) \equiv Z(u,v,t_c(\nu),\nu)$ has the piecewise \RP:
\begin{equation*}
t_c(\nu) \cdot u = \check U(H,R) \qquad  t_c(\nu) \cdot v  = \check U(K,R)
\qtq{and} Z(u,v,\nu) = \check Z(H,K,R) \ ,
\end{equation*}
where
\begin{equation}\label{eq:RP:U}
\check U(H,R) := \begin{cases}
\displaystyle
\frac{\mn({3-10 R^2+3 R^4}+\mn({1-R^4}H-2 \mn({1-R^2}H^2 -H^3
    }{ R^2 \m({3-R^2}^2 \m({4-3 R+R^3}^2
    } H
\hspace{23mm} \qt{for }R \in (R_1,R_c]
\vspace{1.5ex} \\
\displaystyle
-\frac{ (13+2R-2R^2)^2H
     }{ 256 (5-R)^2 (4+R)^2 (7-R+R^2)^2}
\bigg( 8(1+R)(5-R) \m({ 19-10R-2R^2-3(1-2R)H }
\vspace{0.5ex} \\
\displaystyle ~\hspace{15.5mm}
     +12(1-2R)\m({ 13+2R-2R^2 } H^2 + \m({ 13+2R-2R^2 }^2 H^3 \bigg)
 \qt{for }R \in [R_c,R_\infty)
\end{cases}
\end{equation}
whereas $\check Z(H,K,R)$, too long to be written down here, is given in \cite{CAS2}. Since we look for the asymptotics of $z_{p,q}(\nu)$ when $p,q\to\infty$ with fixed values of $\nu$, we will be interested in the singularity behavior of $\check U(H,R)$ and $\check Z(H,K,R)$ at fixed values of $R$. For this reason we introduce the shorthand notations
\begin{equation*}
\check U_R(H) := \check U(H,R)
\qtq{and}
\check Z_R(H,K) := \check Z(H,K,R) \,.
\end{equation*}

\subsection{Domain of convergence of $Z(u,v,\nu)$ and its parametrization.}\label{sec:dom of cvg}

\paragraph{Definition and parametrization of $u_c(\nu)$.}
For all $R\in (R_1,R_\infty)$, let $\check H_c(R)$ be the smallest positive zero of the derivative $\check U_R'$. Using the expression \eqref{eq:RP:U}, it is not hard to find that
\begin{equation}\label{eq:RP:H_c}
\check H_c(R) := \begin{cases}
\displaystyle
\frac{R^2-3}2
& \text{for }R \in (R_1,R_c]
\vspace{1.5ex} \\
\displaystyle
\frac{5+4 R-R^2- \sqrt{3(5-R) (1+R) \left(R^2-7\right)}}{13+2 R-2 R^2}
& \text{for }R \in [R_c,R_\infty)
\end{cases}
\end{equation}
For $\nu>1$, let $u_c(\nu)$ be the function parametrized by $\nu=\check \nu(R)$ and $t_c(\nu)\cdot u_c(\nu) = \check U_R(\check H_c(R))$, where $R\in (R_1,R_\infty)$.

\begin{lemma}\label{lem:dom of cvg}
For all $\nu>1$, the double power series $(u,v)\mapsto Z(u,v,\nu)$ is absolutely convergent \Iff\ $|u|\le u_c(\nu)$ and $|v|\le u_c(\nu)$.
\end{lemma}

\begin{proof}
First, we notice that the proof can be reduced to the problem of estimating the radii of convergence of two univariate power series: it suffices to show that the series $u\mapsto Z(u,0,\nu) \equiv Z(0,u,\nu)$ is divergent when $|u|>u_c(\nu)$, and the series $u\mapsto Z(u,u,\nu)$ is convergent at $u=u_c(\nu)$.
Indeed, since the double power series $Z(u,v,\nu)$ has nonnegative coefficients, the divergence condition implies that $Z(u,v,\nu)$ is divergent when $|u|>u_c(\nu)$ or $|v|>u_c(\nu)$, and the convergence condition implies that $Z(u,v,\nu)$ is absolutely convergent for all $|u|\le u_c(\nu)$ and $|v|\le u_c(\nu)$.

The univariate series $u\mapsto Z(u,0,\nu)$ has nonnegative coefficients and the following \RP:
\begin{equation*}
t_c(\nu) \cdot u = \check U_R(H) \qtq{and}
Z(u,0,\nu) = \check Z_R(H,0)   \,.
\end{equation*}
It is not hard to check that this \RP s are real and proper (see \cite[Appendix~B]{CT20} for the definitions and characterizations of these properties), and the parametrization $t_c(\nu)\cdot u = \check U_R(H)$ maps a small interval around $H=0$ increasingly to an interval around $u=0$. Hence the parametrization of the radius of convergence of $u\mapsto Z(u,0,\nu)$ can be determined in the framework of \cite[Proposition~21]{CT20}. More precisely, the radius of convergence $u_c^*(\nu)$ should satisfy $t_c(\nu) u_c^*(\nu) = \check U_R(\check H_c^*(R))$, where $\check H_c^*(R)$ is the smallest positive number that is either a zero of $\check U_R'$, or a pole of $H\mapsto \check Z_R(H,0)$. Comparing this to the definition of $\check H_c(R)$, we see that $\check H_c^*(R) \le \check H_c(R)$, and hence $u_c^*(\nu)\le u_c(\nu)$.\footnote{
Using its explicit expression, one can check that $H\mapsto \check Z_R(H,0)$ has no pole on $[0,\check H_c(R)]$. Hence $\check H_c^*(R) = \check H_c(R)$ and $u_c^*(\nu) = u_c(\nu)$. But this is not necessary for the proof.}
This shows that $u\mapsto Z(u,0,\nu)$ is divergent when $|u|>u_c(\nu)$.

We apply the same argument to the series $u\mapsto Z(u,u,\nu)$, which has the \RP
\begin{equation*}
t_c(\nu) \cdot u = \check U_R(H) \qtq{and}
Z(u,u,\nu) = \check Z_R(H,H)   \,.
\end{equation*}
Again, the \RP\ is real and proper. Using its explicit expression, one can check that the rational function $H\mapsto \check Z_R(H,H)$ has no pole on $[0,\check H_c(R)]$. With the same argument as for $u\mapsto Z(u,0,\nu)$, we conclude that $u_c(\nu)$ is the radius of convergence of $u\mapsto Z(u,u,\nu)$ and the series is convergent at $u=u_c(\nu)$ (because $Z(u_c(\nu),u_c(\nu),\nu) = \check Z_R(\check H_c(R),\check H_c(R))$ is finite).
This concludes the proof of the lemma.
The necessary explicit computations in the above proof can be found in \cite{CAS2}.
\end{proof}

\paragraph{Notations:} In the following, we will use the renormalized variables $(x,y)= \m({ \frac u{u_c(\nu)}, \frac v{u_c(\nu)} }$.
A parametrization of the function $(x,y)\mapsto Z(u_cx,u_cy,\nu)$ is given by $x=\check x(H,R)$, $y=\check x(K,R)$ and $\tilde Z(x,y,\nu) = \check Z(H,K,R)$, where $\check x(H,R) \equiv \check x_R(H) := \check U_R(H) /\, \check U_R(\check H_c(R))$ is still a rational function in $H$. In the low temperature regime $\check x(H,R)$ is no longer rational in $R$ due to the square root in \eqref{eq:RP:H_c}. However it remains continuous on $(R_1,R_\infty)$ and smooth away from $R_c$. These regularity properties will be more than sufficient for our purposes.

\paragraph{Definition of holomorphicity and conformal bijections:}
We say that a function is \emph{holomorphic} in a (not necessarily open) domain if it is holomorphic in the interior of the domain and continuous in the whole domain.
This definition is also valid for functions of several complex variables, in which case \emph{holomorphic} means that the function has a multivariate Taylor expansion that is locally convergent. A \emph{conformal bijection} is a bijection which is holomorphic and whose inverse is also holomorphic.

\paragraph{Definition of $\Hdom[0](R)$:}
By \cite[Proposition~21]{CT20}, for each $R\in (R_1,R_\infty)$, the mapping $\check x_R$ induces a conformal bijection
from a compact neighborhood of $H=0$ to the closed unit disk $\cdisk$. We denote by $\cHdom[0](R)$ this neighborhood and by $\Hdom[0](R)$ its interior.
It is not hard to see that $\Hdom[0](R)$ is the connected component of the preimage $\check x_R^{-1}(\disk)$ which contains the origin. This characterization of $\Hdom[0](R)$ will be used in the proof of Lemma~\ref{lem:unique dominant}. Notice that it implies in particular that $\Hdom[0](R)$ is symmetric \wrt\ the real axis.

\section{Dominant singularity structure of $Z(u,v,\nu)$}\label{sec:singularity structure}

In this section, we prove that the bivariate generating function $(x,y)\mapsto Z(u_c(\nu)x,u_c(\nu)y,\nu)$ has a unique dominant singularity at $(x,y)=(1,1)$, and is ``$\Delta$-analytic'' in a sense similar to the one defined in \cite{FS09} for univariate generating functions. Before starting, let us briefly describe the state of the art for the singularity analysis of algebraic generating functions of one or two variables.

For a generating function $F(z)=\sum_{n\ge 0} F_n z^n$ of one complex variable, a dominant singularity of $F$ is by definition a singularity with minimal modulus. Moreover, this minimal modulus is equal to the radius of convergence $\rho$ of the Taylor series $\sum_{n} F_n z^n$, so the dominant singularities of $F$ are simply those on the circle $\Set{z\in \complex}{|z|=\rho}$.
When $F$ is algebraic, it behaves locally near a singularity $z_*$ like $(z-z_*)^{r}$ with some $r\in \rational$. In particular, one can find a disk centered at $z_*$ such that (a branch of) $F$ is analytic in the disk with one ray from $z_*$ to $\infty$ removed.
Since algebraic functions have only finitely many singularities, it follows that any univariate algebraic function $F(z)$ with finite radius of convergence has an analytic continuation in a domain of the form $\bigcap_{i} \m({z_i\cdot \slit}$, where $z_i$ are the dominant singularities of $F$, and $\slit$ is the disk of radius $1+\epsilon>1$ centered at $0$, with the segment $[1,1+\epsilon]$ removed. This ensures that the classical transfer theorem (see \cite[Chapter VI.3]{FS09}) always applies to algebraic functions, and gives coefficient asymptotics of the form $F_n \sim \sum_i c_i\cdot z_i^{-n} \cdot n^{-r_i}$ with $c_i\in \complex$ and $r_i\in \rational$. In particular, when the dominant singularity is unique, the asymptotics has the simple form of $F_n \sim c\cdot z_*^{-n} \cdot n^r$.

When $F(x,y) = \sum_{m,n} F_{m,n} x^m y^n$ is an algebraic function of two complex variables, the situation is much more complicated. First, the singularities of $F(x,y)$ are in general no longer isolated points. Also, the definition of dominant singularities has to be generalized: instead of minimizing $|z|$ in the univariate case, one needs to minimize the product $|x|^\lambda |y|$, where $\lambda=\lim \frac mn$ is defined by the regime of $m,n\to \infty$ in which one looks for the asymptotic of $F_{m,n}$.
The general picture for the singularity analysis of bivariate algebraic functions is still far from being fully understood. The only systematic study we found in the literature concerns the case where $F(x,y)$ is rational or meromorphic. See \cite{ACSVhome} for references. (A non-rational case has also been studied in \cite{Greenwood18}. But it concerns functions of a special form, and does not cover the case we are interested in here.)
When $F(x,y)$ is rational (or of the form studied in \cite{Greenwood18}), the locus of singularities of $F$ is an algebraic sub-variety of $\complex^2$. In that case, sophisticated tools from algebraic geometry can be used to locate the dominant singularities, and to study $F(x,y)$ locally near the dominant singularities.

For the Ising-triangulations, the singularity locus of the generating function $(x,y)\mapsto Z(u_c x,u_c y,\nu)$ is much harder to describe, since it involves describing branch cuts of the function in $\complex^2$. Luckily, the structure of dominant singularities is very simple: regardless of the relative speed at which $p,q\to \infty$, the dominant singularity is always unique and at $(x,y)=(1,1)$. Moreover, the function has an analytic continuation ``beyond the dominant singularity'' in both the $x$ and $y$ coordinates, in the product of two $\Delta$-domains. Proposition~\ref{prop:singularity structure} gives the precise formulation of the above claim.

\paragraph{Notations.}
We denote by $\disk$ the open unit disk in $\complex$ and by $\mathrm{arg}(z) \in (-\pi,\pi]$ the argument of $z \in \complex$. For $\epsilon>0$ and $0\le \theta<\pi/2$, define the $\Delta$-domain
\begin{equation*}
\Ddom = \set{\, z\in (1+\epsilon)\cdot \disk\, }{\, z\ne 1 \text{ and }|\mathrm{arg}(z-1)|>\theta \, }.
\end{equation*}
When $\theta=0$, the above definition gives $\Ddom[0] = (1+\epsilon)\cdot \disk \setminus [1,1+\epsilon)$, which is a disk with a small cut along the real axis. We call this a slit disk, and use the abbreviated notation $\slit \equiv \Ddom[0]$.

We denote by $\partial \Ddom$ and $\cDdom$ be the boundary and the closure of $\Ddom$. When $\theta \in (0,\pi/2)$, these are taken \wrt\ the usual topology of $\complex$. When $\theta=0$ however, we view $\slit$ as a domain in the universal covering space of $\complex\setminus \{1\}$, and define $\partial \slit$ and $\cslit$ \wrt\ that topology. In this way the closed curve $\partial \slit$ will be a nice limit of $\partial \Ddom$ when $\theta \to 0^+$, as illustrated in Figure~\refp{a}{fig:Hdom}.

\begin{figure}
\centering
\includegraphics[scale=0.75]{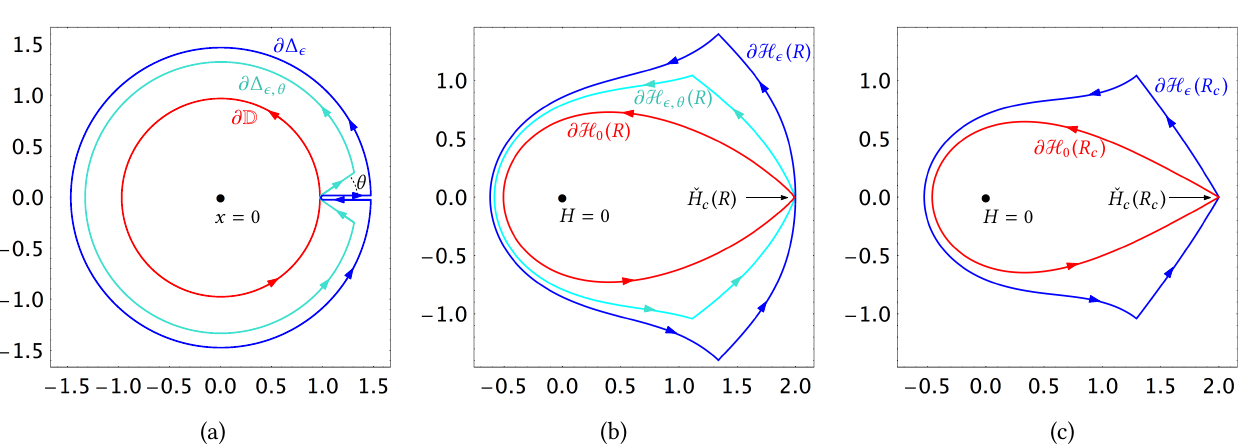}
\caption{(a) Boundaries of the unit disk $\disk$, the $\Delta$-domain $\Ddom$ and the slit disk $\slit$. For the sake of visibility, $\Ddom$ and $\slit$ are drawn for two different values of $\epsilon$.\\
(b) Boundaries of the domains $\Hdom[0](R)$, $\Hdom[\epsilon,\theta](R)$ and $\Hdom(R)$ defined by the parametrization $\check x_R$ at a non-critical temperature $R\ne R_c$. By definition, $\Hdom[0](R)$ (resp.\ $\Hdom[\epsilon,\theta](R)$ and $\Hdom(R)$) is the connected component of the preimage $\check x_R^{-1}(\disk)$ (resp.\ $\check x_R^{-1}(\Ddom)$ and $\check x_R^{-1}(\slit)$) containing the origin.\\
(c) Boundaries of the domains $\Hdom[0](R)$ and $\Hdom(R)$ defined by $\check x_R$ at the critical temperature $R=R_c$. Notice that at the point $\check H_c(R)$, the curve $\partial \Hdom(R)$ in (b) has a tangent, while the curve $\partial \Hdom(R_c)$ in (c) has two half-tangents at an angle $2\pi/3$.
}
\label{fig:Hdom}
\end{figure}

\begin{proposition}\label{prop:singularity structure}
For all $\nu>1$ and $\theta \in (0,\frac\pi2)$, there exists $\epsilon>0$ such that (an analytic continuation of) the function $(x,y)\mapsto Z(u_c(\nu) x, u_c(\nu) y, \nu)$ is holomorphic in $\cslit \times \cDdom$. Moreover, when $\nu\ge \nu_c$, we can take $\theta=0$, i.e., find $\epsilon>0$ such that the function is holomorphic in $\cslit \times \cslit$.
\end{proposition}

\begin{remark}
As mentioned at the end of the previous section,  by ``holomorphic in $\cslit \times \cslit$'', we mean that a function has complex partial derivatives in the interior $\slit \times \slit$ of $\cslit \times \cslit$, and is continuous in $\cslit \times \cslit$.
This will be later used to express the coefficients $z_{p,q}(\nu)$ as double Cauchy integrals on the contour $\partial \slit \times \partial \slit$, so that their asymptotics when $p,q\to\infty$ can be estimated easily.
For this purpose, it is not absolutely necessary to prove the continuity of $(x,y)\mapsto Z(u_c(\nu) x, u_c(\nu) y,\nu)$ on the boundary of $\cslit \times \cslit$ (in particular, at the point $(1,1)$). But not knowing this continuity would require one to approximate the contour $\partial \slit \times \partial \slit$ by a sequence of contours that lie inside $\slit \times \slit$, which complicates a bit the estimation of the double Cauchy integral.
\end{remark}

The rest of this section is devoted to the proof of Proposition~\ref{prop:singularity structure}. To this end, we will construct the desired analytic continuation of $(x,y)\mapsto Z(u_c(\nu) x, u_c(\nu) y,\nu)$ based on the heuristic formula $Z(u_c x, u_c y) = \check Z(\check x^{-1}(x), \check x^{-1}(y))$. The proof comes in two steps: First, we show that for each fixed $R$, the rational function $\check x_R$ defines a conformal bijection from a set $\cHdom(R)$ to $\cslit$ for some $\epsilon>0$. Then, we try to show that for all $\epsilon$ small enough, the rational function $\check Z_R(H,K)$ has no pole, hence is holomorphic, in $\cHdom(R) \times \cHdom(R)$.
It turns out that this is true only when $\nu \ge \nu_c$. When $\nu\in (1,\nu_c)$, one needs to reduce the domain $\cHdom(R) \times \cHdom(R)$ a bit, which corresponds to replacing one factor in the product $\cslit \times \cslit$ by a $\Delta$-domain $\cDdom$ with some opening angle $\theta>0$.

\subsection{The conformal bijection $\check x_R: \cHdom(R) \to \cslit$}

\begin{lemma}[Uniqueness and multiplicity of the critical point of $\check x_R$]
\label{lem:unique dominant}
For all $R\in (R_1,R_\infty)$, $\check H_c(R)$ is the unique zero of the rational function $\check x_R'$ in $\cHdom[0](R)$. It is a simple zero if $R\ne R_c$, and a double zero if $R=R_c$.
\end{lemma}

\begin{proof}
By definition, $\check H_c(R)$ is a zero of $\check x_R'$. One can easily check that it is a simple zero if $R\in (R_1,R_\infty)\setminus \{R_c\}$, and a double zero if $R=R_c$. It remains to show its uniqueness in $\cHdom[0](R)$.

By the definition of $\Hdom[0](R)$, the restriction of $\check x_R$ to this set is a conformal bijection. Therefore the derivative $\check x_R'$ has no zero in $\Hdom[0](R)$. On the other hand, $\check x_R'$ is a polynomial of degree three for all $R\in (R_1,R_\infty)$, so it has three zeros (counted with multiplicity), one of which is $\check H_c(R)$. In the following we show that the two other zeros are not in the set $\partial \Hdom[0](R) \setminus \{ \check H_c(R) \}$, and this will complete the proof.

When $R\in [R_c,R_\infty)$, we check by explicit computation (see \cite{CAS2}) that all three zeros of $\check x_R'$ are on the positive real line. Since $\Hdom[0](R)$ is a topological disk containing $H=0$ and is symmetric with respect to the real axis, its boundary intersects the positive real line only once (at $\check H_c(R)$). Hence $\check x_R'$ has no zero on $\partial \Hdom[0](R) \setminus \{ \check H_c(R) \}$.

When $R\in (R_1,R_c)$, the zeros of $\check x'_R$ are not always real. In this case we resort to a proof by contradiction:
Let $\chi(H,R) = \frac{\partial_H \check x(H,R)}{H-\check H_c(R)}$. Assume that for some $R_*\in (R_1,R_c)$, the quadratic polynomial $H\mapsto \chi(H,R_*)$ has a zero $H_*$ in $\partial \Hdom[0](R_*) \setminus \{ \check H_c(R_*) \}$.
We will show that the pair $(H_*,R_*)$ satisfies the following system of algebraic equations
\begin{equation}\label{eq:no real sol}
\chi(H,R) = 0
\ ,\qquad
\check x(H,R) = \frac{s+i}{s-i}
\qtq{and}
\partial_R \check x - \frac{\partial_R \raisebox{3pt}{$\chi$}}{ \partial_H \raisebox{3pt}{$\chi$}} \cdot \partial_H \check x = ir\cdot \check x
\end{equation}
where $r,s \in \real$ are two auxiliary variables. Notice that this system contains 3 complex equations, but only 5 real variables ($\Re H, \Im H, R, s$ and $r$). So we expect it to have no solution. We can check that this is indeed the case: First, we eliminate $H$ to obtain two complex polynomial equations relating $R$, $r$ and $s$. Since these variables are all real, the real part and the imaginary part of each equation must both vanish. We check that the resulting system of four polynomial equations has no real solution using a general algorithm \cite{KRS16realroot} implemented in Maple as \code{RootFinding[HasRealRoot]}, see \cite{CAS2}. By contradiction, this proves that $\check H_c(R)$ is the unique zero of $\check x_R'$ in $\partial \Hdom[0](R)$ for all $R\in (R_1,R_c)$, and completes the proof of the lemma modulo a justification of the system \eqref{eq:no real sol}.

The first equation of \eqref{eq:no real sol} is true by the definition of $(H_*,R_*)$.
The second equation expresses the fact that $\check x(H_*,R_*) \in \partial \disk \setminus \{1\}$, which is the image of our assumption $H_* \in \partial \Hdom[0](R_*) \setminus \{\check H_c(R_*)\}$ under the mapping $\check x_{R_*}$. Indeed, since $s\mapsto \frac{s+i}{s-i}$ is a bijection from $\real$ to $\partial \disk \setminus \{1\}$, we have $\check x(H,R) \in \partial \disk \setminus \{1\}$ \Iff\ $\check x(H,R) = \frac{s+i}{s-i}$ for some $s\in \real$.
The last equation of \eqref{eq:no real sol} is a consequence of the following two facts:
\begin{enumerate}[(\roman*)]
\item \label{enum:1}
$\partial_H \chi(H_*,R_*)\ne 0$. Hence the equation $\chi(H,R)=0$ defines a smooth implicit function $H=\check H_*(R)$ in a neighborhood of $(H,R)=(H_*,R_*)$.
\item \label{enum:2}
The derivative $\od{}R \log \mn|{\check x(\check H_*(R),R)}$ vanishes at $R_*$.
\end{enumerate}
Indeed, by the implicit function theorem, \ref{enum:1} implies that $\od{}R \check  H_*(R_*) = - \frac{\partial_{R} \chi(H_*,R_*)}{\partial_{H} \chi(H_*,R_*)}$. On the other hand, we have
\begin{equation*}
\od{}R \log \abs{ \check x(\check H_*(R),R) } \,=\, \od{}R \Re\m({ \log \check x(\check H_*(R),R) } \,=\, \Re \m({ \frac{ \od{}R \check x(\check H_*(R),R) }{ \check x(\check H_*(R),R) } } \,.
\end{equation*}
Expanding $\od{}R \check x(\check H_*(R),R)$ using the chain rule, we see that the expression on the \rhs\ vanishes at $R=R_*$ \Iff\ the last equation of \eqref{eq:no real sol} holds for some $r\in \real$ and when $(H,R)=(H_*, R_*)$

One can verify \ref{enum:1} by an explicit computation: If $\partial_H \chi(H_*,R_*)=0$, then we can solve the pair of equations $\chi(H_*,R_*) = 0$ and $\partial_H \chi(H_*,R_*)=0$
(the first equation is quadratic in $H$, while the second one is linear), which has a unique solution that satisfies $R_*\in (R_1,R_c)$. But this solution gives a numerical value $|\check x(H_*,R_*)| \ne 1$ (see \cite{CAS2}), which contradicts the fact that $\check x(H_*,R_*) \in \partial \disk$. Thus we have $\partial_H \chi(H_*,R_*)\ne 0$.

The justification of \ref{enum:2} is a bit more technical. It is a consequence of the following observation: By definition, $\check H_*(R)$ is a critical point of $\check x_R$ for all $R$, thus it can never enter the open set $\Hdom[0](R)$. However, the point $H_*\equiv \check H_*(R_*)$ is on the boundary of $\Hdom[0](R_*)$. Intuitively, this implies that the movement of the point $\check H_*(R)$ must be in some sense \emph{stationary} \wrt\ the domain $\Hdom[0](R)$ when $R=R_*$.
To prove~\ref{enum:2}, we will show that this stationarity constraint translates to the stationarity of the function $R\mapsto |\check x(\check H_*(R),R)|$ at $R=R_*$.
For this, we will change our reference frame to the point $\check H_*(R)$. In other words, we will make a change of variable $H=\check H(h,R)$ such that $\check H(0,R)=\check H_*(R)$, and study the evolution of the domain $\Hdom[0](R)$ in the variable $h$ when $R$ varies around $R_*$.



To construct a change of variable that simplifies the expression of $\Hdom[0](R)$, let us consider the function $f(z,R)=\frac{\check x(\check H_*(R)+z,R)}{\check x(\check H_*(R),R)}-1$. Since $f(0,R)\equiv 0$ and $\partial_z f(0,R) = \frac{\partial_H \check x(\check H_*(R),R)}{\check x(\check H_*(R),R)} \equiv 0$, the function $\alpha(z,R)=z^{-2} f(z,R)$ is analytic in a neighborhood of $(0,R_*)$. Moreover, according to \ref{enum:1} we have $\partial_z^2 f(0,R_*) = \frac{\partial_H^2 \check x(\check H_*(R_*),R_*)}{\check x(\check H_*(R_*),R_*)}\ne 0$, hence $\alpha(0,R_*)\ne 0$. By the inverse function theorem, the mapping $(z,R) \mapsto  (\sqrt{\alpha(z,R)} z,R)$ has a local inverse $(h,R)\mapsto (\check z(h,R),R)$ that is jointly analytic in $(h,R)$ in a neighborhood of $(0,R_*)$. Let $\check H(h,R) = \check H_*(R) + \check z(h,R)$. One can check that the inverse function relation $\sqrt{\alpha(\check z(h,R),R)} \cdot \check z(h,R) = h$ implies
\begin{equation*}
\check x(\check H(h,R),R) = \check x(\check H_*(R),R) \cdot (1+h^2)
\end{equation*}
for all $(h,R)$ in a neighborhood of $(0,R_*)$.
In the variable $h$, the preimage of the unit disk $\disk$ by $\check x_R$ is simply the set $\Setn{h}{|1+h^2|<|\check x(\check H_*(R),R)|^{-1}}$. More precisely, we have
\begin{equation*}
\check x^{-1}_R(\disk) = \Set{\check H(h,R)}{|1+h^2| < |\check x(\check H_*(R),R)|^{-1} } \,.
\end{equation*}
in a neighborhood of $H=\check H_*(R)$.

When $R=R_*$, we have $|\check x(\check H_*(R_*),R_*)|^{-1} = 1$. In this case, $\Setn{h}{|1+h^2|<1} \equiv \Set{h}{|\mathfrak{Re}(h)|<|\mathfrak{Im}(h)|}$ is a two-sided cone, as in Figure~\refp{b}{fig:local-merger}. Recall that $\Hdom[0](R)$ is the connected component of $\check x^{-1}_R(\disk)$ containing $H=0$. Since the point $H_*=\check H_*(R_*)$ is on the boundary of the domain $\Hdom[0](R)$, at least one side of the two-sided cone must belong to $\Hdom[0](R)$.
Now assume that $\od{}R \mn|{\check x(\check H_*(R_*),R_*)} \ne 0$, then we have $|\check x(\check H_*(R),R)|<1$ either for $R>R_*$ or for $R<R_*$ in a neighborhood of $R_*$. But, as shown in Figure~\refp{c}{fig:local-merger}, in this case the preimage $\check x^{-1}_R(\disk)$ has only one connected component locally near $\check H_*(R)$. This connected component must belong to $\Hdom[0](R)$ because of the continuity of $R\mapsto \check x(H,R)$. It follows that $H=\check H_*(R)$ (which is $h=0$ in the variable $h$) belongs to $\Hdom[0](R)$. This contradicts the fact that the domain $\Hdom[0](R)$ contains no critical point of $\check x_R$.
Thus we must have $\od{}R \mn|{\check x(\check H_*(R),R)} = 0$, or equivalently $\od{}R \log \mn|{\check x(\check H_*(R),R)} = 0$, when $R=R_*$. This justifies the claim \ref{enum:2} and completes the proof of the lemma.
\end{proof}

\begin{figure}
\centering
\includegraphics[scale=0.9]{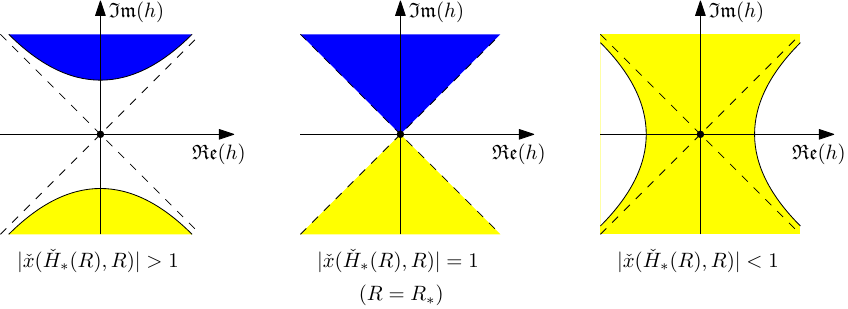}\\
\pbox{0.55}{(a)\hfill(b)\hfill(c)}~~~~~
\caption{Local behavior of the set $\check x_R^{-1}(\disk)$ in the coordinate $h$. The region corresponding to $\Hdom[0](R)$ is colored in yellow, and the region corresponding to $\check x^{-1}_R(\disk) \setminus \Hdom[0](R)$ is colored in blue (the upper region in the graphs (a) and (b)).}
\label{fig:local-merger}
\end{figure}

\begin{remark}
The second equation in \eqref{eq:no real sol} implies $\check x_R(H) \in \partial \disk$, but does not guarantee that $H\in \partial \Hdom[0](R)$, because the mapping $\check x_R$ is not injective on $\complex$.  In fact, if one removes the last equation from \eqref{eq:no real sol}, then the system does have a solution $(H,R,s)$ with $R \in (R_1,R_c)$ and $s \in \real$. This solution corresponds to a critical point of $\check x_R$ which is not on $\partial \Hdom[0](R)$, but is nevertheless mapped to $\partial \disk \setminus \{1\}$ by $\check x_R$.

The purpose of the last equation of \eqref{eq:no real sol} is precisely to avoid this kind of undesired solutions.
Without the last equation, the algebraic system \eqref{eq:no real sol} contains two complex equations with four real unknowns ($\Re(H)$, $\Im(H)$, $R$, $s$). So generically, we do expect it to have a finite number of solutions.
The last equation adds one \emph{complex} equation to the system while introducing only an extra \emph{real} variable. With it, we expect generically that \eqref{eq:no real sol} has no solution.
In general, if the mapping $\check x(H)$ depends on $m$ real parameters $(R_1,\ldots,R_m)$ instead of $R$, then provided that $\check x(H)$ has continuous derivatives \wrt\ each of the parameters, one can replace the last equation of \eqref{eq:no real sol} by $m$ complex equations with $m$ extra real variables. Then we would have a system of $m+2$ complex equations with $3+m+m=2m+3$ real variables, which generically would have no solution.

Our justification of last equation of \eqref{eq:no real sol} came in two steps. The first step \ref{enum:1} asserts that the critical point $H_*$ has multiplicity one. It is checked by an explicit computation and depends on the specific function $\check x_R$. On the contraty, the second step \ref{enum:2} derive the desired equation in \eqref{eq:no real sol} using a variational argument which is mostly independent of specific features of $\check x_R$. Currently, the argument in \ref{enum:2} still depends on the fact that $H_*$ has multiplicity one. In the upcoming paper \cite[Appendix A]{Chen2021}, the first author gives a generalization of this variational argument which applies to critical points of any multiplicity. That general argument would allow us to bypass the verification of \ref{enum:1} in the above proof.
\end{remark}

\paragraph{Definition of $\Hdom(R)$:}
For each $R\in (R_1,R_\infty)$, the above lemma and Proposition~21(iii) of \cite{CT20} imply that there exists $\epsilon>0$ for which $\check x_R$ defines a conformal bijection from a compact set $\cHdom(R) \supset \cHdom[0](R)$ to $\cslit$. For $\theta \in (0,\pi/2)$, let $\cHdom[\epsilon,\theta](R)$ be the preimage of the $\Delta$-domain $\cDdom \subset \cslit$ under this bijection.
We denote by $\partial \Hdom(R)$ and $\Hdom(R)$ the boundary and the interior of $\cHdom(R)$, and similarly for $\cHdom[\epsilon,\theta](R)$.

Notice that the notation $\cHdom(R)$ fits well with the previously defined $\cHdom[0](R)$, since the latter is in bijection with the closed unit disk $\cdisk$, which can be viewed as a special case of the domain $\cslit$ with $\epsilon=0$.

\paragraph{Geometric interpretation of Lemma~\ref{lem:unique dominant}.}
We know that analytic functions preserve angles at non-critical points.
More generally, if $f$ is an analytic function such that $H\in \complex$ is a critical point of multiplicity $n$ (that is, a zero of multiplicity $n$ of $f'$, with $n\ge 0$), then $f$ maps each angle $\theta$ incident to $H$ to an angle $(n+1)\theta$.
Since $\Hdom[0](R)$ is mapped bijectively by $\check x_R$ to the unit disk (whose boundary is smooth everywhere), the boundary of $\Hdom[0](R)$ forms an angle of $\pi/(n+1)$ at each $H\in \partial \Hdom[0](R)$ which is a critical point of multiplicity $n$ of $\check x_R$.
Therefore, Lemma~\ref{lem:unique dominant} tells us that the boundary of $\Hdom[0](R)$ is smooth everywhere except at $H=\check H_c(R)$, where it has two half-tangents forming an angle of $\pi/2$ if $R\ne R_c$, or an angle of $\pi/3$ if $R=R_c$.
This is illustrated by the red curves in Figure~\refp{b}{fig:Hdom}~and~\refp{c}{fig:Hdom}.

For the same reason, the boundary of $\Hdom[\epsilon,\theta](R)$ has also two half-tangents at $H=\check H_c(R)$. They form an angle of $\pi-\theta$ if $R\ne R_c$, and an angle of $\frac23(\pi-\theta)$ if $R=R_c$. (In particular, when $\theta=0$ and $R\ne R_c$, the angle is equal to $\pi$, i.e.\ the two half-tangents become a tangent.) This is illustrated by the blue and cyan curves in Figure~\refp{b}{fig:Hdom}~and~\refp{c}{fig:Hdom}.
From this we deduce the following corollary, which will be used to derive the local expansion of the bivariate function $\check Z(H,K,R)$ at $(H,K) = (\check H_c(R), \check H_c(R))$ at critical and high temperatures.

\begin{corollary}\label{cor:double angle bound}
For all $R\in (R_1,R_\infty)$ and $\theta \in (0,\frac\pi2)$, there exist a neighborhood $\mc N$ of $(\check H_c(R),\check H_c(R))$ and a constant $M_\theta<\infty$ such that
\begin{equation}\label{eq:double angle bound}
\max \m({ |\check H_c(R)-H| , |\check H_c(R)-K| } \le M_\theta \cdot
    \abs{ (\check H_c(R)-H) + (\check H_c(R)-K) }
\end{equation}
for all $(H,K) \in \mc N \cap \m({ \cHdom(R) \times \cHdom[\epsilon,\theta](R) }$.

When $R=R_c$, one can take $\theta=0$ so that \eqref{eq:double angle bound} holds for all
$(H,K) \in \mc N \cap \m({ \cHdom(R) \times \cHdom(R) }$.
\end{corollary}

\begin{proof}
For $R \in (R_1,R_\infty) \setminus \{R_c\}$, the boundary of $\cHdom[\epsilon,\theta](R)$ has two half-tangents at $H=\check H_c(R)$, both at an angle of $\frac{\pi-\theta}2$ with the negative real axis. When $\theta=0$, the two half-tangents becomes a tangent that is orthogonal to the real axis. For any $\theta\in (0,\frac\pi2)$, we can choose $\theta_1>\frac\pi2$ and $\theta_2>\frac{\pi-\theta}2$ such that $\theta_1+\theta_2<\pi$. Then there exists a neighborhood $N$ of $\check H_c(R)$ such that $\arg(\check H_c(R)-H) \in (-\theta_1,\theta_1)$ for all $H\in N \cap \cHdom(R)$, and $\arg(\check H_c(R)-K) \in (-\theta_2, \theta_2)$ for all $K\in N \cap \cHdom[\epsilon,\theta](R)$.
In polar coordinates, this means that $\check H_c(R)-H = r_1 e^{i \phi_1}$ and $\check H_c(R)-K = r_2 e^{i \phi_2}$ satisfy $|\phi_1|\le \theta_1$ and $|\phi_2|\le \theta_2$, so that $|\phi_1-\phi_2|\le \theta_1+\theta_2<\pi$. It follows that
\begin{align*}
\abs{ (\check H_c(R)-H) + (\check H_c(R)-K) }^2
&= \abs{r_1 e^{i\phi_1} + r_2 e^{i\phi_2}}^2
 = r_1^2 + r_2^2 + 2r_1r_2 \cos(\phi_1-\phi_2)
\\ & \ge
   r_1^2 + r_2^2 + 2r_1r_2 \cos(\theta_1+\theta_2)
 \\&
= \mb({ r_1-r_2 \cos(\theta_1+\theta_2) }^2 + \mb({ r_2 \sin(\theta_1+\theta_2) }^2 \,.
\end{align*}
This implies that $r_2 = |\check H_c(R)-K| \le \frac{1}{\sin(\theta_1+\theta_2)} \cdot \abs{ (\check H_c(R)-H) + (\check H_c(R)-K) }$, and by symmetry, the inequality \eqref{eq:double angle bound} with $M_\theta = \frac{1}{\sin(\theta_1+\theta_2)}$, for all $(H,K) \in (N\times N) \cap \m({ \cHdom(R) \times \cHdom[\epsilon,\theta](R) }$.

When $R=R_c$, the boundary of $\cHdom(R)$ has two half-tangents at $H=\check H_c(R)$ at an angle of $\frac\pi3$ with the negative real axis. In this case, we can take $\theta_1=\theta_2=\frac{5\pi}{12} >\frac\pi3$ so that $\theta_1+\theta_2<\pi$. Then, the same proof as in the $R\ne R_c$ case shows that there exists a neighborhood $N$ of $\check H_c(R)$ such that \eqref{eq:double angle bound} holds with $M_0 = \frac{1}{\sin(5\pi/6)}$ for all $(H,K)\in (N\times N) \cap \m({ \cHdom(R) \times \cHdom(R) }$.
\end{proof}

\subsection{Holomorphicity of $\check Z$ on $\cHdom(R) \times \cHdom(R)$.}

The previous subsection showed that for $\epsilon>0$ small enough, $\cslit \times \cslit$ is mapped analytically by the inverse function of $(H,K)\mapsto (\check x_R(H),\check x_R(K))$ to the domain $\cHdom(R) \times \cHdom(R)$. Ideally, we want to show that the other part of the rational parametrization $(H,K) \mapsto \check Z_R(H,K)$ does not have poles on $\cHdom(R) \times \cHdom(R)$. Then the formula $Z(u_c x,u_c y) = \check Z_R( (\check x_R)^{-1}(x), (\check x_R)^{-1}(y) )$ would imply that $(x,y)\mapsto Z(u_c x,u_c y)$ has an analytic continuation on $\cslit \times \cslit$.

By continuity, any neighborhood of the compact set $\cHdom[0](R) \times \cHdom[0](R)$ contains $\cHdom(R) \times \cHdom(R)$ for all $\epsilon$ small enough. On the other hand, the poles of $\check Z_R$ form a closed set. Hence to prove that the domain $\cHdom(R) \times \cHdom(R)$ does not contain any poles for $\epsilon$ small enough, it suffices to show that the compact set $\cHdom[0](R) \times \cHdom[0](R)$ does not contain any poles of $\check Z_R$.
It turns out that this is almost the case:

\begin{lemma}\label{lem:poles of Z in H0*H0}
For all $R\in (R_c,R_\infty)$, the rational function $\check Z_R$ has no pole in $\cHdom[0](R) \times \cHdom[0](R)$.\\
\phantom{\textbf{Lemma~$16$.} }
For all $R\in (R_1,R_c]$, $(\check H_c(R),\check H_c(R))$ is the only pole of $\check Z_R$ in $\cHdom[0](R) \times \cHdom[0](R)$.
\end{lemma}

\begin{proof}
By definition, a pole of $\check Z_R$ is a zero of the polynomial $D_R$ in the denominator of $\check Z_R(H,K) = \frac{N_R(H,K)}{D_R(H,K)}$, where $N_R$ and $D_R$ are coprime polynomials of $(H,K)$. With an appropriate choice of the constant term $D_R(0,0)$, we can take $N(H,K,R):= N_R(H,K)$ and $D(H,K,R) := D_R(H,K)$ to be polynomial in all three variables $(H,K,R)$. We check by explicit computation (see \cite{CAS2}) that $D(\check H_c(R),\check H_c(R),R) \ne 0$ for all $R\in (R_c,R_\infty)$, and $D(\check H_c(R),\check H_c(R),R) = 0$ for all $R\in (R_1,R_c]$.
Then it remains to show that $D$ does not vanish for any $(H,K) \in \cHdom[0](R) \times \cHdom[0](R) \setminus \{ (\check H_c(R),\check H_c(R)) \}$ and $R\in (R_1,R_\infty)$. For this we use the following lemma, whose proof will be given later:

\begin{lemma}\label{lem:pole system}
If the polynomial $D$ vanishes at a point $(H,K,R)$ such that $(H,K) \in \cHdom[0](R) \times \cHdom[0](R)$ and $R\in (R_1,R_\infty)$, then both $N$ and $\partial_H N \cdot \partial_K D - \partial_K N \cdot \partial_H D$ vanish at $(H,K,R)$.
\end{lemma}

This lemma tells us that the poles of $\check Z_R$ in the physical range of the parameters (that is, for $R\in (R_1,R_\infty)$ and $(H,K) \in \cHdom[0](R) \times \cHdom[0](R)$) satisfy the system of three polynomial equations
\begin{equation}\label{eq:pole system 3}
D \,=\, N \,=\,
\partial_H N \cdot \partial_K D - \partial_K N \cdot \partial_H D \,=\, 0
\end{equation}
instead of just $D=0$. However, it is not easy to verify whether \eqref{eq:pole system 3} has a solution $(H,K,R)$ satisfying $(H,K) \in \cHdom[0](R) \times \cHdom[0](R) \setminus \{ (\check H_c(R),\check H_c(R)) \}$, for two reasons: On the one hand, the solution set of \eqref{eq:pole system 3} contains at least one continuous component: $(H,K,R)=(\check H_c(R), \check H_c(R), R)$ is a solution of \eqref{eq:pole system 3} for all $R\in (R_1,R_c]$. On the other hand, it is not easy to distinguish between points in $\cHdom[0](R)$ from points in the preimage $\check x_R^{-1}(\cdisk)$ which are not in $\cHdom[0](R)$. To mitigate these issues, we construct an auxiliary equation that eliminates some solutions of the system which are known to be outside $\cHdom[0](R) \times \cHdom[0](R) \setminus \{ (\check H_c(R),\check H_c(R)) \}$.

Since $\check x_R$ is a conformal bijection from $\cHdom[0](R)$ to the unit disk, we know that $H=0$ is its unique (simple) zero in $\cHdom[0](R)$. Hence the polynomial $H\mapsto \check U_R(H)/H$ does not vanish on $\cHdom[0](R)$. (Recall that $\check x_R$ is defined as $\check U_R$ divided by a constant that only depends on $R$.) On the other hand, $\check H_c(R)$ is the unique zero of $\check x_R'$ in $\cHdom[0](R)$ by Lemma~\ref{lem:unique dominant}. Thus if $(H,K) \in \cHdom[0](R) \times \cHdom[0](R)$ is different from $(\check H_c(R),\check H_c(R))$, then either $\check U_R'(H)\ne 0$ or $\check U_R'(K)\ne 0$.
Let $\mc{NZ}(H,K,R)= \frac{\check U_R(H)}H \cdot \frac{\check U_R(K)}K \cdot \check U_R'(H)$. Then the above discussion shows that for $R\in (R_1,R_\infty)$ and $(H,K) \in \cHdom[0](R) \times \cHdom[0](R) \setminus \{ (\check H_c(R),\check H_c(R)) \}$, either $\mc{NZ}(H,K,R) \ne 0$, or $\mc{NZ}(K,H,R) \ne 0$.

It follows that if $(H,K)$ is a pole of $\check Z_R$ in $\cHdom[0](R) \times \cHdom[0](R) \setminus \{ (\check H_c(R),\check H_c(R)) \}$, then either $(H,K,R)$ or $(K,H,R)$ is a solution to the system of equations
\begin{equation}\label{eq:pole system 4}
D \,=\, N \,=\,
\partial_H N \cdot \partial_K D - \partial_K N \cdot \partial_H D \,=\, 0
\qtq{and}
X\cdot \mc{NZ} = 1
\end{equation}
where $X\in \complex$ is an auxiliary variable used to express the condition $\mc{NZ}\ne 0$ as an algebraic equation.
A Gr\"obner basis computation (see \cite{CAS2}) shows that this system has no solution with real value of $R$. By contradiction, $\check Z_R$ has no pole in $\cHdom[0](R) \times \cHdom[0](R) \setminus \{(\check H_c(R),\check H_c(R))\}$ for all $R\in (R_1,R_\infty)$. This completes the proof.
\end{proof}

\begin{proof}[Proof of Lemma~\ref{lem:pole system}]
In this proof we fix an $R\in (R_1,R_\infty)$ and drop it from the notations.
Since the double power series $(x,y)\mapsto Z(u_c x,u_c y)$ is absolutely convergent for all $x,y$ in the unit disk $\cdisk$, and $\check x$ is a homeomorphism from $\cHdom[0]$ to $\cdisk$, the rational function $\check Z(H,K) = Z(u_c\cdot \check x(H), u_c\cdot \check x(K))$ is continuous on the compact set $\cHdom[0] \times \cHdom[0]$.

Assume that $D$ vanishes at some $(H,K) \in \cHdom[0] \times \cHdom[0]$. The boundedness of $\check Z$ on $\cHdom[0] \times \cHdom[0]$ implies that $N$ also vanishes at $(H,K)$. If $\partial_H D(H,K) = \partial_K D(H,K) = 0$, then $\partial_H N \cdot \partial_K D - \partial_K N \cdot \partial_H D$ obviously vanishes at $(H,K)$. Otherwise, consider the limit of $\check Z(H+\varepsilon h,K+\varepsilon k)$ when $\varepsilon \to 0^+$, where $h,k \in \complex$. By L'H\^opital's rule, for all $(h,k)$ such that $h \cdot \partial_H D(H,K) + k \cdot \partial_K D(H,K) \ne 0$, we have
\begin{equation}\label{eq:alpha limit}
\lim_{\varepsilon \to 0^+} \check Z(H+ \varepsilon h,K+\varepsilon k) = \frac{
h \cdot \partial_H N(H,K) + k \cdot \partial_K N(H,K)}{
h \cdot \partial_H D(H,K) + k \cdot \partial_K D(H,K)} \,.
\end{equation}
By the continuity of $\check Z$ on $\cHdom[0] \times \cHdom[0]$, the above limit is independent of $(h,k)$ as long as the pair satisfies that $(H+\varepsilon h,K+\varepsilon k) \in \cHdom[0] \times \cHdom[0]$ for all $\epsilon>0$ small enough. From Figure~\ref{fig:Hdom} (or more rigorously the geometric interpretation of Lemma~\ref{lem:unique dominant}), we see that for all $H\in \cHdom[0]$, there exists $h_*\ne 0$ such that $H+\varepsilon h_* \in \cHdom[0]$ for all $\epsilon>0$ small enough. Similarly, there exists $k_*\ne 0$ such that $K+\varepsilon k_* \in \cHdom[0]$ for all $\epsilon>0$ small enough. By taking $(h,k)$ to be equal to $(h_*,0)$, $(0,k_*)$ and $(h_*,k_*)$ in \eqref{eq:alpha limit}, we obtain that
\begin{equation*}
\frac{h_* \partial_H N(H,K)}{h_* \partial_H D(H,K)} \,=\,
\frac{k_* \partial_K N(H,K)}{k_* \partial_K D(H,K)} \,=\,
\frac{h_* \partial_H N(H,K) + k_* \partial_K N(H,K)}{
      h_* \partial_H D(H,K) + k_* \partial_K D(H,K)} \,,
\end{equation*}
provided that the denominators of the three fractions are nonzero.
By assumption, $\partial_H D(H,K)$ and $\partial_K D(H,K)$ do not both vanish. It follows that at least two of three fractions have nonzero denominators. From the equality between these two fractions, we deduce that  $\partial_H N \cdot \partial_K D - \partial_K N \cdot \partial_H D=0$ at $(H,K)$.
\end{proof}

Now we use the continuity argument mentioned at the beginning of this subsection to deduce the holomorphicity of $\check Z_R$ on $\cHdom \times \cHdom$ (or $\cHdom \times \cHdom[\epsilon,\theta]$, see below) from Lemma~\ref{lem:poles of Z in H0*H0}. The low temperature case is easy, since $\check Z_R$ does not have any pole in $\cHdom[0] \times \cHdom[0]$ for all $R\in (R_c,R_\infty)$. When $R\in (R_1,R_c]$, one has to study the restriction of $\check Z_R$ on $\cHdom \times \cHdom$ more carefully near its the pole $(\check H_c(R), \check H_c(R))$. This is done with the help of Corollary~\ref{cor:double angle bound}.

\begin{lemma}\label{lem:holomorphicity Z}
For all $R\in [R_c,R_\infty)$, there exists $\epsilon>0$ such that $\check Z_R$ is holomorphic in $\cHdom(R) \times \cHdom(R)$.\\
For  $R\in (R_1,R_c)$ and $\theta \in (0,\pi/2)$, there exists $\epsilon>0$  such that $\check Z_R$ is holomorphic in $\cHdom(R) \times \cHdom[\epsilon,\theta](R)$.
\end{lemma}

\begin{proof}
As in the previous proof, we fix a value of $R \in (R_1,R_\infty)$ and drop it from the notation.

\paragraph{Low temperatures.}
When $R\in (R_c,R_\infty)$, Lemma~\ref{lem:poles of Z in H0*H0} tells us that $\check Z$ has no pole in $\cHdom[0] \times \cHdom[0]$. Since the set of poles of $\check Z$ is closed, and $\cHdom[0] \times \cHdom[0]$ is compact, there exists a neighborhood of $\cHdom[0] \times \cHdom[0]$ containing no pole of $\check Z$.
By continuity, this neighborhood contains $\cHdom \times \cHdom$ for $\epsilon>0$ small enough. It follows that there exists $\epsilon>0$ such that $\check Z$ is holomorphic in $\cHdom \times \cHdom$.

\paragraph{Critical temperature.}
When $R=R_c$, Lemma~\ref{lem:poles of Z in H0*H0} tells us that $(\check H_c,\check H_c)$ is the only pole of $\check Z$ in $\cHdom[0] \times \cHdom[0]$. 

First, let us show that $\check Z$, when restricted to $\cHdom \times \cHdom$, is continuous at $(\check H_c,\check H_c)$. Notice that this statement does not depend on $\epsilon$, since two domains $\cHdom \times \cHdom$ with different values of $\epsilon>0$ are identical when restricted to a small enough neighborhood of $(\check H_c,\check H_c)$.
We have seen in the proof of Lemma~\ref{lem:poles of Z in H0*H0} that the numerator $N$ and the denominator $D$ of $\check Z$ both vanish at $(\check H_c, \check H_c)$. Therefore their Taylor expansions give:
\begin{equation}\label{eq:Zl limit at (Hc,Hc)}
\check Z(\check H_c -h, \check H_c -k) = \frac{
  \partial_H N(\check H_c,\check H_c) \cdot (h+k) + O\m({ \max(|h|,|k|)^2 }  }{
  \partial_H D(\check H_c,\check H_c) \cdot (h+k) + O\m({ \max(|h|,|k|)^2 }  }
\qt{as }(h,k)\to (0,0).
\end{equation}
We check explicitly that $\partial_H D(\check H_c, \check H_c) \ne 0$, see \cite{CAS2}. On the other hand, thanks to Corollary~\ref{cor:double angle bound} (the critical case), we have $\max(|h|,|k|) = O\m({|h+k|}$ when $(h,k)\to (0,0)$ in such a way that $(\check H_c-h,\check H_c-k) \in \cHdom \times \cHdom$. Then it follows from \eqref{eq:Zl limit at (Hc,Hc)} that $\check Z(H,K) \to \partial_H N(\check H_c,\check H_c) / \partial_H D(\check H_c,\check H_c)$ when $(H,K)\to (\check H_c,\check H_c)$ in $\cHdom \times \cHdom$. That is, $\check Z$ restricted to $\cHdom \times \cHdom$ is continuous at $(\check H_c,\check H_c)$.

Next, let us show that for some fixed $\epsilon_0>0$, every point $(H,K)\in \cHdom[0] \times \cHdom[0]$ has a neighborhood $\mc V(H,K)$ such that $\check Z$ is holomorphic in $\mc V(H,K) \cap (\cHdom[\epsilon_0] \times \cHdom[\epsilon_0])$. (Recall that this means $\check Z$ is holomorphic in the interior, and continuous in the whole domain).
For $(H,K)=(\check H_c,\check H_c)$, the expansion of the denominator in \eqref{eq:Zl limit at (Hc,Hc)} shows that there exists $\epsilon_0>0$ and a neighborhood $\mc V(\check H_c,\check H_c)$ such that $(\check H_c,\check H_c)$ is the only pole of $\check Z$ in $\mc V(\check H_c,\check H_c) \cap (\cHdom[\epsilon_0] \times \cHdom[\epsilon_0])$. Moreover, the previous paragraph has showed that $\check Z$ is continuous at $(\check H_c,\check H_c)$ when restricted to $\cHdom \times \cHdom$. It follows that $\check Z$ is holomorphic in $\mc V(\check H_c,\check H_c) \cap (\cHdom[\epsilon_0] \times \cHdom[\epsilon_0])$.
For $(H,K) \in \cHdom[0] \times \cHdom[0] \setminus \{(\check H_c,\check H_c)\}$, since $(H,K)$ does not belong to the (closed) set of poles of $\check Z$, it has a neighborhood $\mc V(H,K)$ on which $\check Z$ is holomorphic.

By taking the union of all the neighborhoods $\mc V(H,K)$ constructed in the previous paragraph, we see that there is a neighborhood $\mc V$ of the compact set $\cHdom[0] \times \cHdom[0]$ such that $\check Z$ is holomorphic in $\mc V \cap (\cHdom[0] \times \cHdom[0])$. By continuity, $\mc V$ contains $\cHdom \times \cHdom$ for some $\epsilon>0$ small enough. Hence there exists $\epsilon>0$ such that $\check Z$ is holomorphic in $\cHdom \times \cHdom$.

\paragraph{High temperatures.}
When $R\in (R_1,R_c)$, Lemma~\ref{lem:poles of Z in H0*H0} tells us that $(\check H_c,\check H_c)$ is also a pole of $\check Z$. The rest of the proof goes exactly as in the critical case, except that the domain $\cHdom \times \cHdom$ has to be replaced by $\cHdom \times \cHdom[\epsilon,\theta]$ for an arbitrary $\theta \in (0,\pi/2)$ due to the difference between the critical and non-critical cases in Corollary~\ref{cor:double angle bound}.
\end{proof}

\begin{remark}\label{rem:extended holomorphicity Z}
In fact, the above proof shows the holomorphicity of $\check Z_R$ in a larger domain than the one stated in Lemma~\ref{lem:holomorphicity Z}. In particular, one can check that the following statement is true: \emph{for each compact subset $\mathcal K$ of $\Hdom$, there exists a neighborhood $\mathcal V$ of $\cHdom[0]$ such that $\check Z$ is holomorphic in $\mathcal{K\times V}$.}
This remark will be used to show that $x\mapsto A(u_c x)$ is analytic on $\slit$ in Corollary~\ref{cor:A Delta-analytic}.
\end{remark}

\begin{proof}[Proof of Proposition~\ref{prop:singularity structure}]
The proposition follows from Lemma~\ref{lem:holomorphicity Z} and the definition of $\cHdom(R)$:

At critical or low temperatures, the inverse mapping of $(H,K)\mapsto (\check x_R(H), \check x_R(K))$ is holomorphic from $\cslit \times \cslit$ to $\cHdom(R) \times \cHdom(R)$. For $\epsilon>0$ small enough, $(H,K)\mapsto \check Z_R(H,K)$ is holomorphic in $\cHdom(R) \times \cHdom(R)$. Hence their composition defines an analytic continuation of $(x,y)\mapsto \tilde Z(x,y,\nu)$ on $\cslit \times \cslit$.

At high temperatures, it suffices to replace $\cslit \times \cslit$ by $\cslit \times \cDdom$, and $\cHdom(R) \times \cHdom(R)$ by $\cHdom(R) \times \cHdom[\epsilon,\theta](R)$.
\end{proof}

\section{Asymptotic expansions of $Z(u,v,\nu)$ at its dominant singularity}\label{sec:local expansion}
In this section, we establish the asymptotic expansions (Proposition~\ref{prop:local expansion}) of the generating function $Z(u_c x,u_c y)$ at its dominant singularity $(x,y)=(1,1)$. For this we define the function $\mc Z(\mc h,\mc k,\nu)$ by the change of variable
\begin{equation*}
Z(u_c(\nu) x, u_c(\nu) y, \nu) = \mc Z(\mc h,\mc k, \nu)    \qtq{with}
\mc h = \m({1-x}^\delta \qtq{and}
\mc k = \m({1-y}^\delta
\end{equation*}
Recall that $\delta = 1/2$ when $\nu\ne \nu_c$ (non-critical case), and $\delta = 1/3$ when $\nu=\nu_c$ (critical case).

The proof relies on Lemma~\ref{lem:unique dominant} (location and multiplicity of the zeros of $\check x_R'$) and Lemma~\ref{lem:poles of Z in H0*H0} (location and multiplicity of the poles of $\check Z_R$) of the previous section, as well as the following property of the rational function $\check Z_R$: for all $H\ne 0$,
\begin{equation}\label{eq:vanishing derivatives}
\partial_K   \check Z_R    (H,\check H_c(R)  )=0
\qt{for all }R\in (R_1, R_\infty), \qtq{and}
\partial_K^2 \check Z_{R_c}(H,\check H_c(R_c))=0 \,.
\end{equation}
These identities can be easily checked using Maple (see \cite{CAS2}).

The purpose of the following lemma is to translate the above constraints (Lemma~\ref{lem:unique dominant}, Lemma~\ref{lem:poles of Z in H0*H0} and \eqref{eq:vanishing derivatives}) on the rational functions $\check x_R$ and $\check Z_R$ in terms of the structure of the local expansion of $\mc Z(\mc h,\mc k,\nu)$ near $(\mc h,\mc k)=(0,0)$.
These constraints imply that some ``leading coefficients'' in the local expansion must vanish, and we check that no other leading coefficients vanish.
In other words, if $(\check x_R,\check Z_R)$ was a pair of generic rational functions satisfying the above constraints, then the local expansions of $Z(\mc h,\mc k,\nu)$ will have exactly the same structure and leading nonzero coefficients as those specified in Lemma~\ref{lem:singular expansion}.
After establishing Lemma~\ref{lem:singular expansion} (and Lemma~\ref{lem:singular division} which is used in its proof), we will plug the change of variables $(\mc h,\mc k) = ((1-x)^\delta, (1-y)^\delta)$ into $\mc Z(\mc h,\mc k,\nu)$ to derive the asymptotic expansion of $Z(u,v,\nu)$ near $(u,v) = (u_c(\nu),u_c(\nu))$ in Proposition~\ref{prop:local expansion}.
Apart from expressing the results in different sets of variables, another key difference between Lemma~\ref{lem:singular expansion} and Proposition~\ref{prop:local expansion} is that the former gives an exact decomposition in terms of converging series, while the latter gives asymptotic expansions useful for the study of coefficient asymptotics.


From now on we hide the parameter $\nu$ and the corresponding parameter $R$ from the notations.

\begin{lemma}\label{lem:singular expansion}
For $\nu>\nu_c$, $\mc Z(\mc h,\mc k)$ is analytic at $(0,0)$. Its Taylor expansion $\mc Z(\mc h,\mc k) = \sum_{m,n\ge 0} \mc Z_{m,n} \mc h^m \mc k^n$ satisfies $\mc Z_{1,n} = \mc Z_{n,1} = 0$ for all $n \ge 0$ and $\mc Z_{3,3}>0$.

For $\nu \in (1,\nu_c]$, we have a decomposition of the form $\mc Z(\mc h,\mc k) = Q(\mc h,\mc k) + \frac{J(\mc h\mc k)}{\mc D(\mc h,\mc k)}$, where $Q(\mc h,\mc k)$, $J(r)$ and $\mc D(\mc h,\mc k)$ are  analytic at the origin. The denominator satisfies $\mc D(0,0) = 0$ and $\partial_{\mc h} \mc D(0,0) = \partial_{\mc k} \mc D(0,0) = 1$, whereas
\begin{equation*}
Q(\mc h,\mc k) = \sum_{m,n\ge 0} Q_{m,n} \mc h^m \mc k^n
\qtq{and}
J(r) = \sum_{l\ge 1} J_l r^l
\end{equation*}
satisfy: If $\nu\in (1,\nu_c)$, then $J_1>0$.\\
~\phantom{satisfy:} If $\nu=\nu_c$, then $Q_{1,n} = Q_{n,1} = Q_{2,n} = Q_{n,2} = 0$ for all $n \ge 0$, $J_1=J_2=0$ and $J_3>0$.

The three nonzero coefficients in the above statements can be computed by:
\begin{align}
\label{eq:Z_3,3}
\mc Z_{3,3} &= \frac{1}{\check x_2^3} \mh({ \check Z_{3,3} -2 \, \frac{\check x_3}{\check x_2}\, \check Z_{2,3} + \m({\frac{\check x_3}{\check x_2}}^2 \check Z_{2,2} }
&&\hspace{-1cm}\text{when }\nu>\nu_c ,
\\
\label{eq:J_1}
J_1 &= \frac1{\check x_2^{1/2}} \lim_{H\to \check H_c} \partial_H \check Z(H,\check H_c)
&&\hspace{-1cm}\text{when }\nu\in (1,\nu_c),
\\
\label{eq:J_3}
J_3 &= \frac{1}{\check x_3^{5/3}} \cdot \frac{4}{13}\lim_{H\to \check H_c} \frac{ \partial_H \partial_K \check Z(H,H) }{ (\check H_c - H)^3 }
&&\hspace{-1cm}\text{when }\nu=\nu_c,
\end{align}
where the numbers $\check x_n$ and $\check Z_{m,n}$ are the coefficients in the Taylor expansions $1-\check x(\check H_c - h) = \sum_{n\ge 2} \check x_n h^n$ and $\check Z(\check H_c-h,\check H_c-k) = \sum \limits_{m,n} \check Z_{m,n} h^m k^n$.\\
(The coefficients $\check x_n$ are not to be confused with the functions $\check x_R = \check x(\,\cdot\,,R)$ defined earlier. There should be no confusion because by the convention above this lemma, the parameter $R$ no longer appears in our notations.)
\end{lemma}

\begin{proof}
Recall that $Z$ has the parametrization $x=\check x(H)$, $y=\check x(K)$ and $Z(u_c x, u_c y) = \check Z(H,K)$. The function $h\mapsto \mc h = \mn({ 1-\check x(\check H_c-h) }^\delta$ is analytic and has positive derivative at $h=0$. (The exponent $\delta$ has been chosen for this to be true.)
Let $\psi$ be its inverse function. Then the definition of $\mc Z$ implies that
\begin{equation}\label{eq:Z-tilde = Z-hat(psi)}
\mc Z(\mc h,\mc k) = \check Z \m({\check H_c-\psi(\mc h), \check H_c-\psi(\mc k)} \,.
\end{equation}
The proof will be based on the above formula and uses the following ingredients: The form of the local expansions of $\mc Z$ will follow from whether $(\check H_c,\check H_c)$ is a pole of $\check Z(H,K)$ or not. The vanishing coefficients will be a consequence of the vanishing of $\partial_K \check Z(H,\check H_c)$ and of $\partial_K^2 \check Z(H,\check H_c)$ given in \eqref{eq:vanishing derivatives}. Finally, the non-vanishing of the coefficients $\mc Z_{3,3}$, $J_1$ and $J_3$ will be checked by explicit computation.

\paragraph{Low temperatures ($\nu>\nu_c$).}
By Lemma~\ref{lem:poles of Z in H0*H0}, $(\check H_c,\check H_c)$ is not a pole of $\check Z(H,K)$ when $\nu>\nu_c$. Thus \eqref{eq:Z-tilde = Z-hat(psi)} implies that $\mc Z$ is analytic at $(0,0)$.
By the definition of $\check x_2$ and $\check x_3$, we have $\mb({1-\check x(\check H_c - h)}^{1/2} = \check x_2^{1/2} h \mb({ 1 + \frac{\check x_3}{2 \check x_2} h + O(h^2)}$. Then the Lagrange inversion formula gives,
\begin{equation*}
\psi(\mc h) = \frac{1}{\check x_2^{1/2}} \mc h - \frac{\check x_3}{2 \check x_2^2} \mc h^2 + O(\mc h^3) \,.
\end{equation*}
In particular, $\psi(\mc h) \sim \mathtt{cst}\cdot \mc h$.
Hence \eqref{eq:Z-tilde = Z-hat(psi)} and the fact that $\partial_K \check Z(H,\check H_c)=0$ for all $H$ (Eq.~\eqref{eq:vanishing derivatives}) imply that $\partial_{\mc k} \mc Z(\mc h,0)=0$ for all $\mc h$ close to $0$, that is, $\mc Z_{1,n}=\mc Z_{n,1}=0$ for all $n\ge 0$.
On the other hand, we get the expression \eqref{eq:Z_3,3} of $\mc Z_{3,3}$ by composing the Taylor expansions of $\psi(\mc h)$ and of $\check Z(\check H_c-h,\check H_c-k)$, while taking into account that $\check Z_{1,n}=\check Z_{n,1} = 0$.

By plugging the expressions of $\check x(H)$ and $\check Z(H,K)$ into the relation \eqref{eq:Z_3,3}, one can compute the function $\mc Z_{3,3}(\check \nu(R))$, which gives a parametrization of $\mc Z_{3,3}(\nu)$. The explicit formula, too long to be written down here, is given in \cite{CAS2}. We check in \cite{CAS2} that it is strictly positive for all $R\in (R_c,R_\infty)$.

\paragraph{High temperatures ($1<\nu<\nu_c$).}
When $\nu\in (1,\nu_c)$, Lemma~\ref{lem:poles of Z in H0*H0} tells us that $(\check H_c,\check H_c)$ is a pole of $\check Z(H,K)$. Moreover, this pole is simple in the sense that the denominator $D$ of $\check Z$ satisfies that $D(\check H_c,\check H_c)=0$ and $\partial_H D(\check H_c,\check H_c) = \partial_K D(\check H_c,\check H_c) \ne 0$.
Then it follows from \eqref{eq:Z-tilde = Z-hat(psi)} that $\mc Z = \mc N/\mc D$ for some functions $\mc N(\mc h,\mc k)$ and $\mc D(\mc h,\mc k)$, both analytic at $(0,0)$, such that $\mc D(0,0)=0$ and $\partial_{\mc h} \mc D(0,0)=\partial_{\mc h} \mc D(0,0)=1$.
We will show in Lemma~\ref{lem:singular division} below that there is always a pair of functions $Q(\mc h,\mc k)$ and $J(r)$, both analytic at the origin, such that $\mc N(\mc h,\mc k) = Q(\mc h,\mc k)\cdot \mc D(\mc h,\mc k) + J(\mc h \mc k)$. This implies the decomposition $\mc Z(\mc h,\mc k) = Q(\mc h,\mc k) + \frac{J(\mc h \mc k)}{\mc D(\mc h,\mc k)}$. Notice that $J(0)=0$, because $\mc N(0,0) = \mc D(0,0) = 0$ by the continuity of $\check Z |_{\cHdom[0] \times \cHdom[0]}$ at $(\check H_c,\check H_c)$.

Taking the derivatives of the above decomposition of $\mc Z(\mc h,\mc k)$ at $\mc k=0$ gives
\begin{equation*}
\partial_{\mc h} \mc Z(\mc h,0) = \partial_{\mc h} Q(\mc h,0)  \qtq{and}
\partial_{\mc k} \mc Z(\mc h,0) = \partial_{\mc k} Q(\mc h,0) + \frac{J_1 \cdot \mc h}{\mc D(\mc h,0)} \,.
\end{equation*}
For the same reason as when $\nu>\nu_c$, we have $\partial_{\mc k} \mc Z(\mc h,0)=0$ for all $\mc h$ close to $0$. On the other hand, $\mc D(\mc h,0)\sim \mc h$ as $\mc h\to 0$ because $\partial_{\mc h} \mc D(0,0)=1$. Thus the limit $\mc h\to 0$ of the above derivatives gives
\begin{equation*}
\lim_{\mc h\to 0} \partial_{\mc h} \mc Z(\mc h,0) = \partial_{\mc h} Q(0,0)     \qtq{and}
\partial_{\mc k} Q(0,0) + J_1 = 0.
\end{equation*}
By symmetry, $\partial_{\mc h} Q(0,0) = \partial_{\mc k} Q(0,0)$, therefore $J_1 = - \lim_{\mc h\to 0} \partial_{\mc h} \mc Z(\mc h,0)$.
After expressing $\mc Z(\mc h,0)$ in terms of $\check Z(H,\check H_c)$ and $\psi(\mc h)$ using \eqref{eq:Z-tilde = Z-hat(psi)}, we obtain the formula \eqref{eq:J_1} for $J_1$.

We check by explicit computation in \cite{CAS2} that $J_1(\nu)$ has the parametrization
\begin{equation*}
J_1 \m({\check \nu(R)} =
\frac{
\sqrt{\m({1+R^2} \m({7-R^2}^3 \m({14R^2 -1 -R^4}^5}
}{ \sqrt2\,  (3R^2 -1) \m({29 +75R^2 -17R^4 +R^6}^2   }
\end{equation*}
which is strictly positive for all $R\in (R_1,R_c)$.

\paragraph{Critical temperature ($\nu=\nu_c$).}
When $\nu=\nu_c$, the point $(\check H_c,\check H_c)$ is still a pole of $\check Z(H,K)$ by Lemma~\ref{lem:poles of Z in H0*H0}, and one can check that it is simple in the sense that $\partial_H D(\check H_c,\check H_c) = \partial_K D(\check H_c,\check H_c) \neq 0$. Therefore, the decomposition $\mc Z(\mc h,\mc k) = Q(\mc h,\mc k) + \frac{J(\mc h \mc k)}{\mc D(\mc h,\mc k)}$ remains valid. Contrary to the non-critical case, now we have $\check x_2=0$ and $\delta=1/3$, thus $\psi(\mc h) \sim \check x_3^{-1/3} \mc h$. Together with the fact that $\partial_K \check Z(H,H_c)= \partial_K^2       \check Z(H,H_c)=0$ for all $H$ (Eq.~\eqref{eq:vanishing derivatives}), this implies
$\partial_{\mc k} \mc Z(\mc h,0) = \partial_{\mc k}^2 \mc Z(\mc h,0) = 0$ for all $\mc h$ close to $0$. Plugging in the decomposition $\mc Z(\mc h,\mc k) = Q(\mc h,\mc k) + \frac{J(\mc h \mc k)}{\mc D(\mc h,\mc k)}$, we obtain
\begin{equation*}
\partial_{\mc k} Q(\mc h,0) + \frac{J_1 h}{\mc D(\mc h,0)} = 0
\qtq{and}
\partial_{\mc k}^2 Q(\mc h,0) + \frac{J_2 \mc h^2}{\mc D(\mc h,0)} - J_1 \mc h \cdot \frac{\partial_{\mc k} \mc D(\mc h,0)}{\mc D(\mc h,0)^2} = 0 \,.
\end{equation*}
Since $\partial_{\mc h} \mc D(0,0)=1$ and $\mc D(\mc h,0) \sim \mc h$ as $\mc h\to 0$, the last term in the second equation diverges like $J_1 \mc h^{-1}$ when $\mc h\to 0$, whereas the other two terms are bounded. This implies that $J_1 = 0$.
Plugging $J_1=0$ back into the two equations, we get
\begin{equation*}
\partial_{\mc h} Q(\mc h,0) = 0    \qtq{and}
\partial_{\mc k}^2 Q(\mc h,0) + \frac{J_2 \mc h^2}{\mc D(\mc h,0)} = 0 \,.
\end{equation*}
The first equation translates to $Q_{1,n}=Q_{n,1}=0$ for all $n \ge 0$. Then, $Q_{1,2}=0$ tells us that in the second equation $\partial_{\mc k}^2 Q(\mc h,0) = Q_{0,2} + O(\mc h^2)$, whereas $\frac{J_2 \mc h^2}{\mc D(\mc h,0)} \sim J_2 \mc h$ when $\mc h\to 0$. Therefore we must have $J_2=0$, which in turn implies $\partial_{\mc k}^2 Q(\mc h,0) =0$, that is, $Q_{2,n}=Q_{n,2}=0$ for all $n\ge 0$.

To obtain the formula for $J_3$, we calculate from the decomposition $\mc Z(\mc h,\mc k) = Q(\mc h,\mc k) + \frac{J(\mc h \mc k)}{\mc D(\mc h,\mc k)}$ that
\begin{align}
\partial_{\mc h} \partial_{\mc k} \mc Z(\mc h,\mc h) =
\partial_{\mc h} \partial_{\mc k} Q(\mc h,\mc h)
&+\frac{1}{\mc D(\mc h,\mc h)^2}
  \mh({
    \m({
      \mc D(\mc h,\mc h) - 2\mc h \partial_{\mc h} \mc D(\mc h,\mc h)
    } J'(h^2)
    - \partial_{\mc h} \partial_{\mc k} \mc D(\mc h,\mc h) J(h^2)
  }     \notag
\\&+\frac{1}{\mc D(\mc h,\mc h)}
  \mh({
    J''(h^2)\cdot h^2 +
    2 \mB({
        \frac{
          \partial_{\mc h} \mc D(\mc h,\mc h) }{ \mc D(\mc h,\mc h)
        }
      }^2 J(h^2)
}    \label{eq:dh dk Z-tilde}
\end{align}
When $\mc h\to 0$, we have $\partial_{\mc h} \partial_{\mc k} Q(\mc h,\mc h) = O(h^4)$ because $Q_{1,n} = Q_{n,1} = Q_{2,n} = Q_{n,2} = 0$. Moreover, using
\begin{align*}
\mc D(\mc h,\mc h)
& \sim 2\mc h &
\mc D(\mc h,\mc h) - 2\mc h \partial_{\mc h} \mc D(\mc h,\mc h)
& = O(\mc h^2) &
\partial_{\mc h} \partial_{\mc k} \mc D(\mc h,\mc h)
& = O(1) &
\frac{ \partial_{\mc h}\mc D(\mc h,\mc h) }{ \mc D(\mc h,\mc h) }
& \sim \frac{1}{2\mc h} \\
&\text{and}&
J''(\mc h^2)  & \sim  6J_3 \cdot \mc h^2 &
J' (\mc h^2)  & \sim  2J_3 \cdot \mc h^4 &
J  (\mc h^2)  & \sim   J_3 \cdot \mc h^6
\end{align*}
we see that the first line of \eqref{eq:dh dk Z-tilde} is a $O(\mc h^4)$, whereas the second line is $\frac{13}{4} J_3 \mc h^3 + O(\mc h^4)$.
Therefore we have $J_3 = \frac{4}{13} \lim_{\mc h \to 0} \mc h^{-3} \partial_{\mc h} \partial_{\mc k} \mc Z(\mc h,\mc h)$. Finally, we obtain the expression \eqref{eq:J_3} of $J_3$ using the relation $\mc Z(\mc h,\mc h) = \check Z( \check H_c-\psi(\mc h), \check H_c-\psi(\mc h))$ and the fact that $\psi(\mc h) \sim \check x_3^{-1/3} \mc h$ when $\nu=\nu_c$.

Numerical computation gives $J_3 = \frac{27}{20} \m({\frac32}^{2/3} > 0$.
\end{proof}

\begin{lemma}[Division by a symmetric Taylor series with no constant term]
\label{lem:singular division}
Let $\mc N(\mc h,\mc k)$ and $\mc D(\mc h,\mc k)$ be two symmetric holomorphic functions defined in a neighborhood of $(0,0)$.
Assume that $(0,0)$ is a simple zero of $\mc D$, that is, $\mc D(0,0) = 0$ and $\partial_{\mc h} \mc D(0,0) = \partial_{\mc k} \mc D(0,0) \ne 0$. Then there is a unique pair of holomorphic functions $Q(\mc h,\mc k)$ and $J(r)$ in neighborhoods of $(0,0)$ and $0$ respectively, such that $Q$ is symmetric and
\begin{equation}\label{eq:singular division}
\mc N(\mc h,\mc k) = Q(\mc h,\mc k)\cdot \mc D(\mc h,\mc k) + J(\mc h\mc k) \,.
\end{equation}
\end{lemma}

\begin{remark}
When $\mc D(0,0)=0$, the ratio $\frac{\mc N(\mc h,\mc k)}{\mc D(\mc h,\mc k)}$  between two Taylor series $\mc N(\mc h,\mc k)$ and $\mc D(\mc h,\mc k)$ does not in general have a Taylor expansion at $(0,0)$.
The above lemma gives a way to decompose the ratio into the sum of a Taylor series $Q(\mc h,\mc k)$ and a singular part $\frac{J(\mc h\mc k)}{\mc D(\mc h,\mc k)}$ whose numerator is determined by an univariate function. The lemma deals with the case where $\mc N(\mc h,\mc k)$ and $\mc D(\mc h,\mc k)$ are symmetric, and the zero of $\mc D(\mc h,\mc k)$ at $(0,0)$ is simple. The following remarks discuss how the lemma would change if one modifies its conditions.
\begin{enumerate}
\item
In \eqref{eq:singular division}, instead of requiring $Q(\mc h,\mc k)$ to be symmetric, we can require the remainder term to not depend on $\mc k$. Then the decomposition would become $\mc N(\mc h,\mc k) = Q(\mc h,\mc k) \cdot \mc D(\mc h,\mc k) + J(\mc h^2)$. Notice that the remainder term does not have any odd power of $\mc h$, which is a constraint due to the symmetry of $\mc N$ and $\mc D$.

Without the assumption that $\mc N$ and $\mc D$ are symmetric, we would have a decomposition $\mc N(\mc h,\mc k) = Q(\mc h,\mc k) \cdot \mc D(\mc h,\mc k) + J(\mc h)$ where the remainder is a general Taylor series $J(\mc h)$. The proof of Lemma~\ref{lem:singular division} can be adapted easily to treat the non-symmetric case.

\item
If $(0,0)$ is a zero of order $n>1$ of $\mc D$ (that is, all the partial derivatives of $\mc D$ up to order $n-1$ vanishes at $(0,0)$, while at least one partial derivative of order $n$ is nonzero), then one can prove a division formula similar to \eqref{eq:singular division}, but with a different remainder term. For example, when $n=2$, the remainder term can be written as $J_1(\mc h\mc k)\cdot (\mc h+\mc k) + J_2(\mc h\mc k)$ if $\partial_{\mc h}^2 \mc D(0,0) \ne 0$, or as $J_3(s+t)$ if $\partial_{\mc h}^2 \mc D(0,0)=0$ but $\partial_{\mc h} \partial_{\mc k} \mc D(0,0) \ne 0$.

\item
As we will see in the proof below, the decomposition \eqref{eq:singular division} can be made in the sense of formal power series without using the convergence of the Taylor series of $\mc N$ and $\mc D$. (In fact this is the easiest way to construct $Q(\mc h,\mc k)$ and $J(r)$.) The decomposition \eqref{eq:singular division} will be used in the proof of Proposition~\ref{prop:local expansion} to establish asymptotics expansions of $\mc Z(\mc h,\mc k) = \frac{\mc N(\mc h,\mc k)}{\mc D(\mc h,\mc k)}$ when $(\mc h,\mc k)\to (0,0)$. For this purpose, it is not necessary to know that the series $Q(\mc h,\mc k)$ and $J(r)$ are convergent. Everything can be done by viewing \eqref{eq:singular division} as an asymptotic expansion with a remainder term $O\m({ \max(\abs{\mc h},\abs{\mc k})^n }$ for an arbitrary $n$. However, we find that presenting $Q(\mc h,\mc k)$ and $J(r)$ as analytic functions is conceptually simpler. For this reason, we will still prove that the series $Q(\mc h,\mc k)$ and $J(r)$ are convergent even if it is not absolutely necessary for the rest of this paper.

\end{enumerate}
\end{remark}

\begin{proof}
The proof comes in two steps: first we construct order by order two series $Q(\mc h,\mc k)$ and $J(r)$ which satisfy \eqref{eq:singular division} in the sense of formal power series, and then we show that these series do converge in a neighborhood of the origin.

We approach the construction of $Q(\mc h,\mc k)$ and $J(r)$ as formal power series as follows: Assume first that $Q(\mc h,\mc k)$ and $J(r)$ are given together with the assumptions of the theorem. In that case, for all $n\ge 0$, let $\mc D_n(s,t) = [\lambda^n] \mc D(\lambda s,\lambda t)$, and similarly define $\mc N_n(s,t)$ and $Q_n(s,t)$. By construction, $\mc D_n$, $\mc N_n$ and $Q_n$ are homogeneous polynomials of degree $n$. The assumptions of the lemma ensure that $\mc D_n$ and $\mc N_n$ are symmetric, $\mc D_0=0$, and $\mc D_1(s,t) = d_{1,0}(s+t)$ where $d_{1,0} := \partial_h \mc D(0,0) \ne 0$. On the other hand, let $J_l = [r^l] J(r)$. Then \eqref{eq:singular division} is equivalent to
\begin{equation}\label{eq:singular division rec}
\mc N_n \,=\, \mc D_1 Q_{n-1} \,+\, \m({ \mc D_2 Q_{n-2} +\cdots+ \mc D_n Q_0 } \,+\, J_l \cdot (st)^l \cdot \id_{n=2l \text{ is even}}
\end{equation}
for all $n \ge 0$.
Let us show that this recursion relation indeed uniquely determines $Q_n$ and $J_l$, such that $Q_n(s,t)$ is a homogeneous polynomial of degree $n$ and $J_l\in \complex$. When $n=0$, \eqref{eq:singular division rec} gives $J_0 = \mc N_0 \in \complex$.
When $n\ge 1$, we assume as induction hypothesis that $Q_m(s,t)$ is a symmetric homogeneous polynomial of degree $m$ for all $m<n$. Then
\eqref{eq:singular division rec} can be written as
\begin{equation*}
\tilde{\mc N}_n = d_{1,0}(s+t) \cdot Q_{n-1} + J_l \cdot (st)^l \cdot \id_{n=2l \text{ is even}} \,,
\end{equation*}
where $\tilde{\mc N}_n := \mc N_n - \m({ \mc D_2 Q_{n-2} +\cdots+ \mc D_n Q_0 }$ is a symmetric homogeneous polynomial of degree $n$.
By the fundamental theorem of symmetric polynomials, a bivariate symmetric polynomial can be written uniquely as a polynomial of the elementary symmetric polynomials $s+t$ and $st$. Isolating the terms of degree zero in $s+t$, we deduce that there is a unique pair $Q_{n-1}(s,t)$ and $\tilde J_n(r)$ such that $Q_{n-1}(s,t)$ is symmetric, and
\begin{equation*}
\tilde{\mc N}_n(s,t) = d_{1,0}(s+t) \cdot Q_{n-1}(s,t) + \tilde J_n(st)\,.
\end{equation*}
Moreover, since $\tilde{\mc N}_n(s,t)$ is homogeneous of degree $n$, the polynomials $Q_{n-1}(s,t)$ and $\tilde J_n(st)$ must be homogeneous of degree $n-1$ and $n$ respectively. When $n$ is odd, this implies $\tilde J_n(st)=0$, and when $n=2l$ is even, we must have $\tilde J_n(st) = J_l \cdot (st)^l$ for some $J_l\in \complex$.
By induction, this completes the construction of $Q_n(s,t)$ and $J_l \in \complex$, such that the series defined by $Q(\lambda s,\lambda t) = \sum_n Q_n(s,t) \lambda^n$ and $J(r) = \sum_l J_l r^l$ satisfy \eqref{eq:singular division} in the sense of formal power series.

Now let us show that the series $J(r)$ has a strictly positive radius of convergence.
Since $\mc D(0,0) = 0$ and $\partial_{\mc h} \mc D(0,0) = \partial_{\mc k} \mc D(0,0) \ne 0$, by the implicit function theorem, the equation $\mc D( \mc{h,\tilde k(h)} )=0$ defines locally a holomorphic function $\mc{\tilde k}$ such that $\mc{\tilde k}(0)=0$ and $\mc{\tilde k}'(0)=-1$.
In particular, $\mc{h\cdot \tilde k(h)}$ has a Taylor expansion with leading term $-\mc h^2$, so the inverse function theorem ensures that there exists a holomorphic function $\varphi$ such that $s^2 = \varphi(s) \cdot \tilde{\mc k}(\varphi(s))$ near $s=0$.
Taking $\mc h=\varphi(s)$ and $\mc k=\tilde{\mc k}(\varphi(s))$ in \eqref{eq:singular division} gives that
\begin{equation*}
\mc N \m({ \varphi(s), \tilde{\mc k}(\varphi(s)) } = J(s^2)
\end{equation*}
in the sense of formal power series.
Since $\mc N$, $\tilde{\mc k}$ and $\varphi$ are all locally holomorphic, the series on both sides have a strictly positive radius of convergence.

It remains to prove that $Q(\mc h,\mc k)$ also converges in a neighborhood of the origin. Even though $\mc D(0,0)=0$, the Taylor series of $\mc D(\mc h,\mc k)$ still has a multiplicative inverse in the space of formal Laurent series $\complex(\!(x)\!)[\![y]\!]$. Therefore we can rearrange Equation \eqref{eq:singular division} to obtain in that space
\begin{equation*}
Q(\mc h,\mc k) = \frac{\mc N(\mc h,\mc k) - J(\mc h\mc k)}{\mc D(\mc h,\mc k)}.
\end{equation*}
The \rhs, which will be denoted by $f(\mc h,\mc k)$ below, is a holomorphic function in a neighborhood of $(0,0)$ outside the zero set of $\mc D(\mc h,\mc k)$.
As seen in the previous paragraph, this zero set is locally the graph of the function $\mc{\tilde k(h) \sim -h}$ when $\mc h\to 0$.
It follows that there exists $\delta>0$ such that $f$ is holomorphic in a neighborhood of $(\cdisk_{3 \delta} \setminus \disk_{2\delta}) \times \cdisk_\delta$, where $\cdisk_{3 \delta} \setminus \disk_{2\delta}$ is the closed annulus of outer and inner radii $3\delta$ and $2\delta$ centered at the origin. The usual Cauchy integral formula for the coefficient of Laurent series gives
\begin{equation*}
Q_{m,n} = \m({ \frac{1}{2\pi i} }^2 \oint_{\partial \disk_\delta} \frac{\dd \mc k}{\mc k^{n+1}} \m({
\oint_{\partial \disk_{3\delta}} \frac{\dd \mc h}{\mc h^{m+1}} f(\mc h,\mc k) -
\oint_{\partial \disk_{2\delta}} \frac{\dd \mc h}{\mc h^{m+1}} f(\mc h,\mc k) }.
\end{equation*}
However, by construction, the Laurent series $\sum_{m \in \integer} Q_{m,n} \mc h^m$ does not contain any negative power of $\mc h$. This implies that the integral over $\partial \disk_{2 \delta}$ in the above formula has zero contribution.
Therefore we have
\begin{equation*}
\abs{Q_{m,n}}
=   \abs{ \m({ \frac{1}{2\pi} }^2 \oiint_{\partial \disk_{3\delta} \times \partial \disk_\delta} \frac{\dd \mc h\, \dd \mc k}{\mc h^{m+1} \mc k^{n+1}} f(\mc h,\mc k) }
\le (3 \delta)^{-m} \delta^{-n} \! \cdot \!\! \sup_{\partial \disk_{3\delta} \times \partial \disk_\delta} \!\! |f| \,.
\end{equation*}
It follows that the series $Q(\mc h,\mc k) = \sum Q_{m,n} \mc h^m \mc k^n$ converges in a neighborhood of $(0,0)$.
\end{proof}

\begin{proposition}[Asymptotic expansions of $Z(u,v)$]\label{prop:local expansion}
Let $\epsilon=\epsilon(\nu,\theta)>0$ be a value for which the holomorphicity result of Proposition~\ref{prop:singularity structure} and the  bound in Corollary~\ref{cor:double angle bound} hold. Then for $(x,y)$ varying in $\cslit \times \cslit$ (when $\nu\ge \nu_c$) or $\cslit \times \cDdom$ (when $1<\nu<\nu_c$), we have
\begin{align}
\label{eq:asym Z 1}
Z(u_c x,u_c y) &= Z\1{reg}\01(x,y) + A(u_c x)\cdot (1-y)^{\alpha_0} + O\mb({ (1-y)^{\alpha_0 + \delta} } \qt{for }x\ne 1\text{ and}\hspace{-2ex}    &&\text{as } y\to 1 \,,
\\
\label{eq:asym Z 2}
A(u_c x) &= A\1{reg}(x) + b\cdot (1-x)^{\alpha_1} + O\mb({ (1-x)^{\alpha_1 + \delta} }  &&\text{as } x\to 1 \,,
\\
\label{eq:asym Z diag}
Z(u_c x,u_c y) &= Z\1{reg}\02(x,y) + b\cdot Z\1{hom}(1-x,1-y) + O\m({ \max( \abs{1-x}, \abs{1-y})^{\alpha_2+\delta} }
&&\text{as } (x,y)\to (1,1) \,,
\end{align}
where $b=b(\nu)$ is a number determined by the nonzero constants $\mc Z_{3,3}$, $J_1$ and $J_3$ in Lemma~\ref{lem:singular expansion}, and $Z\1{hom}(s,t)$ is a homogeneous function of order $\alpha_2$ (i.e.\ $Z\1{hom}(\lambda s,\lambda t) = \lambda^{\alpha_2} Z\1{hom}(s,t)$ for all $\lambda>0$) that only depends on the phase of the model. Explicitly:
\begin{equation*}
b(\nu) = \begin{cases}
\mc Z_{3,3}(\nu) & \text{when } \nu > \nu_c     \\
J_1(\nu)         & \text{when } \nu\in(1,\nu_c) \\
-J_3(\nu_c)      & \text{when } \nu = \nu_c
\end{cases}
\qtq{and}
Z\1{hom}(s,t) = \begin{cases}
      s^{3/2} t^{3/2}                   & \text{when } \nu > \nu_c     \\
\frac{s^{1/2} t^{1/2}}{s^{1/2}+t^{1/2}} & \text{when } \nu\in(1,\nu_c) \\
\frac{-st            }{s^{1/3}+t^{1/3}} & \text{when } \nu = \nu_c \,.
\end{cases}
\end{equation*}
%
On the other hand, $Z\1{reg}\01(x,y) = Z(u_cx,u_c) - \partial_v Z(u_cx,u_c) \cdot u_c\cdot (1-y)$ is an affine function of $y$ satisfying
\begin{equation}\label{eq:local integrability}
Z(u_cx,u_c) = O(1) \qtq{and}
\partial_v Z(u_cx,u_c) = O((1-x)^{-1/2}) \qt{when } x\to 1\,,
\end{equation}
whereas $A\1{reg}(x)$ is an affine function of $x$, and $Z\1{reg}\02(x,y) = Z\1{reg}\01(x,y) + Z\1{reg}\01(y,x) - P(x,y)$ for some polynomial $P(x,y)$ that is affine in both $x$ and $y$. The functions $A\1{reg}(x)$ and $P(x,y)$ will be given in the proof of the proposition.
\end{proposition}

\begin{remark}\label{rem:local expansion}
For a fixed $x$, \eqref{eq:asym Z 1} is an univariate asymptotic expansion in the variable $y$. It has the form
\begin{center}
(analytic function of $y$ near $y=1$) + \texttt{constant} $\cdot\, (1-y)^{\alpha_0}$ + $o((1-y)^{\alpha_0})$,
\end{center}
which makes it a suitable input to the classical transfer theorem of analytic combinatorics. More precisely, when we extract the coefficient of $[y^q]$ from \eqref{eq:asym Z 1} using contour integrals on $\partial \slit$, the contribution of the first term will be exponentially small in $q$, whereas the contributions of the second and the third terms will be of order $q^{-(\alpha_0+1)}$ and $o(q^{-(\alpha_0+1)})$, respectively. Similar remarks can be made for \eqref{eq:asym Z 2} \wrt\ the variable $x$.

The asymptotic expansion \eqref{eq:asym Z diag} has a form that generalizes \eqref{eq:asym Z 1} and \eqref{eq:asym Z 2} in the bivariate case.
Instead of being analytic \wrt\ $x$ or $y$, the first term $Z\1{reg}\02(x,y)$ is a linear combination of terms of the form $F(x) G(y)$ or $G(x) F(y)$, where $F(x)$ is analytic in a neighborhood of $x=1$, and $G(x)$ is locally integrable on the contour $\partial \slit$ near $x=1$. (The local integrability is a consequence of \eqref{eq:local integrability}.) As we will see in Section~\ref{sec:proof diagonal}, a term of this form will have an exponentially small contribution to the coefficient of $[x^py^q]$ in the diagonal limit where $p,q \to\infty$ and that $q/p$ is bounded away from $0$ and $\infty$. On the other hand, the homogeneous function $Z\1{hom}(1-x,1-y)$ is a generalization of the power functions $(1-y)^{\alpha_0}$ and $(1-x)^{\alpha_1}$ of the univariate case.
Indeed, the only homogeneous functions of order $\alpha$ of one variable $s$ are constant multiples of $s^\alpha$. We will see in Section~\ref{sec:proof diagonal} that the term $Z\1{hom}(1-x,1-y)$ gives the dominant contribution of order $p^{-(\alpha_2+2)}$ to the coefficient of $[x^py^q]$ in the diagonal limit.
\end{remark}

\begin{proof}
First consider the non-critical temperatures $\nu\ne \nu_c$. In this case we have $\delta=1/2$, and the definition of $\mc Z(\mc h,\mc k)$ reads $Z(u_c x,u_c y) = \mc Z \mn({(1-x)^{1/2}, (1-y)^{1/2}}$.
As seen in the proof of Lemma~\ref{lem:singular expansion}, for any $\mc h\ne 0$ close to zero, the function $\mc k\mapsto \mc Z(\mc h,\mc k)$ is analytic at $\mc k=0$ and satisfies $\partial_{\mc k} \mc Z(\mc h,0)=0$. Hence it has a Taylor expansion of the form
\begin{equation*}
\mc Z(\mc h,\mc k) = \mc Z(\mc h,0)
+ \frac12 \partial_{\mc k}^2 \mc Z(\mc h, 0) \cdot \mc k^2
+ \frac16 \partial_{\mc k}^3 \mc Z(\mc h, 0) \cdot \mc k^3 + O(\mc k^4) \,.
\end{equation*}
Plugging $\mc h = (1-x)^{1/2}$ and $\mc k=(1-y)^{1/2}$ into the above formula gives the expansion \eqref{eq:asym Z 1} with $\alpha_0=3/2$, $Z\1{reg}\01(x,y) = \mc Z((1-x)^{1/2},0) + \frac12 \partial_{\mc k}^2 \mc Z((1-x)^{1/2},0)\cdot (1-y)$, and
\begin{equation*}
A(u_cx) = \frac16 \partial_{\mc k}^3 \mc Z((1-x)^{1/2},0) \,.
\end{equation*}
We can identify the coefficients in the affine function $y\mapsto Z\01\1{reg}(x,y)$ as $\mc Z((1-x)^{1/2},0)= Z(u_c x,u_c)$ and $\frac12 \partial_{\mc k}^2 \mc Z((1-x)^{1/2},0) = -u_c \cdot \partial_v Z(u_cx,u_c)$. The first term is continuous at $x=1$, thus of order $O(1)$ when $x\to 1$. For the second asymptotics of \eqref{eq:local integrability}, it suffices to show that $\partial_{\mc k}^2 \mc Z(\mc h,0) = O(\mc h^{-1})$.

\paragraph{Low temperatures ($\nu>\nu_c$).}
In this case, $\mc Z(\mc h,\mc k) = \sum_{m,n} \mc Z_{m,n}\, \mc h^m \mc k^n$ with $\mc Z_{1,n}=\mc Z_{n,1}=0$. Hence
\begin{equation*}
A(u_c x) = \sum_{m\ne 1} \mc Z_{m,3}(1-x)^{m/2} = \mc Z_{0,3} + \mc Z_{2,3}\cdot (1-x) + \mc Z_{3,3} \cdot (1-x)^{3/2} + O\m({ (1-x)^2 } \,,
\end{equation*}
which gives the expansion \eqref{eq:asym Z 2} with $\alpha_1=3/2$, $A\1{reg}(x) = \mc Z_{0,3} + \mc Z_{2,3}\cdot (1-x)$ and $b=\mc Z_{3,3}>0$.
Moreover, since $\mc Z$ is analytic at $(0,0)$, we have obviously $\partial_{\mc k}^2 \mc Z(\mc h,0) = O(1)$, which is also an $O(\mc h^{-1})$.

On the other hand, by regrouping terms in the expansion $\mc Z(\mc h,\mc k) = \sum_{m,n} \mc Z_{m,n}\, \mc h^m \mc k^n$, one can write
\begin{align*}
\mc Z(\mc h,\mc k) = &
  \sum_{m\ge 0} \mc Z_{m,0}\, \mc h^m
+ \sum_{n\ge 0} \mc Z_{0,n}\, \mc k^n
- \mc Z_{0,0}
+ \sum_{m\ge 2} \mc Z_{m,2}\, \mc h^m \cdot \mc k^2
+ \sum_{n\ge 2} \mc Z_{2,n}\, \mc k^n \cdot \mc h^2
- \mc Z_{2,2}\, \mc h^2 \mc k^2  \\&
+ \mc Z_{3,3}\, \mc h^3 \mc k^3
+ O\m({ \max(|\mc h|,|\mc k|)^7 } \,.
\end{align*}
After plugging in $\mc h = (1-x)^{1/2}$ and $\mc k=(1-y)^{1/2}$, we can identify the first line on the \rhs\ as $Z\1{reg}\02(x,y) = Z\1{reg}\01(x,y) + Z\1{reg}\01(y,x) - P(x,y)$ with $P(x,y) = \mc Z_{0,0} + \mc Z_{2,0} \cdot (1-x) + \mc Z_{0,2} \cdot (1-y) + \mc Z_{2,2} \cdot (1-x)(1-y)$. The term $\mc Z_{3,3} \, \mc h^3 \, \mc k^3$ becomes $b\cdot (1-x)^{3/2}(1-y)^{3/2}$. Thus we obtain the expansion \eqref{eq:asym Z diag} with $\alpha_2=3$ and $Z\1{hom}(s,t) = s^{3/2}t^{3/2}$.

\paragraph{High temperatures ($1<\nu<\nu_c$).}
In this case, we have $\mc Z(\mc h,\mc k) = Q(\mc h,\mc k) + \frac{J(\mc h\mc k)}{\mc D(\mc h,\mc k)}$. Straightforward computation gives that
\begin{align*}
\partial_{\mc k}^2 \mc Z(\mc h,0) &= \partial_{\mc k}^2 Q(\mc h,0)
+ \frac{2 J_2 \mc h^2}{\mc D(\mc h,0)}
- 2J_1 \mc h \cdot \frac{\partial_{\mc k} \mc D(\mc h,0)}{\mc D(\mc h,0)^2}
\\
\partial_{\mc k}^3 \mc Z(\mc h,0) &= \partial_{\mc k}^3 Q(\mc h,0)
+ \frac{6 J_3 \mc h^3}{\mc D(\mc h,0)}
- 6J_2 \mc h^2 \cdot \frac{\partial_{\mc k} \mc D(\mc h,0)}{\mc D(\mc h,0)^2}
- 3J_1 \mc h \cdot \frac{\partial_{\mc k}^2 \mc D(\mc h,0)}{\mc D(\mc h,0)^2}
+ 6J_1 \mc h \cdot \frac{( \partial_{\mc k} \mc D(\mc h,0) )^2}{\mc D(\mc h,0)^3} \,.
\end{align*}
Using the fact that $Q(\mc h,\mc k)$ is analytic at $(0,0)$, and $\mc D(\mc h,0)\sim \mc h$, $\partial_{\mc k} \mc D(0,0) = 1$ and $\partial_{\mc k}^2 \mc D(\mc h,0) = O(1)$ when $\mc h\to 0$, we see that $\partial_{\mc k}^2 \mc Z(\mc h,0) = O(\mc h^{-1})$, whereas all terms in the expansion of $\partial_{\mc k}^3 \mc Z(\mc h,0)$ are of order $O(\mc h^{-1})$, except the last term, which is asymptotically equivalent to $6J_1 \mc h^{-2}$.
It follows that
\begin{equation*}
A(u_c x) = \frac16 \partial_{\mc k}^3 \mc Z((1-x)^{1/2},0)
= J_1\cdot (1-x)^{-1} + O\m({ (1-x)^{-1/2} }\,,
\end{equation*}
which gives the expansion \eqref{eq:asym Z 2} with $\alpha_1=-1$, $A\1{reg}(x) = 0$ and $b=J_1>0$.

On the other hand, Corollary~\ref{cor:double angle bound} and the relations $\check H_c-H\sim \mathtt{cst} \cdot \mc h$ and $\check H_c-K\sim \mathtt{cst} \cdot \mc k$ imply that $\max(|\mc h|,|\mc k|)$ is bounded by a constant times $|\mc h+\mc k|$ when $(x,y)\to (1,1)$ in $\cslit \times \cDdom$. It follows that
\begin{equation}\label{eq:1/D estimate}
\frac{1}{\mc h+\mc k} = O\m({ \max(|\mc h|,|\mc k|)^{-1} }
\qtq{and}
\frac{1}{\mc D(\mc h,\mc k)} =
\frac{1}{\mc h+\mc k + O\m({(\mc h+\mc k)^2} }
= \frac{1}{\mc h+\mc k} + O(1) \,.
\end{equation}
From these we deduce that $\frac{J(\mc h\mc k)}{\mc D(\mc h,\mc k)} = \frac{J_1 \mc h \mc k}{\mc h + \mc k} + O\m({ \max(|\mc h|,|\mc k|)^2 }$.
Thus we can regroup terms in the decomposition $\mc Z(\mc h,\mc k) = Q(\mc h,\mc k) + \frac{J(\mc h\mc k)}{\mc D(\mc h,\mc k)}$ to get
\begin{equation*}
\mc Z(\mc h,\mc k) =
  \sum_{m\ge 0} Q_{m,0} \mc h^m
+ \sum_{n\ge 0} Q_{0,n} \mc k^n - Q_{0,0}
+ \frac{J_1\mc h \mc k}{\mc h + \mc k}
+ O\m({ \max(|\mc h|,|\mc k|)^2 }
\end{equation*}
After plugging in $\mc h = (1-x)^{1/2}$ and $\mc k=(1-y)^{1/2}$, we can identify the first three terms on the \rhs\ as $Z\1{reg}\02(x,y) = Z\1{reg}\01(x,y) + Z\1{reg}\01(y,x) - Q_{0,0}$ up to a term of order $O\m({ \max(|1-x|,|1-y|) }$. The term $\frac{J_1\mc h \mc k}{\mc h + \mc k}$ becomes $b\cdot \frac{(1-x)^{1/2} (1-y)^{1/2}}{(1-x)^{1/2} + (1-y)^{1/2}}$. Thus we obtain \eqref{eq:asym Z diag} with $\alpha_2=1/2$ and $Z\1{hom}(s,t) = \frac{s^{1/2} t^{1/2}}{s^{1/2} + t^{1/2}}$.

\paragraph{Critical temperature ($\nu=\nu_c$).}
At the critical temperature, $\delta=1/3$ and the definition of $\mc Z(\mc h,\mc k)$ reads $Z(u_c x,u_c y) = \mc Z((1-x)^{1/3},(1-y)^{1/3})$. In this case, $\mc k \mapsto \mc Z(\mc h,\mc k)$ has a Taylor expansion of the form
\begin{equation*}
\mc Z(\mc h,\mc k) = \mc Z(\mc h,0)
+ \frac16 \partial_{\mc k}^3 \mc Z(\mc h, 0) \cdot \mc k^3
+ \frac1{24} \partial_{\mc k}^4 \mc Z(\mc h, 0) \cdot \mc k^4 + O(\mc k^5) \,,
\end{equation*}
because $\partial_{\mc k} \mc Z(\mc h,0)=\partial_{\mc k}^2 \mc Z(\mc h,0)=0$.
Plugging $\mc h = (1-x)^{1/3}$ and $\mc k=(1-y)^{1/3}$ into the above formula gives \eqref{eq:asym Z 1} with $\alpha_0=4/3$, $Z\1{reg}\01(x,y) = \mc Z((1-x)^{1/3},0) + \frac16 \partial_{\mc k}^3 \mc Z((1-x)^{1/3},0)\cdot (1-y)$, and
\begin{equation*}
A(u_c x) = \frac1{24} \partial_{\mc k}^4 \mc Z((1-x)^{1/3},0) \,.
\end{equation*}
As in the non-critical case, we identify $\mc Z((1-x)^{1/3},0)= Z(u_c x,u_c)$ and $\frac16 \partial_{\mc k}^3 \mc Z((1-x)^{1/3},0) = -u_c \cdot \partial_v Z(u_cx,u_c)$. The first term is still continuous at $x=1$, thus of order $O(1)$ when $x\to 1$. Let us show that $\partial_{\mc k}^3 \mc Z(\mc h,0)$ is analytic at $\mc h=0$ so that the second term is also continuous.

From the expansion $\mc Z(\mc h,\mc k) = Q(\mc h,\mc k) + \frac{J(\mc h\mc k)}{\mc D(\mc h,\mc k)}$ with $J_1=J_2=0$ and $Q_{1,n}=Q_{2,n}=0$ for all $n$, we obtain
\begin{align*}
\partial_{\mc k}^3 \mc Z(\mc h,0) &= \partial_{\mc k}^3 Q(\mc h,0)
+ \frac{6 J_3 \mc h^3}{\mc D(\mc h,0)}
\\
\frac{1}{24} \partial_{\mc k}^4 \mc Z(\mc h,0) &= Q_{0,4}
+ \sum_{m\ge 3} Q_{m,4} \mc h^m
+ \frac{J_4 \mc h^4}{\mc D(\mc h,0)}
- J_3 \mc h^3 \cdot \frac{\partial_{\mc k} \mc D(\mc h,0)}{\mc D(\mc h,0)^2} \,.
\end{align*}
Recall that $\mc D(\mc h,0)\sim \mc h$ and $\partial_{\mc k} \mc D(0,0)=1$. Then it is not hard to see that $\partial_{\mc k}^3 \mc Z(\mc h,0)$ is analytic at $\mc h=0$.
On the other hand, the second and the third terms in the expansion of $\frac{1}{24} \partial_{\mc k}^4 \mc Z(\mc h,0)$ are $O(\mc h^3)$, whereas the last term is equivalent to $J_3 \mc h$. It follows that
\begin{equation*}
A(u_c x) = \frac1{24} \partial_{\mc k}^4 \mc Z((1-x)^{1/3},0)
= Q_{0,4} - J_3 \cdot (1-x)^{1/3} + O\m({ (1-x)^{2/3} } \,,
\end{equation*}
which gives the expansion \eqref{eq:asym Z 2} with $\alpha_1=1/3$, $A\1{reg}(x) = Q_{0,4}$ and $b=-J_3<0$.

As in the high temperature case, we still have the estimate \eqref{eq:1/D estimate} when $(\mc h,\mc k)\to (0,0)$ such that the corresponding $(x,y)$ varies in $\cslit \times \cslit$.
Moreover, at the critical temperature we have $J_1=J_2=0$ and $Q_{1,n} = Q_{n,1} = Q_{2,n} = Q_{n,2} = 0$ for all $n$. Therefore $\frac{J(\mc h\mc k)}{\mc D(\mc h,\mc k)} = \frac{J_3 (\mc h \mc k)^3}{\mc h + \mc k} + O\m({ \max(|\mc h|,|\mc k|)^6 }$, and we can regroup terms in the decomposition $\mc Z(\mc h,\mc k) = Q(\mc h,\mc k) + \frac{J(\mc h\mc k)}{\mc D(\mc h,\mc k)}$ to get
\begin{align*}
\mc Z(\mc h,\mc k) =&
  \sum_{m\ge 0} Q_{m,0} \mc h^m
+ \sum_{n\ge 0} Q_{0,n} \mc k^n
- Q_{0,0}
+ \sum_{m\ge 3} Q_{m,3} \mc h^m \cdot \mc k^3
+ \sum_{n\ge 3} Q_{3,n} \mc k^n \cdot \mc h^3 \\&
+ \frac{J_3(\mc h \mc k)^3}{\mc h + \mc k}
+ O\m({ \max(|\mc h|,|\mc k|)^6 }
\end{align*}
After plugging in $\mc h = (1-x)^{1/3}$ and $\mc k=(1-y)^{1/3}$, we can identify the terms on the first line of the \rhs\ as $Z\1{reg}\02(x,y) = Z\1{reg}\01(x,y) + Z\1{reg}\01(y,x) - P(x,y)$ up to a term of order $O\m({ \max(|1-x|,|1-y|)^2 }$, where $P(x,y) = Q_{0,0} + Q_{3,0}\cdot (1-x) + Q_{0,3}\cdot (1-y)$. The term $\frac{J_3(\mc h \mc k)^3}{\mc h + \mc k}$ becomes $-b\cdot \frac{(1-x)(1-y)}{(1-x)^{1/3} + (1-y)^{1/3}}$. Thus we obtain \eqref{eq:asym Z diag} with $\alpha_2=5/3$ and $Z\1{hom}(s,t) = \frac{-st}{s^{1/3} + t^{1/3}}$.
\end{proof}

\begin{corollary}\label{cor:A Delta-analytic}
The function $x\mapsto A(u_c x)$ has an analytic continuation on $\slit$.
\end{corollary}

\begin{proof}
We have seen in the proof of Proposition~\ref{prop:local expansion} that $A(u_c x) = \frac1{m!} \partial_{\mc k}^m \mc Z((1-x)^\delta,0)$, where $m=\frac1\delta+1$ is equal to $3$ when $\nu\ne \nu_c$, and equal to $4$ when $\nu=\nu_c$. The change of variable $\mc h = (1-x)^\delta$ defines a conformal bijection from $x\in \slit$ to some simply connected domain $\mc U_\epsilon$ whose boundary contains the point $\mc h=0$.

In the proof of Lemma~\ref{lem:singular expansion}, we have shown that the mapping $h\mapsto \mc h = \mn({1-\check x(\check H_c-h)}^\delta$ has an analytic inverse $\psi(\mc h)$ in a neighborhood of $\mc h=0$ such that $\mc Z(\mc h,\mc k) = \check Z(\check H_c -\psi(\mc h),\check H_c -\psi(\mc k))$.
Let $\Psi(\mc h)=\check H_c-\psi(\mc h)$, then $\Psi$ is a local analytic inverse of the mapping $H\mapsto \m({1-\check x(H)}^\delta$, and
\begin{equation}\label{eq:Z-tilde = Z-hat(psi) bis}
\mc Z(\mc h,\mc k) = \check Z(\Psi(\mc h),\Psi(\mc k)) \,.
\end{equation}
By Lemma~\ref{lem:unique dominant}, $\check x$ defines a conformal bijection from $\Hdom$ to $\slit$. On the other hand, $x\mapsto (1-x)^\delta$ is a conformal bijection from $\slit$ to $\mc U_\epsilon$. It follows that $\Psi$ can be extended to a conformal bijection from $\mc U_\epsilon$ to $\Hdom$.

Now fix some $x_*\in \slit$ and the corresponding $\mc h_* = (1-x_*)^\delta \in \mc U_\epsilon$ and $H_* =\Psi(\mc h_*) \in \Hdom$.
Let $\mc K \subset \Hdom$ be a compact neighborhood of $H_*$. According to Remark~\ref{rem:extended holomorphicity Z}, there exists an open set $\mc V$ containing $\cHdom[0]$ such that $\check Z$ is holomorphic in $\mc K\times \mc V$. As $\check H_c\in \mc V$, this implies in particular that $\check Z$ is analytic at $(H_*,\check H_c)$.
Since $\Psi(\mc h_*)=H_*$ and $\Psi(0)=\check H_c$, and we have seen that $\Psi$ is analytic at both $\mc h_*$ and $0$, the relation \eqref{eq:Z-tilde = Z-hat(psi) bis} implies that $\mc Z$ is analytic at $(\mc h_*,0)$. It follows that $A(u_c x) = \frac{1}{m!} \partial_{\mc k}^m \mc Z((1-x)^\delta,0)$ is analytic at $x=x_*$.
\end{proof}

\begin{corollary}
A parametrization of $x\mapsto A(u_c x)$ is given by $x=\check x(H)$ and
\begin{equation*}
\check A(H) = \begin{cases}
\frac{1}{\check x_2^{3/2}} \mB({ \check Z_3(H) - \frac{\check x_3}{\check x_2} \check Z_2(H) } & \text{when }\nu \ne \nu_c \\
 \frac{1}{\check x_3^{4/3}} \mB({ \check Z_4(H) - \frac{\check x_4}{\check x_3} \check Z_3(H) } & \text{when }\nu = \nu_c
\end{cases}
\end{equation*}
where $\check x_n$ are defined as in Lemma~\ref{lem:singular expansion}, and $\check Z_n(H)$ are defined by the Taylor expansion $\check Z(H,\check H_c-k) = \sum_n \check Z_n(H) k^n$.
\end{corollary}

\begin{proof}
We have seen in the previous proof that $A(u_c x) = \frac{1}{m!} \partial_{\mc k} \mc Z((1-x)^\delta,0)$ with $m=3$ if $\nu\ne \nu_c$ and $m=4$ if $\nu=\nu_c$. Moreover, $\mc Z$ satisfies $\mc Z(\mc h,\mc k) = \check Z(\check H_c -\psi(\mc h), \check H_c -\psi(\mc k))$, where $\psi(h)$ is the local inverse of $h\mapsto (1-\check x(\check H_c-h))^\delta$. It follows that
\begin{equation}\label{eq:A as Taylor coefficient}
\check A(H) \equiv A(u_c \cdot \check x(H)) = \frac{1}{m!} \partial_{\mc k}^m \check Z\m({ H,\check H_c - \psi(\mc k)}\big|_{\mc k=0 \ .}
\end{equation}
Using the definition of the coefficients $\check x_n$ and the Lagrange inversion formula, it is not hard to obtain that
\begin{equation*}
\psi(\mc k) = \begin{cases}
\frac{1}{\check x_2^{1/2}} \mc k - \frac{\check x_3}{2 \check x_2^2} \mc k^2 + O(\mc k^3)       & \text{when }\nu \ne \nu_c \\
\frac{1}{\check x_3^{1/3}} \mc k - \frac{\check x_4}{3 \check x_3^{5/3}} \mc k^2 + O(\mc k^3) & \text{when }\nu = \nu_c \,.
\end{cases}
\end{equation*}
Now plug $k=\psi(\mc k)$ into $\check Z(H,\check H_c-k) = \sum_n \check Z_n(H) k^n$, and compute the Taylor expansion in $\mc k$ while taking into account the fact that $\check Z_1(H)=0$ for all $\nu$ and $\check Z_2(H)=0$ when $\nu=\nu_c$ (see Equation~\eqref{eq:vanishing derivatives}). According to \eqref{eq:A as Taylor coefficient}, $\check A(H)$ is given by the coefficient of $\mc k^m$ in this Taylor expansion. Explicit expansion gives the expressions in the statement of the corollary.
\end{proof}

\section{Coefficient asymptotics of $Z(u,v,\nu)$ --- proof of Theorem~\ref{thm:asympt}}\label{sec:coeff asymp}

Theorem~\ref{thm:asympt} gives the asymptotics of $z_{p,q}$ when $p,q\to \infty$ in two regimes: either $p \to\infty$ after $q \to\infty$, or $p \to\infty$ and $q \to\infty$ simultaneously while $q/p$ stays in some compact interval $[\lambda_{\min},\lambda_{\max}] \subset (0,\infty)$. We will call the first case \emph{two-step asymptotics}, and the second case \emph{diagonal asymptotics}. Let us prove the two cases separately.

\subsection{Two-step asymptotics}\label{sec:proof two-step}

At the critical temperature $\nu=\nu_c$, the two-step asymptotics of $z_{p,q}$ has already been established in \cite{CT20}. The basic idea is to apply the classical transfer theorem \cite[Corollary VI.1]{FS09} to the function $y \mapsto Z(u_c x,u_c y)$ to get the asymptotics of $z_{p,q}$ when $q\to\infty$, and then to the function $x\mapsto A(u_c x)$ to get the asymptotics of $a_p$ when $p\to\infty$.
Proposition~\ref{prop:singularity structure}~and~\ref{prop:local expansion} provide all the necessary input for extending the same schema of proof to non-critical temperatures.

\begin{proof}[Proof of Theorem~\ref{thm:asympt} --- two-step asymptotics]
According to Proposition~\ref{prop:singularity structure}, for any fixed $x\in \slit$, the function $y\mapsto Z(u_c x,u_c y)$ is holomorphic in the $\Delta$-domain $\Ddom$. And \eqref{eq:asym Z 1} of Proposition~\ref{prop:local expansion} states that, as $y \to 1$ in $\Ddom$, the dominant singular term in the asymptotic expansion of $y\mapsto Z(u_c x,u_c y)$ is $A(u_c x) \cdot (1-y)^{\alpha_0}$.
It follows from the transfer theorem that
\begin{equation}\label{eq:asym Z_q}
u_c^q \cdot Z_q(u_c x) \eqv q \frac{A(u_c x)}{\Gamma(-\alpha_0)} \cdot q^{-(\alpha_0+1)}
\end{equation}
(Recall that $Z_q(u)$ is the coefficient of $v^q$ in the generating function $Z(u,v)$.) The above asymptotics is valid for all $x\in \cslit \setminus \{1\}$. It does not always hold at $x=1$ because $A(u_c)=\infty$ in the high temperature regime.
However, if we replace $x$ by $\frac{u_0}{u_c}x$ for some arbitrary $u_0\in (0,u_c)$, then the asymptotics is valid for all $x\in \cslit$. Then, by dividing the asymptotics by the special case of itself at $x=1$, we obtain the convergence
\begin{equation*}
\frac{Z_q(u_0 x)}{Z_q(u_0)} \cv[]q \frac{A(u_0 x)}{A(u_0)}
\end{equation*}
for all $x\in \cslit$. For each $q$, the \lhs\ is the generating function of a nonnegative sequence $\m({ \frac{u_0^p \cdot z_{p,q}}{Z_q(u_0)} }_{p\ge 0}$ which always sums up to $1$ (that is, a probability distribution on $\natural$). According to a general continuity theorem \cite[Theorem~IX.1]{FS09}, this implies the convergence of the sequence term by term:
\begin{equation*}
\frac{u_0^p \cdot z_{p,q}}{ Z_q(u_0) } \cv[]q \frac{u_0^p \cdot a_p}{A(u_0)}
\end{equation*}
for all $p\ge 0$. On the other hand, \eqref{eq:asym Z_q} implies that $u_c^q\cdot Z_q(u_0) \eqv q \frac{A(u_0)}{\Gamma(-\alpha_0)} \cdot q^{-(\alpha_0+1)}$. Multiplying this equivalence with the above convergence gives the asymptotics of $z_{p,q}$ when $q\to\infty$ in Theorem~\ref{thm:asympt}.

The asymptotics of $a_p$ in Theorem~\ref{thm:asympt} is a direct consequence of the transfer theorem, given the asymptotic expansion \eqref{eq:asym Z 2} of $x\mapsto A(u_c x)$ in Proposition~\ref{prop:local expansion} and its $\Delta$-analyticity in Corollary~\ref{cor:A Delta-analytic}.
\end{proof}

\subsection{Diagonal asymptotics}\label{sec:proof diagonal}
\newcommand{\Vet}{V_{\epsilon,\theta}}
\newcommand{\Ve}{V_\epsilon}

In the diagonal limit, we have not found a general transfer theorem in the literature that allows one to deduce asymptotics of the coefficients $z_{p,q}$ from asymptotics of the generating function $Z(u_cx,u_cy)$.
However, it turns out that with the ingredients given in Proposition~\ref{prop:singularity structure}~and~\ref{prop:local expansion}, we can generalize the proof of the classical transfer theorem in \cite{FS09} to the diagonal limit in the case of the generating function $Z(u_cx,u_cy)$. Let us first describe (a simplified version of) the proof in \cite{FS09}, before generalizing it to prove the diagonal asymptotics in Theorem~\ref{thm:asympt}:

Given a generating function $F(x) = \sum_n F_n x^n$ with a unique dominant singularity at $x=1$ and an analytic continuation up to the boundary of a $\Delta$-domain $\cDdom$, one first expresses the coefficients of $F(x)$ as contour integrals on the boundary of $\Ddom$
\begin{equation*}
F_n = \frac{1}{2\pi i} \oint_{\partial \Ddom} \frac{F(x)}{x^{n+1}} \dd x \,.
\end{equation*}
Next, one shows that the integral on the circular part of $\Ddom$ is exponentially small in $n$ and therefore
\begin{equation*}
F_n = \frac{1}{2\pi i} \int_{\Vet} \frac{F(x)}{x^{n+1}} \dd x + O \m({ (1+\epsilon)^{-n} }\,,
\end{equation*}
where $\Vet = \partial \Ddom \setminus (1+\epsilon)\cdot \partial \disk$ is the rectilinear part of the contour $\partial \Ddom$ (see Figure~\refp{b}{fig:V-path}). Then one plugs the asymptotic expansion of $F(x)$ when $x\to 1$ into the integral. One shows that any term that is analytic at $x=1$ in the expansion will have an exponentially small contribution, and terms of the order $(1-x)^\alpha$ and $O\m({ (1-x)^\alpha }$ have contributions of the order $n^{-(\alpha+1)}$ and $O\m({ n^{-(\alpha+1)} }$, respectively.

\begin{proof}[Proof of Theorem~\ref{thm:asympt} --- diagonal asymptotics]
By Proposition~\ref{prop:singularity structure}, the function $(x,y)\mapsto Z(u_cx,u_cy)$ is holomorphic in $\cslit \times \cDdom$, and hence we can express the coefficient $[x^py^q] Z(u_cx,u_cy)$ as a double contour integral and deform the contours of integral to the boundary of that domain. This gives
\begin{equation*}
u_c^{p+q} \cdot z_{p,q} =
\m({\frac{1}{2\pi i}}^2 \oiint_{\partial \slit \times \partial \Ddom} \frac{Z(u_cx,u_cy)}{x^{p+1}y^{q+1}} \dd x\dd y \,.
\end{equation*}

First, let us show that the contour integral can be restricted to a neighborhood of the dominant singularity $(x,y)=(1,1)$ with an exponentially small error.
Let $\Ve = \partial \slit \setminus (1+\epsilon)\cdot \partial \disk$ be the rectilinear portion of the contour $\partial \slit$.
It consists of two oriented line segments living in the Riemann sphere with a branch cut along $(1,\infty)$.
Similarly, define $\Vet = \partial \Ddom \setminus (1+\epsilon)\cdot \partial \disk$. The two paths $\Ve$ and $\Vet$ are depicted in Figure~\refp{a--b}{fig:V-path}.
For all $(x,y)\in \partial \slit \times \partial \Ddom$, we have $|x|\ge 1$ and $|y|\ge 1$. Moreover, if $(x,y)\notin \Ve \times \Vet$, then either $|x|=1+\epsilon$ or $|y|=1+\epsilon$. Since $Z(u_cx,u_cy)$ is continuous on $\partial \slit \times \partial \Ddom$, it follows that
\begin{align*}
&\ \abs{
\m({\frac{1}{2\pi i}}^2
\iint_{(\partial \slit \times \partial \Ddom) \setminus (\Ve \times \Vet)}
\frac{Z(u_cx,u_cy)}{x^{p+1}y^{q+1}} \dd x\dd y
}
\\ \le &\
\m({\frac{1}{2\pi}}^2
\sup_{(x,y)\in \partial \slit \times \partial \Ddom}
\hspace{-8mm} \abs{Z(u_cx,u_cy)} \hspace{4mm}
\cdot (1+\epsilon)^{-\min(p,q)}
\ = \ O\m({ (1+\epsilon)^{-\lambda_{\min} p} }
\end{align*}
where we assume \wlg\ $\lambda_{\min}\le 1$, so that $\min(p,q)\ge  \lambda_{\min} p$ whenever $q/p \in [\lambda_{\min},\lambda_{\max}]$.
Thus we can forget about the integral outside $\Ve \times \Vet$ with an exponentially small error in the diagonal limit.

Using the expansion \eqref{eq:asym Z diag} in Proposition~\ref{prop:local expansion}, we can decompose the integral on $\Ve \times \Vet$ as
\begin{equation*}
\m({\frac{1}{2\pi i}}^2
\iint_{\Ve \times \Vet}
\frac{Z(u_cx,u_cy)}{x^{p+1}y^{q+1}} \dd x\dd y
\ = \ I\1{reg} + b\cdot I\1{hom} + I\1{rem}
\end{equation*}
where $I\1{reg}$, $I\1{hom}$ and $I\1{rem}$ are defined by replacing $Z(u_cx,u_cy)$ in the integral on the \lhs\ by $Z\02\1{reg}(x,y)$, $Z\1{hom}(1-x,1-y)$ and $O(\max(|1-x|,|1-y|)^{\alpha_2+\delta})$ respectively.

As mentioned in Remark~\ref{rem:local expansion}, $Z\02\1{reg}(x,y)$ is a linear combination of terms of the form $F(x)G(y)$ or $G(x)F(y)$, where $F$ is analytic in a neighborhood of $1$, and $G$ is integrable on $\Ve$ and $\Vet$ for $\epsilon$ small enough.
Consider the component of $I\1{reg}$ corresponding to one such term: the integral factorizes as
\begin{equation}\label{eq:factorized iint}
\iint_{\Ve \times \Vet}
\frac{F(x)G(y)}{x^{p+1}y^{q+1}} \dd x\dd y
\ =\
\m({ \int_{\Ve } \frac{F(x)}{x^{p+1}} \dd x } \cdot
\m({ \int_{\Vet} \frac{G(y)}{y^{q+1}} \dd y }\,.
\end{equation}
Since $F$ is analytic in a neighborhood of $1$, we can deform the contour of integration $\Ve$ in the first factor away from $x=1$, so that it stays away from a disk of radius $r>1$ centered at the origin. It follows that the integral is bounded as an $O(r^{-p})$.
On the other hand, the second integral is bounded by a constant $\int_{\Vet}|G(y)|\dd |y|<\infty$ thanks to the integrability of $G$ on $\Vet$.
Hence the \lhs\ of \eqref{eq:factorized iint} is also an $O(r^{-p})$. Since $I\1{reg}$ is a linear combination of terms of this form, we conclude that there exists $r_*>1$ such that $I\1{reg} = O(r_*^{-p})$ when $p,q\to \infty$ and $\frac qp\in [\lambda_{\min},\lambda_{\max}]$.

\begin{figure}
\centering
\includegraphics[scale=1]{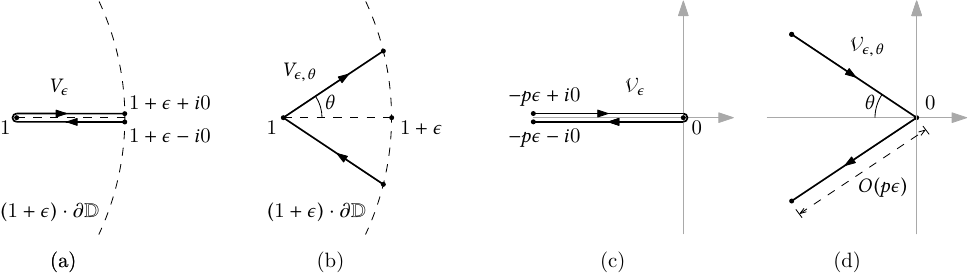}
\caption{The paths of integration $\Ve$, $\Vet$ and $\mc \Ve$, $\mc \Vet$.}
\label{fig:V-path}
\end{figure}

Next, let us prove that $I\1{rem} = O(p^{-\alpha_2+2+\delta})$ in the same limit. Consider the change of variables $s=p(1-x)$ and $t=p(1-y)$, and denote by $\mc \Ve$ and $\mc \Vet$ respectively the images of $\Ve$ and $\Vet$ under this change of variable, as in Figure~\refp{c--d}{fig:V-path}. (Notice that these paths now depend on $p$.)
Then $I\1{rem}$ can be written as
\begin{equation*}
I\1{rem} = \m({ \frac{1}{2\pi i} }^2 \iint_{\mc \Ve \times \mc \Vet} \frac{ O\m({ \max(p^{-1}|s|,p^{-1}|t|)^{\alpha_2+\delta} } }{ (1-p^{-1}s)^{p+1} (1-p^{-1}t)^{q+1} }  \frac{\dd s\dd t}{p^2} \,.
\end{equation*}
On the one hand, there exists a constant $C_1$ such that $\abs{ O\m({ \max(p^{-1}|s|,p^{-1}|t|)^{\alpha_2+\delta} } } \le C_1 \cdot p^{-(\alpha_2+\delta)} \cdot \m({ |s| + |t|}^{\alpha_2+\delta}$ for all $(s,t)\in \mc \Ve \times \mc \Vet$. On the other hand, since $|s|\le p\epsilon$ and $|t\cos \theta|\le p\epsilon$ for all $(s,t)\in \mc \Ve \times \mc \Vet$, and $q\ge \lambda_{\min} p$, it is not hard to see that
\begin{equation*}
\abs{ \m({1-\frac sp}^{p+1} }
= \m({ 1+\frac{|s|}{p} }^{p+1}
\ge e^{C_2 |s|}
\qtq{and}
\abs{ \m({1-\frac tp}^{q+1} }
\ge \m({ 1+\frac{|t|\cos \theta}{p} }^{q+1} \ge e^{C_2 |t|}
\end{equation*}
for some constant $C_2>0$ that only depends on $\epsilon$, $\theta$ and $\lambda_{\min}$. It follows that
\begin{equation*}
\abs{I\1{rem}} \le \frac{C_1}{4 \pi^2} \cdot p^{-(\alpha_2+2+\delta)} \iint_{\mc \Ve \times \mc \Vet} \m({|s|+|t|}^{\alpha_2+\delta} e^{-C_2(|s|+|t|)} \dd|s|\,\dd|t| \,.
\end{equation*}
The integral on the \rhs\ is bounded by the constant $4\int_0^\infty \dd r_1 \int_0^\infty \dd r_2 \cdot e^{-C_2(r_1+r_2)}\cdot \m({r_1+r_2}^{\alpha_2+\delta} <\infty$. Hence $I\1{rem} = O\m({ p^{-(\alpha_2+2+\delta)} }$.

To estimate the term $I\1{hom}$, we make the same change of variables as for $I\1{rem}$. Since $Z\1{hom}$ is homogeneous of order $\alpha_2$, we have
\begin{equation*}
I\1{hom} = \m({ \frac{1}{2\pi i} }^2 \iint_{\mc \Ve \times \mc \Vet} \frac{ p^{-\alpha_2} Z\1{hom}(s,t) }{ (1-p^{-1}s)^{p+1} (1-p^{-1}t)^{q+1} }  \frac{\dd s\dd t}{p^2} \,
\end{equation*}
Using again the fact that $|s|\le p\epsilon$ and $|t|\cos \theta \le p\epsilon$ for $(s,t)\in \mc \Ve \times \mc \Vet$,
we can expand in the denominator in the integral as $\m({1-p^{-1}s}^{-(p+1)} \m({1-p^{-1}t}^{-(q+1)} \!= \exp \m({s+\frac qpt}  \cdot \m({1+O \m({ \max(|p^{-1}s^2|,|p^{-1}t^2| }}$. More precisely, the big-O means that there exists a constant $C_3$ depending only on $\epsilon$, $\theta$ and $\lambda_{\min}$ such that for all $(s,t)\in \mc \Ve \times \mc \Vet$
\begin{equation*}
\abs{ \m({1-\frac sp}^{-(p+1)} \m({1-\frac tp}^{-(q+1)} - e^{s+\frac qpt} } \le C_3 \frac{(|s|+|t|)^2}p \cdot \abs{e^{s+\frac qpt}} \,.
\end{equation*}
Moreover, there exists $C_4>0$ such that $|e^{s+\frac qpt}| \le e^{-C_4(|s|+|t|)}$ for all $(s,t)\in \mc \Ve \times \mc \Vet$.
It follows that
\begin{align*}
&\ \abs{
    I\1{hom} - \m({ \frac{1}{2\pi i} }^2 \iint_{\mc \Ve \times \mc \Vet}
    p^{-\alpha_2}Z\1{hom} (s,t) \cdot e^{s+\frac qpt} \cdot \frac{\dd s\dd t}{p^2}
}
\\ \le &\
\frac{1}{4\pi^2} \iint_{\mc \Ve \times \mc \Vet} p^{-\alpha_2} |Z\1{hom}(s,t)|\cdot C_3 \frac{(|s|+|t|)^2}p e^{-C_4(|s|+|t|)} \cdot \frac{\dd |s|\, \dd |t|}{p^2}
\\ \le &\
C_5 \cdot p^{-(\alpha_2+3)} \iint_{\mc \Ve \times \mc \Vet} (|s|+|t|)^{\alpha_2+2} e^{-C_4(|s|+|t|)} \dd |s|\, \dd |t| \,.
\end{align*}
where for the last line we used the bound $\abs{Z\1{hom}(s,t)}\le \mathtt{cst}\cdot (|s|+|t|)^{\alpha_2}$, which is a consequence of the fact that $Z\1{hom}(s,t)$ is homogeneous of order $\alpha_2$ and continuous on $\Ve \times \Vet$.
The integral on the last line is bounded by $4 \int_0^\infty \dd r_1 \int_0^\infty \dd r_2 \cdot e^{-C_4(r_1+r_2)}\cdot \m({r_1+r_2}^{\alpha_2+2} <\infty$. Thus we have
\begin{equation*}
I\1{hom} = p^{-(\alpha_2+2)} \cdot \m({ \frac{1}{2\pi i} }^2
    \iint_{\mc \Ve \times \mc \Vet} Z\1{hom} (s,t) e^{s+\frac qpt} \dd s\dd t
    + O\m({ p^{-(\alpha_2+3)} } \,.
\end{equation*}
Let $\mc V_\infty$ and $\mc V_{\infty,\theta}$ be obtained by extending the line segments in $\mc \Ve$ and $\mc \Vet$ to rays joining the origin to the infinity. Thanks to the exponentially decaying factor $e^{s+\frac qpt}$ in the above integral, one can replace the domain $\mc \Ve \times \mc \Vet$ of the integral by $\mc V_\infty \times \mc V_{\infty,\theta}$ while committing an error that is exponentially small in $p$ (recall that $\mc \Ve$ and $\mc \Vet$ depend on $p$). Therefore $I\1{hom} = \tilde c(q/p)\cdot p^{-(\alpha_2+2)} + O \m({ p^{-(\alpha_2+3)} }$, where
\begin{equation*}
\tilde c(\lambda) := \m({ \frac{1}{2\pi i} }^2
\iint_{\mc V_\infty \times \mc V_{\infty,\theta}} Z\1{hom} (s,t) e^{s+\lambda t} \dd s\dd t \,.
\end{equation*}
With the previous estimates for $I\1{reg}$ and $I\1{rem}$, we get
\begin{equation*}
u_c^{p+q} z_{p,q} = b\cdot \tilde c(q/p)\cdot p^{-(\alpha_2+2)} + O \m({ p^{-(\alpha_2+2+\delta)} } \,.
\end{equation*}
where the big-O estimate is uniform for all $p,q\to \infty$ such that $q/p\in [\lambda_{\min},\lambda_{\max}]$.

We finish the proof by computing $\tilde c(\lambda)$ or, in the notation of Theorem~\ref{thm:asympt}, $c(\lambda) = \Gamma(-\alpha_0)\Gamma(-\alpha_1) \cdot \tilde c(\lambda)$. Notice that the value of $\tilde c(\lambda)$ does not depend on the angle $\theta$ appearing in the contour of integration $\mc V_{\infty,\theta}$.

\paragraph{Low temperatures.}
When $\nu > \nu_c$, we have $Z\1{hom}(s,t)=s^{3/2}t^{3/2}$ and Proposition~\ref{prop:singularity structure}~and~\ref{prop:local expansion} allow us take $\theta=0$. Then the double integral defining $\tilde c(\lambda)$ factorizes as
\begin{equation*}
\tilde c(\lambda) =
      \m({ \frac{1}{2\pi i} \int_{\mc V_\infty} s^{3/2} e^s \dd s }
\cdot \m({ \frac{1}{2\pi i} \int_{\mc V_\infty} t^{3/2} e^{\lambda t} \dd t } \,.
\end{equation*}
After the change of variable $t'=\lambda t$ in the second factor, the formula simplifies to $\tilde c(\lambda) \!=\!
\m({ \frac{1}{2\pi i} \int_{\mc V_\infty}\! s^{3/2} e^s \dd s }^2 \lambda^{-5/2}$.
Since the contour $\mc V_\infty$ lives in the Riemann sphere with a branch cut along $(-\infty,0)$, the function $s^{3/2}$ should be understood as its principal branch \wrt\ this branch cut. Therefore
\begin{align*}
\frac{1}{2\pi i} \int_{\mc V_\infty} s^{3/2} e^s \dd s &\ =
\frac{1}{2\pi i} \int_0^\infty \m({ (-r+i0)^{3/2} - (-r-i0)^{3/2} } e^{-r} \dd r
\\ &\ =
\frac{1}{2\pi i} \int_0^\infty \m({ -ir^{3/2} - ir^{3/2} } e^{-r} \dd r
\ =\ -\frac{\Gamma(5/2)}{\pi} \,.
\end{align*}
Recall that in the low temperature regime, $\alpha_0=\alpha_1=3/2$, and by Euler's reflection formula, $\Gamma(5/2)\Gamma(-3/2)=\pi$. It follows that $c(\lambda) = \Gamma(-3/2)^2\cdot \tilde c(\lambda) = \lambda^{-5/2}$.

\paragraph{High temperatures.}
When $\nu\in (1,\nu_c)$, we have $Z\1{hom}(s,t) = \frac{s^{1/2} t^{1/2}}{s^{1/2} + t^{1/2}}$ and thus
\begin{equation*}
\tilde c(\lambda) =
\m({ \frac{1}{2\pi i} }^2 \int_{\mc V_{\infty,\theta}} \m({
    \int_{\mc V_\infty} \frac{s^{1/2}}{s^{1/2} + t^{1/2}} e^s \dd s
} t^{1/2} e^{\lambda t} \dd t \,.
\end{equation*}
The inner integral can be expanded in a similar way as in the low temperature case
\begin{align*}
\int_{\mc V_\infty} \frac{s^{1/2}}{s^{1/2} + t^{1/2}} e^s \dd s
& =
\int_0^\infty \m({ \frac{(-r+i0)^{1/2}}{(-r+i0)^{1/2} + t^{1/2}}
                  -\frac{(-r-i0)^{1/2}}{(-r-i0)^{1/2} + t^{1/2}} } e^{-r} \dd r
\\ & =
\int_0^\infty \m({ \frac{ i r^{1/2}}{ i r^{1/2} + t^{1/2}}
                  -\frac{-i r^{1/2}}{-i r^{1/2} + t^{1/2}} } e^{-r} \dd r
= \int_0^\infty \frac{2i r^{1/2} t^{1/2}}{r+t} e^{-r} \dd r \,.
\end{align*}
Plugging the \rhs\ into the expression of $\tilde c(\lambda)$ and changing the order of the integrals on $r$ and on $t$ yield
\begin{equation*}
\tilde c(\lambda) =
\frac1\pi \int_0^\infty \m({
    \frac{1}{2\pi i} \int_{\mc V_{\infty,\theta}} \frac{t \cdot e^{\lambda t}}{r+t} \dd t
} r^{1/2} e^{-r} \dd r \,.
\end{equation*}
The function $t\mapsto \frac{t\cdot e^{\lambda t}}{r+t}$ is meromorphic on $\complex$ and has a unique (simple) pole at $t=-r$, with a residue of $-r \cdot e^{-\lambda r}$. By closing the contour $\mc V_{\infty,\theta}$ far from the origin in the direction of the negative real axis, we see that the integral on $t$ is given by $-1$ times the residue. Therefore
\begin{equation*}
\tilde c(\lambda) = \frac{1}{\pi} \int_0^\infty r^{3/2} e^{-(1+\lambda)r} \dd r = \frac{\Gamma(5/2)}{\pi} (1+\lambda)^{-5/2} \,.
\end{equation*}
In the high temperature regime, we have $\alpha_0=3/2$ and $\alpha_1=-1$. Thus  $c(\lambda) = \Gamma(-3/2)\Gamma(1) \cdot \tilde c(\lambda) = (1+\lambda)^{-5/2}$.

\paragraph{Critical temperature.}
When $\nu=\nu_c$, we have $Z\1{hom}(s,t) = \frac{-st}{s^{1/3} + t^{1/3}}$ and one can take $\theta = 0$. In the low and high temperature regimes, we have used the relation $\int_{\mc V_\infty} f(x)\dd x = \int_0^\infty \m({f(-r+i0)-f(-r-i0)} \dd r$ to expand integrals on $\mc V_\infty$. By applying this relation to the integral on $s$ and the integral on $t$ simultaneously, we get
\begin{align*}
\tilde c(\lambda) &= \m({ \frac{1}{2\pi i} }^2 \iint_{\mc V_\infty \times \mc V_\infty} \frac{-st}{s^{1/3} + t^{1/3}} e^{s+\lambda t} \dd s \dd t
\\&=
\m({ \frac{1}{2\pi i} }^2 \iint_{(0,\infty)^2}
\m({ \sum_{(\sigma_1,\sigma_2)\in \{-1,+1\}^2}
     \frac{\sigma_1\sigma_2}{(-r_1+\sigma_1 \cdot i0)^{1/3} + (-r_2+\sigma_2 \cdot i0)^{1/3}}}
\cdot (-r_1r_2) \cdot e^{-(r_1+\lambda r_2)} \dd r_1 \dd r_2 \,.
\end{align*}
The principal branch of the function $s^{1/3}$ prescribes that $(-r\pm i0)^{1/3} = r^{1/3} e^{\pm i\frac\pi3}$.
One can check by direct computation that
\begin{equation*}
\sum_{(\sigma_1,\sigma_2)\in \{-1,+1\}^2}
     \frac{\sigma_1\sigma_2}{(-r_1+\sigma_1 \cdot i0)^{1/3}
                           + (-r_2+\sigma_2 \cdot i0)^{1/3}}
= \frac{-3 r_1^{1/3} r_2^{1/3}}{r_1+r_2} \,.
\end{equation*}
Therefore
\begin{equation*}
\tilde c(\lambda) = \m({ \frac{1}{2\pi i} }^2 \cdot 3 \iint_{(0,\infty)^2}
\frac{r_1^{4/3} r_2^{4/3} e^{-(r_1+\lambda r_2)}}{r_1+r_2} \dd r_1 \dd r_2 \,.
\end{equation*}
One can ``factorize'' this double integral using the relation $\frac{1}{r_1+r_2} = \int_0^\infty e^{-r_1 r} e^{-r_2 r} \dd r$\,:
\begin{align*}
\tilde c(\lambda) &= -\frac{3}{4\pi^2} \int_0^\infty
\m({ \int_0^\infty r_1^{4/3} e^{-(1+r)r_1} \dd r_1 }
\m({ \int_0^\infty r_2^{4/3} e^{-(\lambda+r)r_2} \dd r_2 }
\dd r
\\ &= -\frac{3}{4\pi^2} \int_0^\infty
      \m({ \Gamma(7/3) \cdot (1+r)^{-7/3} }
      \m({ \Gamma(7/3) \cdot (\lambda+r)^{-7/3} } \dd r
\\& = -\m({ \frac{\sqrt 3}{2\pi} \Gamma(7/3) }^2
      \int_0^\infty (1+r)^{-7/3} (\lambda+r)^{-7/3} \dd r \,.
\end{align*}
When $\nu=\nu_c$, we have $\alpha_0=4/3$ and $\alpha_1=1/3$. And by Euler's reflection formula, $\Gamma(7/3)\Gamma(-4/3) = \frac{\pi}{\sin(7\pi/3)} = \frac{2\pi}{\sqrt 3}$. It follows that
\begin{equation*}
c(\lambda) = \Gamma(-4/3) \Gamma(-1/3) \cdot \tilde c(\lambda)
= \frac{4}{3} \int_0^\infty (1+r)^{-7/3} (\lambda+r)^{-7/3} \dd r \,.
\qedhere
\end{equation*}
\end{proof}

\section{Peeling processes and perimeter processes}\label{sec:peeling perim}

We recall first the essentials of the \emph{peeling process} for Ising-triangulations with spins on faces, introduced in \cite{CT20}. The peeling process is the central object both in the construction of the local limits and in the proofs of the local convergences. It can be viewed as a deterministic exploration of a fixed map, driven by a \emph{peeling algorithm} $\mathcal{A}$. The basic definition of the process is identical to that of \cite{CT20}, with the exception that in this work, the peeling algorithm $\mathcal{A}$ is defined in a slightly different way. In particular, the algorithm chooses an edge on the explored boundary instead of a vertex, and after revealing the internal face incident to that edge, decides where to continue the peeling. Here, the peeling process can be seen as a decorated map version of the (filled-in) \emph{simple peeling} of undecorated maps \cite{BC21}, where the peeling algorithm first chooses a boundary edge, reveals a face adjacent to it, and finally decides the new unexplored part of the map (in the case when the unexplored part is disconnected by the revealed face). While the algorithm used in \cite{CT20} still works in the low temperature regime, we will need different algorithms in the high temperature regime, as explained in Section~\ref{sec:hightemplimit}. We will also note that, unlike in \cite{CT20}, the different peeling algorithms result different laws of the peeling process.

Throughout this work we assume the following: if an Ising-triangulation has a bicolored boundary, the algorithm $\mathcal{A}$ chooses an edge at the junction of the $\+$ and $\<$ boundary segments on the boundary of the explored map. This edge may either have spin $\+$ or $\<$. It is easy to see that deletion of the chosen edge and exposure of the adjacent face preserves the Dobrushin boundary condition of the map, while another type of a peeling algorithm may complicate the boundary condition. Thus, we call such an algorithm $\mathcal{A}$ \emph{Dobrushin-stable}. We make the following convention: if the algorithm always chooses a \< edge to peel, we denote it by $\algo_\<$; otherwise if it always chooses a \+ edge, we denote it by $\algo_\+$; otherwise, the algorithm is ``mixed'', choosing either type of the edges, and denoted by $\algo_m$.

The choice of the peeling algorithm in each of the temperature regime stems from the different expected interface geometries in the respective regimes. At $\nu=\nu_c$, we already saw in \cite{CT20} that the peeling algorithm $\algo_\<$ is particularly well-suited, due to the fact that we take the limit $q\to\infty$ first, after which there is an infinite \< boundary. For $\nu>\nu_c$, we can still make the same choice. However, for $\nu\in (1,\nu_c)$, we will notice that whether we choose the peeling to explore the left-most or the right-most interface from the root $\rho$, the interface will stay close to the boundary of the half-plane. Thus, in order to explore the local limit by roughly distance layers, we need to combine two different explorations. This leads us to choose a mixed algorithm $\algo_m$. In the first limit $q\to\infty$, however, the simplest choice which works is $\algo_\+$.

\newcommand*{\hl}{\\\cline{1-6}}
\tabulinesep=1.2mm
\newcolumntype{L}{>{\!$\displaystyle}l<{$\!}}
\newcolumntype{S}{>{\!$\displaystyle}l<{$\!\!}}
\newcolumntype{R}{>{\!$\displaystyle}r<{$\!}}
\newcolumntype{C}{>{\!$\displaystyle}c<{$\!}}

\newcommand*{\zp}[1]{t\zz{#1}}
\newcommand*{\zzp}[2]{t\zzz{#1}{#2}}

\newcommand*{\zph}[1]{t\zzh{#1}}
\newcommand*{\zzph}[2]{t\zzzh{#1}{#2}}

\newcommand*{\zq}[1]{t\zz[p,q+1]{#1}}
\newcommand*{\zqh}[1]{t\zz[p+1,q]{#1}}
\newcommand*{\zqo}[1]{t\zz[0,q+1]{#1}}

\newcommand*{\zzq}[2]{t\zzz[p,q+1]{#1}{#2}}
\newcommand*{\zzqh}[2]{t\zzz[p+1,q]{#1}{#2}}
\newcommand*{\zzqo}[2]{t\zzz[0,q+1]{#1}{#2}}

\newcommand*{\ap}[1]{\,\frac{a_{#1}}{a_{p}}\,}
\newcommand*{\aph}[1]{\,\frac{a_{#1}}{a_{p+1}}\,}
\newcommand*{\apo}[1]{\,\frac{a_{#1}}{a_0}}

\newcommand{\pq}[1][P_n,Q_n]{\raisebox{-2pt}{$\!_{#1}\!$}}

\begin{table}[b!]
\begin{center}
\begin{tabu}{|L|L| |L|}
\hline \text{Local convergence}	& \nu\in(1,\nu_c)   & \nu\in[\nu_c,\infty)
\\\hline \prob_{p,q}^\nu\cv{q}\prob_p^\nu	& \algo_\+			& \algo_\<
\\\hline \prob_p^\nu\cv{p}\prob_\infty^\nu	& \algo_m		& \algo_\<
\\\hline \prob_{p,q}^\nu\cv{p,q}\prob_\infty^\nu	& \algo_m^\dagger	&  \algo_\<^\dagger
\\\hline
\end{tabu}
\caption{A summary of the choices of the peeling algorithm in each phase for proving the local convergences. The peeling algorithms $\algo_\+$ and $\algo_\<$ are defined in Section~\ref{sec:peeling}, while $\algo_m$ is defined in Section~\ref{sec:hightemplimit}. The notation $\algo^\dagger$ refers to the variant of the peeling algorithm $\algo$ which targets the vertex $\rho^\dagger$. See the end of Section~\ref{sec:peeling} for definition.}\label{tab:peelingalgo}%
\end{center}\vspace{-0.7em}
\end{table}

\begin{figure}
\centering
\includegraphics[scale=1]{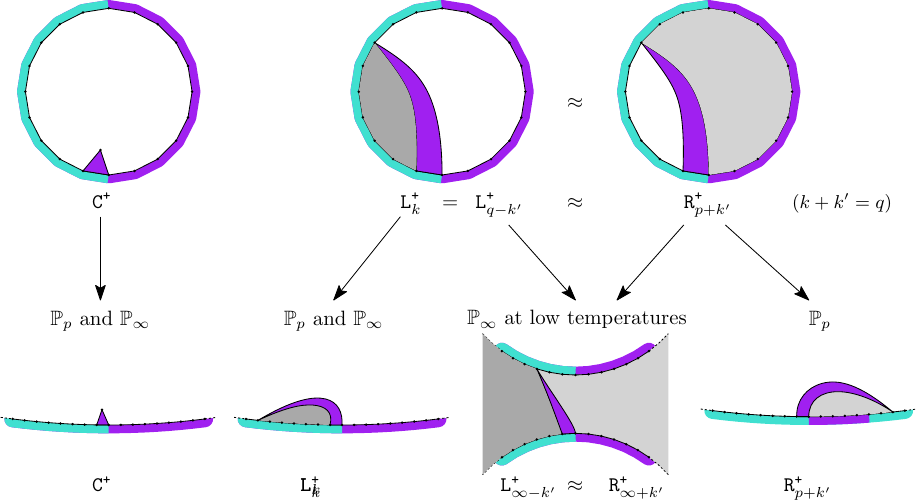}
\caption{Illustration of the peeling events. Only peeling events revealing a \+ face are included.\\ 
Top: peeling events in a finite triangulation with Dobrushin boundary condition. 
The $\approx$ sign indicates different peeling events which only differ by the choice of the unexplored component.\\
Bottom: peeling events in a typical infinite triangulation sampled from the laws $\prob_p$ or $\prob_\infty$. Each arrow indicates that the lower picture can be obtained as a local limit of the upper picture either when $q\to \infty$ or when both $p,q\to \infty$.}\label{fig:Tutte-eq}
\end{figure}

When we take the local limits $q\to\infty$ and $p\to\infty$ one-by-one, we always peel from a boundary with  more \< edges. This is to ensure the peeling process is compatible with the $q=\infty$ case. In the limit $(p,q)\to\infty$, it is more natural to peel from the boundary which contains the vertex $\rho^\dagger$ opposite to the root and at the junction of the \+ and \< boundaries, which can be seen as a point in the infinity. For this purpose, we introduce the \emph{target} $\rho^\dagger$ for the peeling. See the following subsection for a more precise definition. A summary of the peeling algorithms and the existence of the target is presented in Table~\ref{tab:peelingalgo}.

In the following subsection, we define the versions of the peeling process used in this work in the finite setting. After that, we generalize those for infinite Ising triangulations of the half-plane, and study the properties of the associated perimeter processes.

\subsection{Peeling of finite triangulations}\label{sec:peeling}

\paragraph{Peeling along the left-most interface: peeling algorithm $\algo_\<$.} Assume that an Ising-triangulation $\bt$ has at least one boundary edge with spin \<. In this case, the peeling algorithm $\algo_\<$ chooses the edge $e$ with spin \< immediately on the left to the origin. We remove $e$ and reveal the internal face $f$ adjacent to it. If $f$ does not exist, then $\tmap$ is the edge map and $\bt$ has a weight 1 or $\nu$. If $f$ exists, let $*\in\{\+,\<\}$ be the spin on $f$ and $v$ be the vertex at the corner of $f$ not adjacent to $e$. Then the possible positions of $v$ are:
\begin{description}[noitemsep]
\item[Event $\CC^*$:]	$v$ is not on the boundary of $\tmap$;
\item[Event $\RR^*_k$:] $v$ is at a distance $k$ to the right of $e$ on the boundary of $\tmap$; ($0\le k\le p$);
\item[Event $\LL^*_k$:] $v$ is at a distance $k$ to the left of $e$ on the boundary of $\tmap$. ($0\le k< q$).
\end{description} If $p,q<\infty$, we also make the identification $\RR^*_{p\mp k}=\LL^*_{q\pm k}$, which is useful in the sequel. We define $\tilde\steps := \{\cp,\cm\}\cup\{\lp,\lm,\rp,\rn:\ k\ge 0\}$ as the set of peeling events. See Figure \ref{fig:Tutte-eq} for graphical illustration of the peeling events.

The peeling process along the left-most interface $\iroot$ is constructed by iterating the face-revealing operation described above, yielding an increasing sequence $\nseq \emap$ of \emph{explored maps}. In order to iterate the peeling, we apply a rule that chooses one of the two unexplored regions, when the peeling step of type $\RR^*_k$ or $\LL^*_k$ separates the unexplored map into two pieces. Here, we assume that the boundary contains no target vertex which determines the unexplored part (this case is treated separately later). In the case of peeling along the left-most interface, the algorithm $\algo_\<$ chooses the unexplored region with greater number of \< boundary edges (and in case of a tie, the region on the right is chosen). This in particular guarantees that when $q=\infty$ and $p<\infty$, we will choose the unbounded region as the next unexplored map. See Figure \ref{fig:def-peeling} for illustration.

\begin{figure}[t!]
\centering
\includegraphics[scale=1.25]{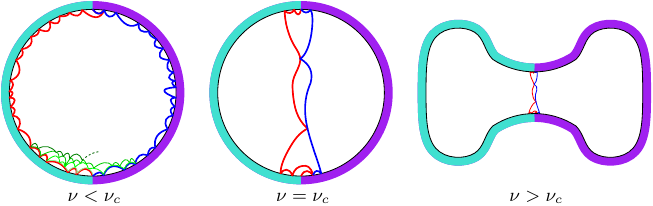}
\caption{Illustration of the interfaces explored by the different versions of the peeling process.\\
Red: left-most interface (explored by $\algo_\<$). Blue: right-most interface (explored by $\algo_\+$). Shades of green: interface explored by $\algo_m$.
}\label{fig:def-peeling}
\end{figure}

We can use the sequence of explored maps $\nseq\emap$ as the definition of the peeling process as follows. At each time $n$, the explored map $\emap_n$ consists of a subset of faces of $\bt$ containing at least the external face and separated from its complementary set by a simple closed path. We view $\emap_n$ as a bicolored triangulation of a polygon with a special uncolored, not necessarily triangular, internal face called the \emph{hole}. It inherits its root and its boundary condition from $\bt$. The complement of $\emap_n$ is called the \emph{unexplored map at time $n$} and denoted by $\umap_n$. It is a bicolored triangulation of a polygon.
Notice that $\umap_n$ may be the edge map, in which case $\emap_n$ is simply $\bt$ in which an edge is replaced by an uncolored digon. This may, however, only happen at the last step of the peeling process.

\begin{table}
\centering
\begin{tabu}{|L|S|C|  |L|S|C| L}
\cline{1-6}
\step &\Prob_{p,q+1}(\Step_1 = \step) &(X_1,Y_1) &
\step &\Prob_{p,q+1}(\Step_1 = \step) &(X_1,Y_1)
\hl \cp	& 		\zp{p+2,q}				&(2,-1)
&	\cm	& \nu \zp{p,q+2} 				&(0,1)
\hl \lp	& 		\zzp{p+1,q-k}{1,k}		&(1,-k-1)
&	\lm	& \nu 	\zzp{p,q-k+1}{0,k+1}		&(0,-k)	&	(0\le k \le \frac{q}{2})
\hl	\rp	& 		\zzp{k+1,0}{p-k+1,q}		&(-k+1,-1)
&	\rn & \nu 	\zzp{k,1}{p-k,q+1}		&(-k,0)	&	(0\le k\le p)
\hl \rp[p+k]	& 		\zzp{p+1,k}{1,q-k}	&(-p+1,-k-1)
&	\rn[p+k]	& \nu 	\zzp{p,k+1}{0,q-k+1}	&(-p,-k)
&(0<k<\frac{q}2)
\hl
\end{tabu}
\caption{Law of the first peeling event $\Step_1$ under $\Prob_{p,q+1}$ and the corresponding $(X_1,Y_1)$, where the peeling is without target. We use the shorthand notations $t=t_c(\nu)$ and $z_{p,q}=z_{p,q}(t,\nu)$.
}\label{tab:prob(p,q)}
\end{table}

We apply this rule recursively starting from $\umap_0=\bt$. At each step, the construction depends on the boundary condition of $\umap_n$:
\begin{enumerate}
\item	If $\umap_n$ has a bichromatic Dobrushin boundary, let $\rho_n$ be the boundary junction vertex of $\umap_n$ with a \< on its left and a \+ on its right ($\rho_0=\rho$). Then $\umap_{n+1}$ is obtained by revealing the internal face of $\umap_n$ adjacent to the boundary edge on the left of $\rho_n$ and, if necessary, the algorithm chooses one of the two unexplored regions according to the \mbox{rule described above}.
\item	If $\umap_n$ has a monochromatic boundary condition of spin \<, then the peeling algorithm $\algo_\<$ chooses the boundary edge with the vertex $\rho_n$ as an endpoint according to some deterministic function of the explored map $\emap_n$, which we specify later in Sections~\ref{sec:lowtemplimit}~and~\ref{sec:hightemplimit}. We then construct $\umap_{n+1}$ from $\umap_n$ and $\rho_n$ in the same way as in the previous case.
\item	If $\umap_n$ has a monochromatic boundary condition of spin \+ or has no internal face, then we set $\emap_{n+1}=\bt$ and terminate the peeling process at time $n+1$.
\end{enumerate}

We denote the law of this peeling process by $\Prob_{p,q}^\nu\equiv\Prob_{p,q}$, where on the right we have dropped the dependence of $\nu$ in order to ease the notation and continue to do so in the sequel (with the exception of Section \ref{sec:orderparam}, where the index $\nu$ will be included for clarity). Let $(P_n,Q_n)$ be the boundary condition of $\umap_n$, and $(X_n,Y_n)=(P_n-P_0,Q_n-Q_0)$. Also, let $\Step_n\in\tilde\steps$ denote the peeling event that occurred when constructing $\umap_n$ from $\umap_{n-1}$. Then the \emph{peeling process} following the left-most interface can also be defined as the random process $\nseq{\Step}$ on $\tilde\steps$, with the law $\Prob_{p,q}$. We view the above quantities as random variables defined on the sample space $\Omega=\bts=\bigcup_{p,q} \bts_{p,q}$. In the sequel, one should understand that any of the sequences $\nseq{\emap}$, $\nseq{\umap}$ and $\nseq{\Step}$ can be viewed as the peeling process, since together with the boundary condition, they contain the same essential information. Table~\ref{tab:prob(p,q)} collects the distribution of the first peeling step $\Step_1$ and the associated perimeter change in the peeling process driven by $\algo_\<$.

\begin{table}[b!]
\begin{center}
\begin{tabu}{|L|S|C|  |L|S|C| L}
\cline{1-6}
\step &\Probh_{p+1,q}(\Step_1=\step) &(X_1,Y_1) &
\step &\Probh_{p+1,q}(\Step_1=\step) & (X_1,Y_1)
\hl \cp	& \nu \zph{p+2,q}				& (1,0)
&	\cm	&		\zph{p,q+2} 				& (-1,2)
\hl \lp	& \nu \zzph{p+1,q-k}{1,k}		& (0,-k)
&	\lm	&		\zzph{p,q-k+1}{0,k+1}		& (-1,-k+1) &(0\le k\le \frac{q}{2})
\hl	\rp	& \nu \zzph{k+1,0}{p-k+1,q}		& (-k,0)
&	\rn &		\zzph{k,1}{p-k,q+1}		&	(-k-1,1)						&(0\le k\le p)
\hl \rp[p+k]	& \nu \zzph{p+1,k}{1,q-k}	& (-p,-k)
&	\rn[p+k]	&		\zzph{p,k+1}{0,q-k+1}	& (-p-1,-k+1)	&(0<k< \frac{q}{2})
\hl
\end{tabu}

\vspace{3ex}

\begin{tabu}{|L|L| |L|L|}
\hline \step	& \Probh_{0,q+1}(\Step_1=\step)& \step  & \Probh_{0,q+1}(\Step_1=\step)
\\\hline \cp	& \zqo{2,q}		&\cm		& \nu \zqo{0,q+2}
\\\hline \lp	& \zzqo{1,q-k}{1,k}	&\lm		& \nu \zzqo{0,q-k+1}{0,k+1}
\\\hline \rp	& \zzqo{1,k}{1,q-k}	&\rn		& \nu \zzqo{0,k+1}{0,q-k+1}
\\\hline
\end{tabu}
\caption{Law of the first peeling event $\Step_1$ under $\Probh_{p+1,q}$ and the corresponding $(X_1,Y_1)$, where the peeling is without target. Due to the possibility that there is no $\+$ edge on the boundary, we also present the step probabilities under the law $\Probh_{0,q+1}$. The notational conventions coincide with Table~\ref{tab:prob(p,q)}.}\label{tab:probh(p,q)}%
\end{center}\vspace{-0.7em}
\end{table}

\paragraph{Peeling along the right-most interface: peeling algorithm $\algo_\+$.} The peeling process along the right-most interface is similar to the previous one, except that the algorithm $\algo_\+$ chooses the \+ edge adjacent to $\rho_n$ if possible. Again, in case there are more than one holes, the algorithm fills in the one with less \< edges by an independent Boltzmann Ising-triangulation, and if the hole has a monochromatic \< boundary, the peeling continues on that according to some deterministic function. A small subtlety here is that the distribution of this peeling process differs from the previous one, such that the step distribution involves a spin-flip due to the deleted boundary edge of different spin. In particular, we also need to take into account that the peeling algorithm chooses a \< edge if the unexplored part has a monochromatic boundary. We denote the distribution of this peeling by $\Probh_{p,q}^\nu\equiv\Probh_{p,q}$. For the explicit probabilities of the first peeling step, see Table~\ref{tab:probh(p,q)}.

\paragraph{Peeling with the target $\rho^\dagger$.} Let $\algo$ be any Dobrushin-stable peeling algorithm (in the sense of the previous paragraphs). Considering the local limits when $p,q\to\infty$ simultaneously, it is convenient to define a \emph{peeling process with a target}, where the target is the vertex $\rho^\dagger$ at the junction of the \< and \+ boundaries opposite to $\rho$. The definition of this peeling process is as in the previous paragraphs, except when the peeling step separates the unexplored map into two pieces: in this case, the unexplored part corresponds to the one containing $\rho^\dagger$, and the other one is filled. If $\rho^\dagger$ is contained in both of the separated regions, the one with more \< edges is chosen for the unexplored part. We denote by $\algo^\dagger$ this targeted variant of the peeling algorithm $\algo$.

\subsection{Peeling of infinite triangulations}\label{seq:infinitepeeling} Obtaining the limits of the peeling process for a general temperature $\nu$ is just a straightforward generalization of the analysis in our previous work \cite{CT20}. Indeed, the asymptotics of Theorem~\ref{thm:asympt} give the limit according to the recipe given in \cite{CT20}. The first limit $q\to\infty$ yields exactly the same form for the peeling process, where the step probabilities only depend on $\nu$. Following the notation of \cite{CT20}, let  $\Prob^\nu_p(\Step_1=\step):=\lim_{q\to\infty}\Prob^\nu_{p,q}(\Step_1=\step)$ and $\Prob^\nu_\infty(\Step_1=\step):=\lim_{p\to\infty}\Prob^\nu_p(\Step_1=\step)$. We again make the shorthand conventions $\Prob^\nu_p\equiv\Prob_p$ and $\Prob^\nu_\infty\equiv\Prob_\infty$ which we continue to use in the sequel except in Section \ref{sec:orderparam}. The quantities after the first limit $q\to\infty$ are collected in Table~\ref{tab:prob(p)}.

\begin{table}[h]
\centering
\begin{tabu}[t]{|L|L|L| |L|L|L| L}
\cline{1-6}
\step	& \Prob_p(\Step_1=\step)				& (X_1,Y_1)	&
\step	& \Prob_p(\Step_1=\step)				& (X_1,Y_1)	&
\hl	\cp & 		t \ap{p+2} u					& (2,-1)	&
	\cm & \frac{\nu t}{u}					& (0,1)	&
\hl	\lp & 		t \ap{p+1} z_{1,k} u^{k+1}	& (1,-k-1)	&
	\lm & \nu 	t z_{0,k+1} u^k				& (0,-k)	& (k\ge 0)
\hl	\rp & 		t z_{k+1,0}\ap{p-k+1} u		& (-k+1,-1)	&
	\rn & \nu 	t z_{k,1} \ap{p-k} 			& (-k,0)	& (0\le k\le \frac{p}{2})
\hl \rp[p-k] & t z_{p-k+1,0} \ap{k+1} u     & (-p+k+1,-1) &
	\rn[p-k] & \nu t z_{p-k,1} \ap{k} & (-p+k,0) & (0\le k <\frac{p}{2})
\hl	\rp[p+k] & 		t z_{p+1,k} \ap1 u^{k+1}	& (-p+1,-k-1)	&
	\rn[p+k] & \nu	t z_{p,k+1} \ap0 u^k		& (-p,-k)	& (k>0)
\hl
\end{tabu}
\caption{Law of the first peeling event $\Step_1$ under $\Prob_{p}$ and the corresponding $(X_1,Y_1)$, where the peeling is without target. We use the shorthand notations $t=t_c(\nu)$, $u=u_c(\nu)$, $z_{p,k}=z_{p,k}(t,\nu)$ and $a_p=a_p(\nu)$. Note the cutoff $p/2$ in the finite boundary segment, which is used for the convergence $p\to\infty$ in the $\nu>\nu_c$ regime (see Table~\ref{tab:pinfty}).}\label{tab:prob(p)}
\end{table}

Taking the second limit $p\to\infty$ yields a similar peeling process for all $1<\nu\leq\nu_c$, but for $\nu>\nu_c$ the asymptotics of Theorem~\ref{thm:asympt} yield additional non-trivial peeling events. Indeed, since the perimeter exponents $\alpha_0$ and $\alpha_1$ of $z_{p,k}$ and $a_p$ coincide in that case, the probabilities $\Prob_p(\Step_1=\rp[p\pm k])$ and $\Prob_p(\Step_1=\rn[p\pm k])$ have non-trivial limits when $p\to\infty$. For that reason, when $p=\infty$ or $q=\infty$, we identify $\rho^\dagger$ with $\infty$ and introduce the following additional peeling step events:
\begin{description}[noitemsep]
\item[Event $\RR^*_{\infty-k}$:] $v$ is at a distance $k$ to the right of $\infty$ on the boundary of $\tmap$, viewed from the origin ($0\le k< \infty$);
\item[Event $\LL^*_{\infty-k}$:] $v$ is at a distance $k$ to the left of $\infty$ on the boundary of $\tmap$, viewed from the origin ($0\le k< \infty$).
\end{description}

Let $\steps=\tilde\steps\cup\{\RR^*_{\infty-k}, \LL^*_{\infty-k}:\ *\in\{\+,\<\},\ k\ge 0\}$. Observe that the set $\steps$ in \cite{CT20} corresponds to the set $\tilde\steps$ here. We make the identification $\RR^*_{\infty\mp k}=\LL^*_{\infty\pm k}$, as well as the convention $\Prob_{p,q}(\Step_1=\RR^*_{\infty\pm k})=\Prob_p(\Step_1=\RR^*_{\infty\pm k})=0$. Thus, the peeling process can always be defined on $\steps$.

We define $\Prob_\infty(\Step_1=\rp[\infty\pm k]):=\lim_{p\to\infty}\Prob_p(\Step_1=\rp[p\pm k])$ and $\Prob_\infty(\Step_1=\rn[\infty\pm k]):=\lim_{p\to\infty}\Prob_p(\Step_1=\rn[p\pm k])$. The events $\rp[\infty\pm k]$ and $\rn[\infty\pm k]$ can be viewed as jumps of the peeling process to the vicinity of $\infty$. This property of infinite jumps results a positive probability of bottlenecks in the local limit when $\nu>\nu_c$. See Section~\ref{sec:lowtemplimit} for a more precise analysis of the local limit structure in the low temperature regime. The peeling step probabilities for $p,q=\infty$ are collected in Table~\ref{tab:pinfty}.

\begin{table}[h]
\centering
\begin{tabu}[t]{|L|L|L||L|L|L| L}
\cline{1-6}
\step  & \Prob_\infty(\Step_1=\step)& (X_1,Y_1)	&
\step  & \Prob_\infty(\Step_1=\step)& (X_1,Y_1)	&
\hl	\cp	& \frac{t}{u}		& (2,-1)		&
	\cm	& \frac{\nu t}{u}		& (0,1) 	&
\hl	\lp	&       t u^k z_{1,k}	& (1,-k-1)		&
	\lm	& \nu t u^k z_{0,k+1}	& (0,-k)	&	(k\ge 0)
\hl	\rp	&       t u^k z_{k+1,0}	& (-k+1,-1)		&
	\rn	& \nu t u^k z_{k,1}	& (-k,0)	&	(k\ge 0)
\hl \rp[\infty-k] & \frac{t a_0}{b}a_{k+1}u^k\id_{\nu>\nu_c} &(-\infty,-1) &
	\rn[\infty-k] & \frac{\nu t a_1}{b}a_ku^k\id_{\nu>\nu_c} & (-\infty,0) & (k\ge 0)
\hl \rp[\infty+k] & \frac{t a_1}{b}a_ku^k\id_{\nu>\nu_c} & (-\infty, -k-1) &
	\rn[\infty+k] & \frac{\nu t a_0}{b}a_{k+1}u^k\id_{\nu>\nu_c} & (-\infty, -k) & (k>0)
\hl
\end{tabu}
\caption{Law of the first peeling event $\Step_1$ under $\Prob_\infty$ and the corresponding $(X_1,Y_1)$, where the peeling is without target. We have the same shorthand notation as in the previous tables as well as $b=b(\nu)$.}\label{tab:pinfty}
\end{table}

\begin{table}
\begin{center}
\begin{tabu}{|L|L|L| |L|L|L| L}
\cline{1-6}
\step	& \Probh_{p+1}(\Step_1=\step)				& (X_1,Y_1)	&
\step	& \Probh_{p+1}(\Step_1=\step)				& (X_1,Y_1) &
\phantom{k\le }
\hl	\cp & \nu t \aph{p+2}						&		(1,0)
&	\cm & 		t \aph{p} \frac1{u^2}			& (-1,2) 	&
\hl	\lp & \nu t z_{1,k} u^k				&		(0,-k)
&	\lm &		t z_{0,k+1} \aph{p} u^{k-1}	& (-1,-k+1)	& (k\ge 0)
\hl	\rp & \nu t z_{k+1,0}\aph{p-k+1}			& (-k,0)
&	\rn & 		t z_{k,1} \aph{p-k} \frac1{u}	&		(-k-1,1)				& (0\le k\le p)
\hl	\rp[p+k] & \nu t z_{p+1,k} \aph1 u^{k}	&		(-p,-k)
&	\rn[p+k] &		t z_{p,k+1} \aph0 u^{k-1}	& (-p-1,-k+1) & (k>0)
\hl
\end{tabu}
\vspace{3ex}

\begin{tabu}{|L|L|L||L|L|L| L}
\hline \step	& \Probh_\infty(\Step_1=\step)& (X_1,Y_1) &\step  & \Probh_\infty(\Step_1=\step) & (X_1,Y_1)
\\\hline \cp	& \frac{\nu t}{u} & (1,0)	&\cm		& \frac{t}{u} & (-1,2)
\\\hline \lp	& \nu t u^k z_{1,k} & (0,-k)	&\lm		& t u^k z_{0,k+1} & (-1,-k+1)
\\\hline \rp	& \nu t u^k z_{k+1,0} & (-k,0)	&\rn		& t u^k z_{k,1} & (-k-1,1)
\\\hline
\end{tabu}\phantom{$k\le$}\raisebox{3.2em}{}~~~
\caption{Laws of $\Step_1$ under $\Probh_{p+1}$ ($p\ge 0$) and $\Probh_\infty$, respectively, obtained by taking two successive limits in Table~\ref{tab:probh(p,q)}. The peeling is without target. Since we only need this distribution in the high temperature regime $\nu\in(1,\nu_c)$, the bottleneck events are omitted.
}\label{tab:probh(p)}
\end{center}\vspace{-1em}
\end{table}

The proof that $\Prob_p$ defines a probability distribution on $\steps$ goes similarly as in \cite[Lemma~6]{CT20}, as well as that $\Prob_\infty$ is a probability distribution on $\steps$ for $1<\nu<\nu_c$. For $\nu>\nu_c$, the total probability from Table~\ref{tab:pinfty} sums to
\begin{equation*}
t(\nu+1)\left(\frac{Z_0(u)}{u}+Z_1(u)+\frac{\frac{a_0}{u_c}+a_1}{b}\left(A(u)-a_0\right)\right),
\end{equation*} which is shown to be equal to one either by a coefficient extraction argument similar to the one of \cite[Lemma~6]{CT20}, or by a computer algebra calculation.

It follows that $\Prob_p$ and $\Prob_\infty$, respectively, can be extended to the distribution of the peeling process $\nseq{\Step}$, and we have the convergence $\Prob_{p,q}\cv[]q\Prob_p\cv[]p\Prob_\infty$ in distribution, where $\Prob_p$ and $\Prob_\infty$ satisfy the spatial Markov property (see \cite[Proposition~2, Corollary 7]{CT20}). By symmetric arguments, we recall the same properties for the laws $\Probh_p$ and $\Probh_\infty$, which are obtained as the distributional limits of $\Probh_{p,q}$. The explicit laws of the first peeling step are collected in Table~\ref{tab:probh(p)}. The expectations corresponding to $\Prob$ and $\Probh$ are called $\EE$ and $\hat{\EE}$, respectively.

By the diagonal asymptotics part of Theorem~\ref{thm:asympt}, it is also easy to see that convergences $\Prob_{p,q}\cv{p,q}\Prob_\infty$ and $\Probh_{p,q}\cv{p,q}\Probh_\infty$ hold for every appropriate $\nu$. More precisely, since the coefficient function $\lambda\mapsto c(\lambda)$ is continuous on every interval bounded away from zero for every fixed $\nu\in(1,\infty)$, we conclude that $\lim_{p,q\to\infty}\frac{c\left(\frac{q-m}{p-k}\right)}{c\left(\frac{q}{p}\right)}=1$ for any fixed $k,m\in\Z$ when $q/p\in [\lambda_{\min},\lambda_{\max}]$, and the convergence of the one-step peeling transition probabilities follows. The rest is a repetition of the proof of the convergence $\Prob_p\cv{p}\Prob_\infty$.

\subsection{Order parameters and connections to pure gravity}\label{sec:orderparam}

Using the information in Table~\ref{tab:pinfty}, it is not hard to express the order parameter $\mathcal O(\nu)$ defined in Proposition~\ref{prop:orderparam intro}. We obtain the following formula:
\begin{equation*}
\mathcal{O}(\nu):=\EE_\infty^\nu((X_1+Y_1)\id_{|X_1|\vee |Y_1|<\infty})=(\nu+1)t_c(\nu)\left(\frac{Z_0(u_c(\nu))}{u_c(\nu)}-Z'_0(u_c(\nu))-u_c(\nu)Z'_1(u_c(\nu))\right).
\end{equation*}
Above, the cases $|X_1|=\infty$ and $|Y_1|=\infty$ may appear if $\nu>\nu_c$, and the latter only if we consider the peeling with target $\rho^\dagger$. 
We have discussed in the introduction that $\mathcal O$ is an order parameter for the phase transition around $\nu=\nu_c$
Its properties are collected in Proposition~\ref{prop:orderparam intro}. See also Figure~\ref{fig:orderparam} for the graph of $\mathcal O$.



\begin{figure}
\centering
\includegraphics[scale=1.1]{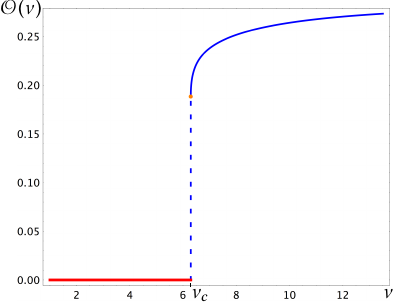}
\caption{The graph of the order parameter $\mathcal{O}$.}
\label{fig:orderparam}
\end{figure}

The proof of Proposition~\ref{prop:orderparam intro} is a computation by a computer algebra, presented in \cite{CAS2}. Note that $\mathcal{O}$ is discontinuous at $\nu=\nu_c$, and that $\mathcal{O}(\nu)=\EE^\nu_\infty(X_1+Y_1)$ for $1<\nu\leq\nu_c$.
Moreover, the above drift condition in this regime shows that the peeling process started from the $\<$ edge next to the origin drifts to the left, swallowing the $\<$ boundary piece by piece. By symmetry, we also obtain $\hat\EE^\nu_\infty(Y_1)=\EE^\nu_\infty(X_1)$ and $\hat\EE^\nu_\infty(X_1)=\EE^\nu_\infty(Y_1)$, which in turn yield that the peeling process following the right-most interface drifts to the right. In Section~\ref{sec:hightemplimit}, we will use these properties to modify the peeling algorithm so that the peeling process will explore a neighborhood of the origin in a metric sense, which will be enough to construct the local limit in the high temperature regime.

\begin{remark}\label{rem:orderparam2}
There is another, and perhaps more natural, order parameter \begin{equation*}
\tilde{\mathcal{O}}(\nu):=\Prob^\nu_\infty(|X_1|\vee|Y_1|=\infty)=\begin{cases}
0         & \text{if}\quad 1<\nu\le\nu_c \\
(\nu+1) t_c(\nu) \m({ \frac{\frac{a_0(\nu)}{u_c(\nu)}+a_1(\nu)}{b(\nu)}
                      \m({ A(u,\nu)-a_0(\nu) } }
& \text{if}\quad\nu>\nu_c \phantom{\le \nu_c}.
\end{cases}
\end{equation*}
It is easy to see that $\tilde{\mathcal{O}}$ is continuous at $\nu=\nu_c$. This order parameter is the probability of the occurrence of a finite bottleneck in a single peeling step in the (to-be-constructed) local limit $\prob_\infty^\nu$. It can also be shown to be increasing and have the limit $\frac{\sqrt{3}}{12}$ as $\nu\to\infty$. However, since the order parameter $\mathcal{O}$ encapsulates all what we need in the proofs of the local convergences, $\tilde{\mathcal{O}}$ is not studied further in this work. Its only non-zero occurrence is related to Lemma~\ref{lem:hit 0} in Section~\ref{sec:lowtemplimit}.
\end{remark}

\begin{remark}
In the physics literature, the order parameter for the two-dimensional Ising model is traditionally the magnetization of the Ising field. We do not know the connection of $\mathcal{O}$ or $\tilde{\mathcal{O}}$ to the magnetization. Unlike the magnetization for the Ising model on a regular lattice, $\mathcal{O}$ is discontinuous at the critical temperature. Moreover, it does not tell us about the global geometry of the spin clusters, rather it serves as a ``measure'' of the interface behaviour in the local limit. An interesting curiosity is that we can show the free energy density per boundary edge has a second order discontinuity, even though it is known that the free energy density per face has a third order discontinuity, telling that the phase transition should be of third order. More precisely, by the work of Boulatov and Kazakov \cite{BouKaz87} or an explicit computation \cite{CAS2}, we have
\begin{equation*}
-\lim_{n\to\infty}\frac{1}{n}\log([t^n]z_{p,q}(t,\nu))=F(\nu)
\end{equation*} where $n$ is the number of interior faces and $F$ has a third order discontinuity at $\nu=\nu_c$. However, we find
\begin{equation*}
-\lim_{q\to\infty}\frac{1}{q}\log(z_{p,q}(\nu))=-\lim_{p,q\to\infty}\frac{1}{q}\log(z_{p,q}(\nu))=\log(u_c(\nu)),
\end{equation*}
which can be shown to have a second order discontinuity at $\nu=\nu_c$.
\end{remark}

\paragraph{Pure gravity-like behavior and some literature remarks.} It has been conjectured by physicists that the Ising model outside the critical temperature falls within the pure gravity universal class (see \cite{ADJ97}). In particular, in the seminal work of Kazakov \cite{Kaz86}, the fact is justified by computing the zero-temperature and the infinite-temperature limits of the free energy, which both coincide with the ones derived from the one-matrix model. The analysis of our peeling process, and the geometry in the further sections, will give a different perspective to this phenomenon.

First, we note that $\lim_{\nu\searrow 1}\EE^\nu_\infty(Y_1)=-\lim_{\nu\searrow 1}\EE^\nu_\infty(X_1)=-\frac{1}{2}$. From \cite[Section 3.2]{AC13}, we check that this coincides with the drift of the perimeter process of an exploration which follows the right-most interface of a finite percolation cluster on the UIHPT decorated with a face percolation configuration with parameter $p=1/2$. This is natural due to the symmetry of the $\+$ and $\<$ spins. We stress that, since percolation on the triangular lattice is \emph{not} self-dual, this falls in the subcritical regime of percolation. Observe also that the geometry of large Boltzmann Ising-triangulations in the high temperature regime essentially should not depend on the exact value of $\EE^\nu_\infty(X_1)=-\EE^\nu_\infty(Y_1)$, as long as it is strictly positive and the perimeter exponents $\alpha_0+1$ and $\alpha_2+2$ of the asymptotics of Theorem~\ref{thm:asympt} are equal to $5/2$. Therefore, the geometry of the Ising-decorated random triangulation of the half-plane in the high temperature regime is similar to the one of the UIHPT decorated with subcritical face percolation. To our knowledge, this phenomenon has never been explicitly written, though intuitively well understood.

In the low temperature regime, we have $\lim_{\nu\to\infty}\mathcal{O}(\nu)=\frac{1}{2\sqrt{3}}$. This, in turn, coincides with the expectation of the number of edges swallowed (both to the right or to the left) after a peeling step of the non-decorated UIHPT of type I. At the level of the peeling process, we find that
\begin{equation*}
\begin{aligned}
      &\ \lim_{\nu\to\infty}\Prob^\nu_\infty(\Step_1=\cp)
\!\!&=&\ \lim_{\nu\to\infty} \sum_{k=0}^\infty \Prob^\nu_\infty(\Step_1=\lp)
\!\!&=&\ \lim_{\nu\to\infty} \sum_{k=0}^\infty \Prob^\nu_\infty(\Step_1=\rp)
\\
 =&\ \lim_{\nu\to\infty} \sum_{k=0}^\infty \Prob^\nu_\infty(\Step_1=\rp[\infty-k])
\!\!&=&\ \lim_{\nu\to\infty} \sum_{k=1}^\infty
                             \Prob^\nu_\infty(\Step_1=\lp[\infty-k])
\!\!&=&\ \lim_{\nu\to\infty} \sum_{k=1}^\infty
                             \Prob^\nu_\infty(\Step_1=\rn[\infty-k])
\ =\ 0
\end{aligned}
\end{equation*}
and
\begin{align*}
\lim_{\nu\to\infty}\Prob^\nu_\infty(\Step_1=\cm) &= \frac{1}{\sqrt{3}} &
\lim_{\nu\to\infty}\sum_{k=0}^\infty\Prob^\nu_\infty(\Step_1=\lm) \ \ \ \,
&= \frac{1}{2}-\frac{1}{2\sqrt{3}} \\
\lim_{\nu\to\infty}\sum_{k=0}^\infty\Prob^\nu_\infty(\Step_1=\rn)
&= \frac{1}{2}-\frac{\sqrt{3}}{4} &
\lim_{\nu\to\infty}\sum_{k=0}^\infty\Prob^\nu_\infty(\Step_1=\lm[\infty-k])
&= \frac{\sqrt{3}}{12}\,.
\end{align*}
Since these quantities sum to one, we conclude that either the bottlenecks survive in the zero temperature limit, or the limit does not define a probability distribution. The former follows if we can change the limit and the summation above, and that indeed can be done by the following simple argument: We notice that $\Prob^\nu_\infty(\Step_1=\lm[k])\sim\frac{M(\nu)}{k^{5/2}}$ as $k\to\infty$, where
\begin{equation*}
M(\nu)=\lim_{k\to\infty}k^{5/2}\left(\frac{\nu t_c(\nu)}{u_c(\nu)}\frac{a_0(\nu)}{\Gamma(-3/2)}(k+1)^{-5/2}\right)=\frac{\nu t_c(\nu)}{u_c(\nu)}\frac{a_0(\nu)}{\Gamma(-3/2)}.
\end{equation*}
An explicit computation shows that $\lim_{\nu\to\infty}M(\nu)\in (0,\infty)$, so $M(\nu)$ is bounded. Moreover, we can show that $\Prob^\nu_\infty(\Step_1=\lm[\infty-k])$ has exactly the same asymptotics as $k\to\infty$. By this asymptotic formula, one can then find a summable majorant for the above series for large enough $\nu$, and therefore the exchange of the limit and the sum follows from the dominated convergence theorem.

Hence, we find a zero temperature limit of the peeling process which shares the behaviour of the peeling process in the low temperature regime. In that case, the peeling process constructs an infinite triangulation, which consists of two infinite triangulations with the geometry of the UIHPT that are glued together by just one vertex, which can be viewed as a pinch point in the vicinity of both the origin and the infinity. The construction of this local limit is the same as in the upcoming Section~\ref{sec:lowtemplimit}.

The existence of the finite bottlenecks for $\nu>\nu_c$ is well predicted in the physics literature. More precisely: When $\nu=\infty$, the spins are totally aligned. Therefore, for $\nu>\nu_c$, it is predicted that the energy of a spin configuration is proportional to the length of the boundary separating different spin clusters, and hence the minimal energy configurations should be those with minimal spin interface lengths. In our setting, we consider an annealed model where we sample the triangular lattice together with the spin configuration. Hence, a bottleneck in the surface is formed. This is explained eg. in \cite{ADJ97} and \cite{AJL00}. To our knowledge, this is the first time when the existence of the bottlenecks on Ising-decorated random triangulations is shown rigorously.

\section{Local limits and geometry at $\nu\neq\nu_c$}\label{sec:locallimits}

In this section, we extend our analysis of the local limit at $\nu=\nu_c$, considered in \cite{CT20}, to the off-critical regimes $1<\nu<\nu_c$ and $\nu>\nu_c$. In \cite[Section 5]{CT20}, the idea was to provide a constructive proof of the local convergence in the following sense: The local limits were constructed by iterating the peeling process, and after noticing that the peeling explores any ball around the root with respect to the graph distance in a finite time, the local convergence followed from the convergence of the peeling process. In this work, we notice that similar proof strategy for the local convergence extends to every $\nu>1$, under certain amendments. For this reason, we reformulate the strategy used in \cite{CT20} for proving $\prob_{p,q}^{\nu_c}\ \cv[] q\ \prob_p^{\nu_c}$ as an algorithm with fairly general assumptions on the convergence of a peeling process. This strategy applies almost readily to all of the local limits at $\nu\in(\nu_c,\infty)$ with the same choice of a peeling algorithm as for $\nu=\nu_c$ by the fact that the interface hits a neighborhood of the infinity in a finite time. On the contrary, if $\nu\in (1,\nu_c)$, the peeling process under the aforementioned peeling algorithm will stay close to the boundary of the half-plane infinitely. Therefore, we need more refined arguments in the high temperature phase starting from Proposition~\ref{prop:orderparam intro}. In particular, we construct a mixed peeling algorithm, under which the peeling process explores a neighborhood of the origin layer by layer in the local limit $\prob\yy^{\nu}$. The choices of the peeling algorithm are summarised in Table~\ref{tab:peelingalgo} in the preceding section.

\subsection{Preliminaries: local distance and convergence}

For a map $\map$ and $r\ge 0$, we denote by $[\map]_r$ the \emph{ball of radius $r$} in $\map$, defined as the subgraph of $\map$ consisting of all the \emph{internal} faces which are adjacent to at least one vertex within the graph distance $r-1$ from the origin. By convention, the ball of radius $0$ is just the root vertex. The ball $[\map]_r$ inherits the planar embedding and the root corner of $\map$. Thus $[\map]_r$ is also a map. By extension, if $\sigma$ is a coloring of \emph{some faces} and \emph{some edges} of $\map$, we define the ball of radius $r$ in $(\map,\sigma)$, denoted $[\map,\sigma]_r$, as the map $[\map]_r$ together with the restriction of $\sigma$ to the faces and the edges in $[\map]_r$. In particular, we have $[[\map,\sigma]_{r'}]_r = [\map,\sigma]_r$ for all $r\le r'$.
Also, if an edge $e$ is in the ball of radius $r$ in a bicolored triangulation of a polygon $\bt$, then one can tell whether $e$ is a boundary edge by looking at $\btsq_r$, since only boundary edges are colored.

\newcommand{\CM}{\mathcal{C\hspace{-1pt}M}}
The \emph{local distance} for colored maps is defined in a similar way as for uncolored maps: for colored maps $(\map,\sigma)$ and $(\map',\sigma')$, let
\begin{equation*}
d\1{loc}((\map,\sigma),(\map',\sigma')) = 2^{-R}\qtq{where}
	R = \sup\Set{r\geq 0}{ [\map,\sigma]_r=[\map',\sigma']_r }\,.
\end{equation*}
The set $\CM$ of all (finite) colored maps is a metric space under $d\1{loc}$. Let $\overline \CM$ be its Cauchy completion. As was the case with the uncolored maps (see e.g.\ \cite{CurPeccot}), the space $(\overline \CM, d\1{loc})$ is Polish (i.e.\ complete and separable). The elements of $\overline \CM \setminus \CM$ are called \emph{infinite colored maps}. By the construction of the Cauchy completion, each element of $\CM$ can be identified as an increasing sequence of balls $(\bmap_r)_{r\ge 0}$ such that $[\bmap_{r'}]_r = \bmap_r$ for all $r\le r'$. Thus defining an infinite colored map amounts to defining such a sequence.
Moreover, if $(\prob\0n)_{n\ge 0}$ and $\prob\0\infty$ are probability measures on $\overline \CM$, then $\prob\0n$ converges weakly to $\prob\0\infty$ for $d\1{loc}$ if and only if
\begin{equation*}
\prob \0n([\map,\sigma]_r=\bmap) \ \cv[]n\ \prob \0\infty([\map,\sigma]_r=\bmap)
\end{equation*}
for all $r\ge 0$ and all balls $\bmap$ of radius $r$.

When restricted to the bicolored triangulations of the polygon $\bts$, the above definitions construct the corresponding set $\overline \bts \setminus \bts$ of infinite maps. Recall that $\bts_\infty^{(1)}$ is the set of \emph{infinite bicolored triangulation of the half plane}, that is, elements of $\overline \bts \setminus \bts$ which are one-ended and have an external face of infinite degree. Recall also the set $\bts_\infty^{(2)}$, consisting of \emph{two-ended} bicolored triangulations with an infinite boundary.

\subsection{A general algorithm for constructing local limits}\label{sec:generalalgo}

In this subsection, we provide an algorithm for constructing local limits and proving the local convergence for a generic setup of Boltzmann Ising-triangulations of the disk. The algorithm is already used in our previous work \cite{CT20} in the proof of the local convergence $\prob_{p.q}\cv[]q\prob_p$.

\newcommand{\probl}{\mathds{P}_l}
\newcommand{\problinf}{\mathds{P}_\infty}
\newcommand{\Probl}{\Prob\0l}
\paragraph{Assumptions.} Suppose we are given a family of probability measures $\{\probl:l=1,2,\dots\}$ supported on $\overline{\bts}$, where the index $l$ is either the full perimeter or the length of a finite boundary segment of a bicolored, possibly infinite, Boltzmann Ising-triangulation with Dobrushin boundary conditions. For example in the latter case, we may have $l=q$ if $\probl=\prob_{p,q}^\nu$ and we consider the convergence $q\to\infty$. The point is that $\{\probl:l=1,2,\dots\}$ is assumed to be a one-parameter family. Recall that the peeling process of a fixed triangulation can be viewed as a deterministic sequence $\nseq{\emap_n,\umap}$ of explored and unexplored maps, respectively, driven by a peeling algorithm $\mathcal{A}$. By convention, $\law\0l\nseq{\emap}$ denotes the law of the sequence of the explored maps under $\probl$. \newcommand{\pqn}[1][n]{\tilde p,\tilde q,#1}
\newcommand{\upqn}[1][n]{\umap_{\pqn[#1]}^*}

Let $(\upqn)_{\pqn \ge 0}$ be a family of independent random variables which are also independent of $\nseq \Step$, such that $\upqn$ is a Boltzmann Ising-triangulation of the $(\tilde p,\tilde q)$-gon, where possibly $\tilde p=\infty$ or $\tilde q=\infty$. Consider $\mathbb{Z}$ with its nearest-neighbor graph structure and canonical embedding in $\mathbb{C}$, viewed as an infinite planar map rooted at the corner at $0$ in the lower half plane. Then, the upper half plane is the unexplored map $\law\0l\umap_0$, and $\law\0l\emap_0$ is defined as the deterministic map $\Z$ in which the following holds, depending whether the boundary of length $l$ is monochromatic or not: the monochromatic boundary of length $l$ is contained in $[0,l]$ (if it has spin \+) or in $[-l,0]$ (if it has spin \<), or the bichromatic boundary of length $l=l_1+l_2$ is contained in $[-l_1,l_2]$. Assume that under $\probl$, one can recover the distribution of $\emap_n$ as a deterministic function of $\emap_{n-1}$, $\Step_n$ and $(\upqn)_{\tilde p,\tilde q \ge 0}$. We define $\law\0l \nseq \emap$ by iterating that deterministic function on $\law\0l \emap_0$, $\law\0l \nseq \Step$ and $(\upqn)_{\pqn \ge 0}$. Let $\filtr_n$ be the $\sigma$-algebra generated by $\emap_n$. Then the above construction defines a probability measure on $\filtr_\infty = \sigma(\cup_n \filtr_n)$, which we denote by $\Probl$. Moreover, assume $\Probl\cv[]l\Prob\0\infty$ in distribution with respect to the discrete topology. That is, there exists a distribution $\Prob\0\infty$ such that for any element $\omega$ in the (countable) state space of the sequences $\nseq \Step$ and $(\upqn)_{\pqn \ge 0}$ up to time $n_0<\infty$, we have $\Prob\0l(\omega) \cv[]l \Prob\0\infty(\omega)$. 

For the peeling algorithm $\mathcal{A}$, we make two assumptions. First, we assume that the algorithm is \emph{Dobrushin-stable}, in the sense that $\mathcal{A}$ always chooses a boundary edge at the junction of the \< and \+ boundary segments. This choice guarantees that the boundary condition always remains Dobrushin or monochromatic. Second, we assume that $\mathcal{A}$ is \emph{local}, by which we mean the following: If the boundary is bichromatic, $\mathcal{A}$ chooses the boundary edge according to the previous rule such that it is connected to the root $\rho$ via an explored region by the peeling excluding the boundary. On the other hand, if the boundary is monochromatic, $\mathcal{A}$ chooses an edge whose endpoints have a minimal graph distance to the origin, according to some deterministic rule if there are several such choices.

\paragraph{Convergence of the peeling process.}

Since $(\upqn)_{\pqn \ge 0}$ has a fixed distribution and is independent of $\nseq \Step$, it follows that $\law\0l\nseq \Step$ and $(\upqn)_{\pqn \ge 0}$ converge jointly in distribution when $l\to\infty$ with respect to the discrete topology.
However, because $\law\0l \emap_0$ takes a different value for each $l$, the initial condition $\law\0l \emap_0$ cannot converge in the above sense. This is not a problem, since for any positive integer $K$, the restriction of $\law\0l \emap_0$ to a finite interval $[-K,K]$ stabilizes at the value that is equal to the restriction of $\law\0\infty \emap_0$ on $[-K,K]$.
Therefore, let us consider the truncated map $\emapo_n$, obtained by removing from $\emap_n$ all the boundary edges adjacent to the hole. Then the number of the remaining boundary edges is finite and only depends on $(\Step_k)_{k\le n}$. It follows that for each $n$ fixed, $\emapo_n$ is a deterministic function of $(\Step_k)_{k\le n}$, $(\upqn[k])_{\tilde p,\tilde q\ge 0; k\le n}$ and $\emap_0$ restricted to some finite interval $[-K,K]$ where $K$ is determined by $(\Step_1,\dots, \Step_n)$. As the arguments of this function converge jointly in distribution with respect to the discrete topology (under which every function is continuous), the continuous mapping theorem implies that
\begin{equation}\label{eq:peeling cvg}
\Prob\0l(\emapo_n=\bmap) \cv[]l \Prob\0\infty(\emapo_n=\bmap)
\end{equation}
for every bicolored map $\bmap$ and for every integer $n\ge 0$.  We can extend this convergence for finite stopping times according to the following proposition, which is proven for \cite[Lemma~12]{CT20}, \emph{mutatis mutandis}.

\begin{prop}[Convergence of the peeling process]\label{lem:stopped peeling}
Let $\filtr^\circ_n$ be the $\sigma$-algebra generated by $\emapo_n$.
If $\theta$ is an $\nseq{\filtr^\circ}$-stopping time that is finite $\Prob\0\infty$-almost surely, then for every bicolored map $\bmap$,
\begin{equation}\label{eq:stopped peeling cvg}
\Prob\0l(\emapo_\theta = \bmap) \cv[]l \Prob\0\infty(\emapo_\theta = \bmap) \,.
\end{equation}
\end{prop}

\paragraph{Construction of $\problinf$.}

Recall that the explored map $\emap_n$ contains an uncolored face with a simple boundary called its hole. The unexplored map $\umap_n$ fills in the hole to give $\bt$. We denote by $\frontier_n$, called the \emph{frontier} at time $n$, the path of edges around the hole in $\emap_n$. For all $r\ge 0$, let $\theta_r = \inf\Set{n\ge 0}{ d_{\emap_n}(\rho,\frontier_n)\ge r}$, where $d_{\emap_n}(\rho,\frontier_n)$ is the minimal graph distance in $\emap_n$ between $\rho$ and vertices on $\frontier_n$. It is clear that this minimum is always attained on the truncated map $\emapo_n$, therefore $d_{\emap_n}(\rho,\frontier_n)$ is $\filtr^\circ_n = \sigma(\emapo_n)$-measurable and $\theta_r$ is an $\nseq{\filtr^\circ}$-stopping time. Expressed in words, $\theta_r$ is the first time $n$ such that all vertices around the hole of $\emap_n$ are at a distance at least $r$ from $\rho$. Since $\bt$ is obtained from $\emap_n$ by filling in the hole, it follows that
\begin{equation*}
 \btsq_r\ =\ [\emapo_{\theta_r}]_r
\end{equation*}
for all $r\ge 0$. In particular, the peeling process $\nseq\emap$ eventually explores the entire triangulation $\bt$ if and only if $\theta_r<\infty$ for all $r\ge 0$. A sufficient condition for this is provided by the following lemma.

\begin{lemma}\label{lem:cover r-ball}
If the frontier $\frontier_n$ becomes monochromatic in a finite number of peeling steps $\Prob\0\infty$-almost surely, then $\theta_r$ is almost surely finite for all $r\ge 0$.
\end{lemma}

\begin{proof}
 We have $\theta_0=0$. Assume that $\theta_r<\infty$ almost surely for some $r\ge 0$. Then the set $V_r$ of vertices at a graph distance $r$ from the origin in $\tmap$ is $\Prob\0\infty$-almost surely finite. Since by assumption the frontier $\frontier_n$ becomes monochromatic in a finite time $\Prob\0\infty$-almost surely, the spatial Markov property yields that $\frontier_n$ is monochromatic infinitely often.

On the event $\{\theta_{r+1}=\infty\}$ and at the times $n>\theta_r$ such that $\frontier_n$ is monochromatic, the peeling algorithm $\algo$ chooses to peel an edge with an endpoint in $V_r$ by the locality assumption of $\algo$. Since $V_r$ is finite, there exists a $v\in V_r$ at which such peeling steps occur infinitely many times. But each time the vertex $v$ is swallowed with a non-zero probability, as a consequence the transition probabilities of the one-step peeling. Therefore $v$ can remain forever on the frontier only with zero probability. This implies that $\Prob\0\infty(\theta_{r+1}<\infty)=1$. By induction, $\theta_r$ is finite $\Prob\0\infty$-almost surely for all $r\ge 0$.
\end{proof}

We define the infinite Boltzmann Ising-triangulation of law $\problinf$ by the laws of its finite balls $\law\0\infty \btsq_r := \lim\limits_{n \to\infty} \law\0\infty{} [\emap_n]_r$. The external face of $\law\0\infty \bt$ obviously has infinite degree. Moreover, every finite subgraph of $\law\0\infty \bt$ is covered by $\emap_n$ almost surely for some $n<\infty$. If the peeling process only fills finite holes by the family $(\upqn)_{\pqn \ge 0}$, it follows that the complement of a finite subgraph only has one infinite component. That is, $\problinf$ is one-ended, which together with the infinite boundary tells that the local limit is an infinite bicolored triangulation of the half-plane. If the limiting map, however, includes infinite holes to fill in with the peeling, the map has several infinite connected components with positive probability. In the following section, we see a concrete example of that case.

\subsection{The local limit at low temperatures $(\nu>\nu_c)$} \label{sec:lowtemplimit}
Throughout this subsection, fix $\nu\in (\nu_c,\infty)$. For simplicity, let us first consider the case where $q\to\infty$ and $p\to\infty$ separately. In Section~\ref{sec:orderparam}, the order parameter $\mathcal{O}$ told us that for $\nu>\nu_c$, the peeling process has a tendency to drift to infinity. Moreover, from Table~\ref{tab:pinfty} we already read that $X_1=-\infty$ with a positive probability. Thus, we have $\EE_\infty(X_1)=-\infty$. These properties intuitively mean that the left-most interface drifts to infinity much faster than in the critical temperature, in fact even in a finite time almost surely. Thus, the construction of the local limit and the proof of the local convergence follows by choosing $\algo=\algo_{\<}$ and after we verify the assumptions of the algorithm in the previous section. The geometric view is similar to that in the critical temperature \cite{CT20}, with the exception that in this case the interface in a realization of the local limit is contained in a ribbon which is finite. Therefore, the local limit is not one-ended, unlike in $\nu=\nu_c$, and contains a bottleneck between the origin and infinity.

In order to be more precise, let us consider the $\Prob_p$ -stopping time
\begin{equation*}
T_m\ =\ \inf\Set{n\ge 0}{P_n\le m},
\end{equation*}
where $m\ge 0$ is a cutoff. In particular, $T_0$ is the first time that the boundary of the unexplored map becomes monochromatic. Observe that for $p>2m$, we can write $T_m\ =\ \inf\Set{n\ge 0}{\Step_n\in\Set{\rp[p+k+1], \rn[p+k]}{k\geq-m}}$. This extends to $p=\infty$ in a natural way, and thus $T_m$ is also a well-defined stopping time under $\Prob_\infty$.

Following the notation of \cite{CT20}, denote by $\law_{p,q}X$ (resp.\ $\law_p X$ and $\law_\infty X$) a random variable which has the same law as the random variable $X$ under $\Prob_{p,q}$ (resp.\ under $\Prob_p$ and $\Prob_\infty$).
We start by giving an upper bound for the tail distribution of $T_0$, which implies in particular that the process $ \law_p\nseq P$ hits zero almost surely in finite time. In other words, the peeling process swallows the \+ boundary almost surely, exactly as for $\nu=\nu_c$. What makes the low-temperature regime different is that this property actually holds also for $\Prob_\infty$, since by the infinite jumps of the peeling process we may have $T_m<\infty$. Moreover, we can easily find the explicit distribution of $T_m$ under $\Prob_\infty$.

\begin{lemma}[Law of $T_m$, $\nu>\nu_c$]\label{lem:hit 0}
Let $\nu\in(\nu_c,\infty)$.
\begin{enumerate}
\item There exists $\gamma>0$ such that $\Prob_p(T_0> n)\le e^{-\gamma n}$ for all $p\ge 1$. In particular, $T_0$ is finite $\Prob_p$-almost surely.

\item  Under $\Prob_\infty$, the stopping time $T_m$ has geometric distribution with parameter
\begin{equation*}
r_m:=\Prob_\infty(P_1\leq m)=\Prob_\infty(T_m=1)
\end{equation*} supported on $\{1,2,\dots\}$. That is,
\begin{equation*}
\Prob_\infty(T_m=n)=(1-r_m)^{n-1}r_m
\end{equation*} for $n=1,2,\dots$. In particular, $T_m$ is finite $\Prob_\infty$-almost surely for all $m\geq 0$.
\end{enumerate}
\end{lemma}

\begin{proof}
Since $\nu>\nu_c$, we have $\Prob_p(\Step_1=\rn[p])\longrightarrow\Prob_\infty(\Step_1=\rn[\infty]):=\tilde{r}\in(0,1)$, which yields
$
\Prob_p(T_0= 1)  \ \ge\  \Prob_p(\Step_1=\rn[p])  \ \ge\  r
$
for all $p\ge 1$, for some constant $r\in (0,1)$. It follows by the Markov property and induction that for all $n\ge 0$,
\begin{equation*}
\Prob_p(T_0>n+1)	\ =  \ \EE_p \mb[{ \Prob_{P_n}(T_0\ne 1) \idd{T_0>n}  }
				\ \le\ (1-r)^{n+1}\,,
\end{equation*}from which the first claim follows.

For the second claim, the data of Table~\ref{tab:pinfty} for $\nu>\nu_c$ shows that
\begin{align*}
\Prob_\infty(T_m=1) &=\ \Prob_\infty(P_1\leq m)
\\ &=\ \sum_{k=0}^\infty \m({ \Prob_\infty(\Step_1=\rp[\infty+k])
                     + \Prob_\infty(\Step_1=\rn[\infty+k]) }
  \ +\ \sum_{k=1}^{m-1} \Prob_\infty(\Step_1=\rp[\infty-k])
  \ +\ \sum_{k=1}^m     \Prob_\infty(\Step_1=\rn[\infty-k])
\\ &=\ t\m({ \frac{a_0}{bu}\m({ \nu(A(u)-a_0)+\sum_{k=2}^ma_ku^k }
+\frac{a_1}{b}\m({ A(u)+\nu\sum_{k=1}^ma_ku^k}
} \ =:\ r_m.
\end{align*}
By the spatial Markov property and induction, \begin{equation*}
\Prob_\infty(T_m>n+1)	\ =  \ \EE_\infty \mb[{ \Prob_\infty(T_m\ne 1) \idd{T_m>n}  }
				= (1-r_m)^{n+1}\,
\end{equation*}for all $n\geq 0$, which shows that $T_m$ has geometric distribution with parameter $r_m$.
\end{proof}

\begin{remark}
Observe that by the above proof, $\lim_{m\to\infty}r_m=\tilde{\mathcal{O}}(\nu)$, which was introduced as an order parameter in Remark~\ref{rem:orderparam2}.
\end{remark}

The above lemma entails that $T_m$ can directly, without further conditioning, be regarded as a time of a large jump of the perimeter process. In other words, unlike in \cite{CT20}, a suitably chosen peeling process will explore any finite neighborhood of the origin in a finite time, and thus no gluing argument of locally converging maps is needed. In particular, the general algorithm of Section~\ref{sec:generalalgo} applies. If one wanted to study the local limit via gluing, one could note that conditional on $n<T_m$, the process $P_n$ has a positive drift, a behaviour reflected by the order parameter $\mathcal{O}$.

It is easy to see that the above lemma also holds if we define more generally $T_m:=\inf\{n\ge 0: \min\{P_n,Q_n\}\leq m\}$ and consider the convergence of the peeling process with the target $\rho^\dagger$ under the limit $p,q\to\infty$ while $q/p\in [\lambda',\lambda]$. We omit the details of this here. The stopping time $T_m$ is extensively studied in Section~\ref{sec:onejumpscaling} for $\nu=\nu_c$, and the computation techniques for $\nu>\nu_c$ are similar. The biggest difference compared to the $p=\infty$ case is the fact that in the $p,q<\infty$ case, the triangle revealed at the peeling step realizing $T_m$ must hit the boundary at a distance less than $m+1$ from $\rho^\dagger$. The perimeter variations $(X_1,Y_1)$ will also have a different law, and in particular both $X_1$ and $Y_1$ may have infinite jumps (though not simultaneously).

Recall that in our context of the peeling along the left-most interface, the peeling algorithm $\algo_\<$ is used to choose an edge adjacent to $\rho_n$ on the boundary of the unexplored map $\umap_n$ according to some deterministic function \emph{when its boundary $\frontier_n$ is monochromatic of spin \<} (see Section~\ref{sec:peeling}). Under $\Prob_p$, we can ensure $\theta_r<\infty$ almost surely for all $r\ge 0$ with the following choice of the peeling algorithm $\algo=\algo_\<$: Let $\rho_n$ be the vertex on the frontier realizing the minimal distance $d_{\emap_n}(\rho,\frontier_n)$ from the origin. Then $\algo_\<$ chooses the edge on the left of $\rho_n$.
This algorithm is obviously local. Since $T_0<\infty$ almost surely by Lemma~\ref{lem:hit 0}, Lemma~\ref{lem:cover r-ball} gives $\theta_r<\infty$ almost surely in $\Prob_p$. Moreover, everything in this paragraph clearly also holds after replacing $\Prob_p$ by $\Prob_\infty$.

\begin{proof}[Proof of the convergence $\prob_{p,q}^\nu \protect{\cv q} \prob_p^\nu \protect{\cv p} \prob_\infty^\nu$ for $\nu>\nu_c$]
The $(\filtr^\circ_n)$-stopping time $\theta_r$ is almost surely finite under $\Prob_p$ and $\Prob_\infty$, and $\btsq_r = [\emapo_{\theta_r}]_r$ is a measurable function of $\emapo_{\theta_r}$. Thus, the assumptions of the general algorithm for local convergence hold with the choice $\probl=\prob_{p,q}^\nu$ with $l=q$ in the first limit, and after $\prob_p^\nu$ is defined, also with $\probl=\prob_p^\nu$ with $l=p$ in the second limit. In the first limit, the family $(\upqn)_{\pqn \ge 0}$ consists of independent finite Boltzmann Ising-triangulations, which fill in the finite holes formed in the peeling process (exactly as in $\nu=\nu_c$, see \cite{CT20}). Assuming $\prob_p^\nu$ is defined for all $p\geq 0$, then the family $(\upqn)_{\pqn \ge 0}$ also contains the elements $\umap^*_{\infty,\tilde{q},n}$ with law $\prob_{\tilde{q}}^\nu$, which fill in the hole with infinite \+ boundary after a bottleneck is formed. Putting things together in this order, it follows from Proposition~\ref{lem:stopped peeling} that $\prob_{p,q}^\nu(\btsq_r = \bmap) \cv[]q \prob_p^\nu(\btsq_r = \bmap)\cv[]p\prob_\infty^\nu(\btsq_r = \bmap)$ for all $r\ge 0$ and every ball $\bmap$. This implies the local convergence $\prob_{p,q}^\nu \cv[]q \prob_p^\nu \cv[]p \prob_\infty^\nu$.
\end{proof}
\begin{proof}[Proof of the convergence $\prob_{p,q}^\nu \protect{\cv {p,q}} \prob_\infty^\nu$ while $0<\lambda'\leq\frac{q}{p}\leq\lambda$ for $\nu>\nu_c$]
The assumptions of the general algorithm for local convergence hold with the choice $\probl=\prob_{p,q}^\nu$ with $l=p+q$, where $l\to\infty$ such that $q\in [\lambda' p,\lambda p]$. Since the peeling process with the target $\rho^\dagger$ has the same limit in distribution as the untargeted one, the local limit is indeed $\prob_\infty^\nu$.
\end{proof}
The above constructed local limit $\prob_p^\nu$ is one-ended, since the peeling process only fills in finite holes. By Lemma~\ref{lem:hit 0}, the untargeted peeling process of the local limit $\prob_\infty^\nu$ swallows the infinite \+ boundary $\Prob_\infty$-almost surely in a finite time, resulting a finite bottleneck, after which the peeling process continues to peel the infinite triangulation with infinite \< boundary and finite \+ boundary. Since the latter one is one-ended, it follows that the local limit $\prob_\infty^\nu$ consists of two independent triangulations of laws $\prob^\nu_{\tilde{p}}$ and $\prob^\nu_{\tilde{q}}$, for some $\tilde{p}\ge 0$ and $\tilde{q}\ge 0$, the second one modulo a spin flip, glued together along a finite bottleneck. That is, the local limit $\prob_\infty^\nu$ is two-ended.

\subsection{The local limit at high temperatures $(1<\nu<\nu_c)$}\label{sec:hightemplimit}

Throughout this subsection, fix $\nu\in (1,\nu_c)$. We start by considering first the convergence $\prob_{p,q}^\nu\to\prob_p^\nu$. For that purpose, we choose the peeling algorithm $\algo_\+$ defined in Section~\ref{sec:peeling}. The reason for this choice is explained by Remark~\ref{rmk:pdrift} and Lemma~\ref{lem:hit0h} below. Again, the only thing to show is that $\Probh_p(T_0<\infty)=1$ for every finite $p\geq 0$. However, due to the fact that $\Probh_p(T_m=1)\sim c_m\cdot p^{-5/2}$, we need a different strategy as in \cite{CT20} to prove that result. At this point, recall the drift of the perimeter processes: $\EE_\infty(X_1)=-\EE_\infty(Y_1)>0$ from Proposition~\ref{prop:orderparam intro}. This drift is used to estimate the drift of the perimeter process for a large $p<\infty$.

\begin{lemma}\label{lem:integralconvergence}
Let $\nu\in (1,\nu_c)$. Then,
\begin{equation*}
\lim_{p\to\infty}\EE_p(X_1)=\EE_\infty(X_1)\qquad\text{and}\quad\lim_{p\to\infty}\EE_p(Y_1)=\EE_\infty(Y_1).
\end{equation*}
Likewise,\begin{equation*}
\lim_{p\to\infty}\hat{\EE}_p(X_1)=\hat{\EE}_\infty(X_1)\qquad\text{and}\quad\lim_{p\to\infty}\hat{\EE}_p(Y_1)=\hat{\EE}_\infty(Y_1).
\end{equation*}
\end{lemma}

\begin{proof}
For $p>m>1$, we make the decomposition
\begin{equation*}
\EE_p(X_1)=\EE_p\left(X_1\id_{\{X_1\geq -m\}}\right)+\EE_p\left(X_1\id_{\{X_1\leq -p+m\}}\right)+\EE_p\left(X_1\id_{\{X_1\in (-p+m,-m)\}}\right).
\end{equation*} By the convergence of the peeling process,
\begin{equation*}
\EE_p\left(X_1\id_{\{X_1\geq -m\}}\right)\cv[]p\EE_\infty\left(X_1\id_{\{X_1\geq -m\}}\right)\cv[]m\EE_\infty(X_1).
\end{equation*}
 For the second term, $\Prob_p(X_1\leq -p+m)=\Prob_p(P_1\leq m)\sim c_m\cdot p^{-5/2}$ as $p\to\infty$ for some constant $c_m>0$, which shows that
\begin{equation*}
\EE_p\left(X_1\id_{\{X_1\leq -p+m\}}\right)=-p\Prob_p(X_1\leq -p+m)+\sum_{k=0}^m k\Prob_p(X_1=k-p)\cv[]p 0.
\end{equation*}
Finally, the third term can be explicitly written using the data of Table~\ref{tab:prob(p)} as
\begin{equation*}
\EE_p\left(X_1\id_{\{X_1\in (-p+m,-m)\}}\right)=-\sum_{k=m+1}^{p-m-1}k\Prob_p(X_1=-k)=-\sum_{k=m+1}^{p-m-1}k\left(tz_{k+2,0}\frac{a_{p-k}}{a_p}u+\nu t z_{k,1}\frac{a_{p-k}}{a_p}\right).
\end{equation*}
By the asymptotics $z_{k+2,0}\frac{a_{p-k}}{a_p}\underset{p\to\infty}\sim z_{k+2,0}u_c^k\underset{k\to\infty}\sim\text{ cst}\cdot k^{-5/2}$ and a similar one for $z_{k,1}\frac{a_{p-k}}{a_p}$, the sum on the right hand side can be approximated by a remainder of a convergent series, and therefore taking the limits $p,m\to\infty$ yields the claim.

The case $\lim_{p\to\infty}\EE_p(Y_1)=\EE_\infty(Y_1)$ is similar, except easier, since it only requires one cutoff at $Y_1=-m$. Indeed, the same asymptotics hold for $Y_1$. The cases $\lim_{p\to\infty}\hat{\EE}_p(X_1)=\hat{\EE}_\infty(X_1)$ and $\lim_{p\to\infty}\hat{\EE}_p(Y_1)=\hat{\EE}_\infty(Y_1)$ follow by symmetry.
\end{proof}

\begin{remark}\label{rmk:pdrift}
By Proposition~\ref{prop:orderparam intro} and symmetry, we have then
\begin{equation*}
\lim_{p\to\infty}\EE_p(X_1)>0\qquad\text{and}\quad\lim_{p\to\infty}\hat{\EE}_p(X_1)<0,
\end{equation*}
and likewise\begin{equation*}
\lim_{p\to\infty}\EE_p(Y_1)<0\qquad\text{and}\quad\lim_{p\to\infty}\hat{\EE}_p(Y_1)>0.
\end{equation*} This property is the main implication of Lemma~\ref{lem:integralconvergence}, which we keep on using in this section.
\end{remark}

\begin{remark}
In \cite{CT20}, we used the same decomposition to show that $\EE_p(X_1)\cv[]p-\frac{1}{3}\EE_\infty(X_1)<0$ at $\nu=\nu_c$. This blow-up of the probability mass was due to the fact that $\Prob_p(X_1\leq -p+m)\sim c_m\cdot p^{-1}$ at $\nu=\nu_c$. Currently, we do not have an interpretation of this symmetry breaking.

\end{remark}

Under mild conditions, a Markov chain on the positive integers with an asymptotically negative drift is expected to be recurrent. The next lemma verifies this in our case.

\begin{lemma}\label{lem:hit0h}
If $\nu\in(1,\nu_c)$, then $T_0$ is finite $\Probh_p$-almost surely.
\end{lemma}
\begin{proof}
Since $\nseq{P}$ is an irreducible Markov chain on the positive integers, it is enough to show that $\Probh_{p'}(T_{p'}<\infty)=1$ for some $p'>0$. Namely, this means that the chain will return to the finite set $\{0,\dots p'\}$ infinitely many times, and thus there exists a recurrent state.

Observe that by Lemma~\ref{lem:integralconvergence}, there exists an index $p_*>0$ such that $\hat{\EE}_{p'}(X_1)\leq-a$ for some $a>0$ if $p'>p_*$. On the other hand, $\hat{\EE}_{p'}(X_1)\leq\max_{0\leq i\leq p_*}\hat{\EE}_{i}(|X_1|)<\infty$ for $p'\leq p_*$. Then, it follows that the set $\{0,\dots,p_*\}$ is actually positive recurrent; see \cite[Theorem~1]{FK04} for a more general statement via Lyapunov functions, in which the Lyapunov function is chosen to be the identical mapping.
\end{proof}

Now, the proof of the local convergence $\prob_{p,q}^\nu\cv{q}\prob_p^\nu$ goes along the same lines as in the case $\nu\geq \nu_c$. Let us proceed to the proof of the convergence $\prob_p^\nu\cv{p}\prob_\infty^\nu$. For this, we cannot just choose the peeling algorithm $\algo_\+$ (or $\algo_\<$, respectively) since the peeling exploration under that algorithm drifts to the right (resp. to the left) in the limit by Lemma~\ref{lem:integralconvergence}. These drifts, however, allow us to construct a mixed peeling algorithm $\algo=\algo_m$ as follows.

\paragraph{Peeling algorithm $\algo_m$.} Recall that for $\nu\in(1,\nu_c)$, we have the drift conditions $\EE_\infty(X_1)>0$ and $\EE_\infty(Y_1)<0$ (together with the symmetric conditions $\hat{\EE}_\infty(X_1)<0$ and $\hat{\EE}_\infty(Y_1)>0$). These conditions allows us to construct the following sequence of stopping times:


Set $X_0=Y_0=0$ and $\tau_0^r=0.$ \begin{itemize}
\item Start peeling with $\algo_\<$ until the time $\tau_1^l:=\inf\{n>0:Y_n<-1\}$, which is almost surely finite under $\Prob_\infty$ due to the drift condition.
\item Proceed peeling with $\algo_\+$ until the time $\tau_1^r:=\inf\{n>\tau_1^l:X_n<-X_{\tau_1^l}-1+\min_{0\leq m\leq\tau_1^l}X_m\}$, which is a.s. finite under $\hat{\Prob}_\infty$, conditional on $\tau_1^l$.
\end{itemize} Repeat inductively for $k\geq 1$:
\begin{itemize}
\item At time $\tau_{k-1}^r$, run peeling with $\algo_\<$ until $\tau_k^l:=\inf\{n>\tau_{k-1}^r:Y_n<-Y_{\tau_{k-1}^r}-1+\min_{\tau_{k-1}^l\leq m\leq\tau_{k-1}^r}Y_m\}$.
\item At time $\tau_k^l$, run peeling with $\algo_\+$ until $\tau_k^r:=\inf\{n>\tau_k^l:X_n<-X_{\tau_k^l}-1+\min_{\tau_{k-1}^r\leq m\leq\tau_k^l}X_m\}$.
\end{itemize}

Obviously, the above constructed $\algo_m$ is a local and a Dobrushin-stable peeling algorithm. We denote the law of this peeling process by $\tilde\Prob_p$ ($p\in\N\cup\infty$). Note that the above stopping times may be infinite if $p<\infty$. Let
\begin{equation*}
\tilde\theta_R:=\inf\{n>\tau_R^l:X_n<-X_{\tau_R^l}-1+\min_{\tau_{R-1}^r\leq m\leq\tau_R^l}X_m\}.
\end{equation*}
The stopping time $\tilde\theta_R$ may be infinite for $p<\infty$, but the drift condition assures that $\tilde\Prob_\infty(\tilde\theta_R<\infty)=1$. It follows that under $\tilde\Prob_\infty$, the peeling process with algorithm $\algo_m$ explores the half-plane by distance layers, in the sense that the finite stopping time $\tilde\theta_R$ is an upper bound for the covering time $\theta_R$ of the ball of radius $R$. Hence, choosing $\algo=\algo_m$ in the general construction of the local limit and $\theta=\theta_R$ in Proposition~\ref{lem:stopped peeling} will give the construction of $\prob_\infty^\nu$ and yield the local convergence $\prob_p^\nu\cv{p}\prob_\infty^\nu$. To be a bit more precise, we still need to verify that $\tilde\Prob_p\to\tilde\Prob_\infty$ weakly. This is shown in the following lemma.


\begin{lemma}\label{lem:mixedpeelingconv}
Let $\nu\in(1,\nu_c)$. Then $\tilde\Prob_p\to\tilde\Prob_\infty$ as $p\to\infty$.
\end{lemma}

\begin{proof}
From the construction of $\tilde\Prob_p$ and by the spatial Markov property, for all $n\ge 1$ and all $\step_1,\cdots,\step_n\in \steps$, as well as for all $k\in [1,n]$ and $1\le m_1^l\le m_1^r\le\dots\le m_k^l\le m_k^r\le n$, we have
\begin{align*}
&\tilde\Prob_p(\Step_1=\step_1,\cdots, \Step_n=\step_n, \tau_1^l=m_1^l, \tau_1^r=m_1^r,\dots,\tau_k^l=m_k^l, \tau_k^r=m_k^r)\\ &=\Prob_p(\Step_1=\step_1,\dots,\Step_{m_1^l}=\step_{m_1^l},\tau_1^l=m_1^l)\hat\Prob_{p+x_{m_1^l}}(\Step_1=\step_{m_1^l+1},\dots,\Step_{m_1^r-m_1^l}=\step_{m_1^r},\tau_1^r=m_1^r) \\ &\cdots\Prob_{p+x_{m_{k-1}^r}}(\Step_1=\step_{m_{k-1}^r+1},\dots,\Step_{m_k^l-m_{k-1}^r}=\step_{m_k^l},\tau_k^l=m_k^l)\hat\Prob_{p+x_{m_k^l}}(\Step_1=\step_{m_k^l+1},\dots,\Step_{m_k^r-m_k^l}=\step_{m_k^r},\tau_k^r=m_k^r)\\
&\cdot\Prob_{p+x_{m_k^r}}(\Step_1=\step_{m_k^r+1},\dots,\Step_{n-m_k^r}=\step_n),
\end{align*} where the peeling events $(\step_i)_{1\le i\le n}$ completely determine the perimeter variations $(x_i)_{1\le i\le n}$. By the convergences $\Prob_p\to\Prob_\infty$ and $\hat{\Prob}_p\to\hat{\Prob}_\infty$, and by another application of the spatial Markov property, the right hand side tends to the limit $\tilde\Prob_\infty(\Step_1=\step_1,\cdots, \Step_n=\step_n, \tau_1^l=m_1^l, \tau_1^r=m_1^r,\dots,\tau_k^l=m_k^l, \tau_k^r=m_k^r)$. The claim follows.

\end{proof}

\begin{proof}[Proof of the convergence $\prob_p^\nu \protect{\cv p} \prob_\infty^\nu$ for $1<\nu<\nu_c$]
We write
\begin{equation*}
\prob_p^\nu(\btsq_R=\bmap)=\prob_p^\nu([\emapo_{\tilde\theta_R}]_R=\bmap, \tilde\theta_R<N)+\prob_p^\nu(\btsq_R=\bmap, \tilde\theta_R\geq N),
\end{equation*}
where the last term satisfies
\begin{equation*}
\prob_p^\nu(\btsq_R=\bmap, \tilde\theta_R\geq N)\leq\tilde{\Prob}_p(\tilde\theta_R\ge N)\cv[]p\tilde{\Prob}_\infty(\tilde\theta_R\ge N)\cv[]N 0
\end{equation*}
by Lemma~\ref{lem:mixedpeelingconv} and the drift condition. Thus, letting first $p\to\infty$ and then $N\to\infty$ yields the claim.
\end{proof}

\begin{proof}[Proof of the convergence $\prob_{p,q}^\nu \protect{\cv {p,q}} \prob_\infty^\nu$ while $0<\lambda'\leq\frac{q}{p}\leq\lambda$ for $1<\nu<\nu_c$]
It is not hard to see that the counterparts of the above lemmas also hold, \emph{mutatis mutandis}, in the diagonal regime. The essential matter is that the peeling processes under $\prob_{p,q}^\nu$ converge towards the peeling processes under $\prob_\infty^\nu$, due to the diagonal asymptotics. Again, we take into account $\rho^\dagger$ as a target.
\end{proof}

\section{The local limit at $\nu=\nu_c$ in the diagonal regime}\label{sec:locallimit c diag}

Throughout this section, we assume that $\nu=\nu_c$ and $\frac{q}{p}\in[\lambda',\lambda]$ for some $0<\lambda'\leq 1\leq\lambda<\infty$ as $p,q\to\infty$, and study the local limit of $\prob_{p,q}$ in this setting. We stress that this diagonal regime is slightly less general than in Theorem~\ref{thm:asympt}, since we require that it always contains the main diagonal $p=q$. The reason is purely technical and becomes evident in the proof of Lemma~\ref{lem:cst barrier} in Appendix \ref{sec:bigjumplemma proof}, where we need to control a ratio of random perimeters. We find, unsurprisingly, the same local limit $\prob_\infty=\prob_\infty^{\nu_c}$ as discovered in \cite{CT20}. Moreover, we find the scaling limit of the random time at which the peeling process jumps to a neighborhood of $\rho^\dagger$. The starting point of our analysis is the diagonal asymptotics (Theorem~\ref{thm:asympt}) \begin{equation*}\label{eq:diagasympt}
z_{p,q}(\nu)\sim
\frac{b\cdot c(q/p)}{\Gamma\left(-\frac{4}{3}\right)\Gamma\left(-\frac{1}{3}\right)}u_c^{-(p+q)}p^{-11/3}\qquad (\nu=\nu_c) .
\end{equation*}It is then easy to see that the peeling step probabilities converge to the same limits as in \cite{CT20} in the respective diagonal regime. However, it is natural to make the following modification for the peeling process.

\begin{table}[b!]
\centering
\begin{tabu}{|L|S|C|  |L|S|C| L}
\cline{1-6}
\step &\Prob_{p,q+1}(\Step_1 = \step) &(X_1,Y_1) &
\step &\Prob_{p,q+1}(\Step_1 = \step) &(X_1,Y_1)
\hl \cp	& 		\zp{p+2,q}				&(2,-1)
&	\cm	& \nu \zp{p,q+2} 				&(0,1)
\hl \lp	& 		\zzp{p+1,q-k}{1,k}		&(1,-k-1)
&	\lm	& \nu 	\zzp{p,q-k+1}{0,k+1}		&(0,-k)	&	(0\le k \le\theta q)
\hl	\rp	& 		\zzp{k+1,0}{p-k+1,q}		&(-k+1,-1)
&	\rn & \nu 	\zzp{k,1}{p-k,q+1}		&(-k,0)	&	(0\le k\le \theta p)
\hl \lp[q-k]	& 		\zzp{p+1,k}{1,q-k}	&(1,-q+k-1)
&	\lm[q-k]	& \nu 	\zzp{p,k+1}{0,q-k+1}	&(0,-q+k)
&(0<k<\theta q)
\hl \rp[p-k]	& 		\zzp{p-k+1,0}{k+1,q}	&(-p+k+1,-1)
&	\rn[p-k]	& \nu 	\zzp{p-k,1}{k,q+1}	&(-p+k,0)
&(0\le k<\theta p)
\hl
\end{tabu}
\caption{Law of the first peeling event $\Step_1$ under $\Prob_{p,q+1}$ and the corresponding $(X_1,Y_1)$ under the peeling process of the left-most interface with the target $\rho^\dagger$. In the table, $\theta\in (0,1)$ is an arbitrary cutoff, which roughly measures whether the perimeter process has only small jumps or not. Observe that the last two rows of the table are redundant with the second and the third row, respectively, in order to emphasize the cutoff for taking the limit. Taking the limit $(p,q)\to\infty$ gives the data of Table~\ref{tab:pinfty}.
}\label{tab:prob(p,q)diag}
\end{table}

We choose the peeling process with the target $\rho^\dagger$, driven by the peeling algorithm $\algo_{\<}^\dagger$ and described in Section~\ref{sec:peeling}: If the peeling step $\step_n$ splits the triangulation into two pieces, we choose the unexplored part $\umap_n$ to be the one containing $\rho^\dagger$. If $\rho^\dagger$ is included in both, we choose the one in the right. This gives rise to a different perimeter variation process $(X_n,Y_n)$, whose law is described in Table~\ref{tab:prob(p,q)diag}.

Accordingly, we define for $m\geq 0$
\begin{equation*}
T_m:=\inf\{n\ge 0: \min\{P_n,Q_n\}\leq m\}.
\end{equation*}
In other words, $T_m$ is just the first time at which either the \+ or the \< boundary length of the unexplored map is at most $m$. Using the peeling steps, we also see that $T_m=\inf\{n\ge 0: \Step_n\in\{\rp[p-k+1],\rn[p-k], \lp[q-k-1], \lm[q-k]: 0\leq k\leq m\}\}$. The analysis of the hitting time $T_m$ yield the main new results of this section. The first one is a technical lemma which generalizes the so-called one-jump lemma of \cite{CT20} to the diagonal setting. Its proof follows the recipe given in \cite[Appendix C]{CT20}, although due to taking the limit along a diagonal, additional technicalities arise. The second result is Theorem~\ref{thm:scaling}, whose proof mimics the proof of \cite[Proposition~11]{CT20}. A key novelty of the two aforementioned proofs in our current work is controlling the ratio $Q_n/P_n$ of the perimeter during the course of the peeling exploration described above. Finally, we detail the proof of the local convergence in the diagonal regime, which follows the idea presented in \cite[Sections 5.4-5.5]{CT20}, with important modifications resulting from the fact that there is no presence of an infinite boundary before taking the limit. However, it turns out that applying the one-jump lemma~\ref{lem:one jump diag} works almost exactly like applying the corresponding lemma in \cite{CT20}.

\subsection{The one-jump phenomenon of the perimeter process}\label{sec:onejumpscaling} 

\newcommand{\tauxy}{\tau^\epsilon_x}
\newcommand{\barrier}[1][x]{#1 f_\epsilon}

Next, we investigate an analog of the \emph{large jump} phenomenon discovered in \cite{CT20}. For that, fix $\epsilon>0$ and let \begin{equation*}
\barrier[](n) = \mb({ (n+2)(\log(n+2))^{1+\epsilon} }^{3/4}.
\end{equation*}Define the stopping time

\begin{equation*}
\tauxy = \inf\Set{n\ge 0}{\abs{X_n-\mu n} \vee \abs{Y_n-\mu n} > \barrier(n) }\,.
\end{equation*}
where $x>0$.

\begin{lemma}[One jump to zero]\label{lem:one jump diag}
For all $\epsilon>0$ and $0<\lambda'\leq 1\leq\lambda<\infty$,
\begin{equation*}
\lim_{x,m \to\infty} \limsupp \Prob_{p,q} (\tauxy<T_m) = 0\quad\text{while}\quad\frac{q}{p}\in[\lambda',\lambda].
\end{equation*}
\end{lemma}

The proof of Lemma~\ref{lem:one jump diag} is a modification of the proof of the analogous Lemma~10 in \cite{CT20}. The necessary changes are left to Appendix~\ref{sec:bigjumplemma proof}. Next, we prove the main scaling limit result of this article.


\newcommand*{\anom}{\mathcal{E}}
\newcommand*{\nom}{\mathcal{N}}
\begin{proof}[Proof of Theorem~\ref{thm:scaling}]
First, assuming that a scaling limit of $p^{-1}T_m$ exists for every $m\geq 0$, it actually does not depend on $m$. Namely, since $T_0\ge T_m$, the strong Markov property gives
\begin{align}
 \Prob_{p,q}(T_0-T_m >\epsilon p)	\ &=	 \	\EE_{p,q}\m[{ \Prob_{P_{T_m},Q_{T_m}}(T_0 >\epsilon p) }
\notag	\\ &
\le\	\EE_{p,q}\m[{\sum_{p'=0}^m\Prob_{p',Q_{T_m}}(T_0 >\epsilon p)+\sum_{q'=0}^m\Prob_{P_{T_m},q'}(T_0>\epsilon p)}.\label{eq:Tdifference}
\end{align}
Let $M>0$ be some large constant, and fix $p'\leq m$. We write \begin{align*}
\Prob_{p',Q_{T_m}}(T_0 >\epsilon p)&=\Prob_{p',Q_{T_m}}(T_0 >\epsilon p, \ Q_{T_m}>M)+\Prob_{p',Q_{T_m}}(T_0 >\epsilon p| \ Q_{T_m}\leq M)\Prob_{p',Q_{T_m}}(Q_{T_m}\leq M)\\ &\leq\Prob_{p',Q_{T_m}}(T_0 >\epsilon p, \ Q_{T_m}>M)+\sum_{q'=0}^M\Prob_{p',q'}(T_0>\epsilon p).
\end{align*}
By \cite[Proposition~2]{CT20} (actually, by its analog for the peeling with target), $\Prob_{p',q}\cv[]q\Prob_{p'}$. Therefore, the first term can be bounded from above by $\Prob_{p'}(T_0>\epsilon p)+\epsilon'$ for any $\epsilon'>0$,  provided $M$ is large enough. In that case, we obtain
\begin{equation*}
\sum_{p'=0}^m\Prob_{p',Q_{T_m}}(T_0 >\epsilon p)\le \sum_{p'=0}^m\Prob_{p'}(T_0>\epsilon p)+\sum_{p'=0}^m\sum_{q'=0}^M\Prob_{p',q'}(T_0>\epsilon p)+(m+1)\epsilon'.
\end{equation*}It is easy to see that the right hand side converges to zero as $p\to\infty$ and $\epsilon'\to 0$. The second term in Equation~\eqref{eq:Tdifference} is treated similarly, and finally we deduce $\Prob_{p,q}(T_0-T_m >\epsilon p)\cv[]{p,q} 0$.

Let us then proceed to the existence of the scaling limit. First, fix $x>0$, $m\in\natural$ and $\epsilon\in(0,\mu)$. Take $p$ and $q$ large enough such that $\Prob_{p,q}$-almost surely, $\tauxy\le T_m$. Denote $\anom := \{ \tauxy<T_m \}$ and $\nom_n := \{ \tauxy>n  \}$. Clearly $\nseq \nom$ is a decreasing sequence, and one can check that
\begin{equation}\label{eq:anomaly inclusion}
\nom_{n+1}	\ \subset\ \nom_n\setminus\{T_m=n+1\}
				\ \subset\ \nom_{n+1} \cup \anom\,.
\end{equation}

Let $c_m(\lambda):=\lim_{p,q\to\infty} p\cdot\Prob_{p,q}(T_m=1)$, where the limit is taken such that $q/p\to\lambda$. In other words, $q=\lambda p+o(p)$, and from the asymptotics of Theorem~\ref{thm:asympt},

\begin{equation}\label{eq:def c_m}
\Prob_{p,q}(T_m=1)\sim\frac{c_m\left(\frac{q}{p}\right)}{p}
\end{equation}
as $p,q\to\infty$, $q/p\to\lambda$. 
On the event $\nom_n$, we have $P_0+\mu n -\barrier(n) \le P_n\le P_0+\mu n + \barrier(n)$ and $Q_0+\mu n -\barrier(n) \le Q_n\le Q_0+\mu n + \barrier(n)$. This, in particular, gives
\begin{equation}\label{eq:lambdaestimate}
\frac{\lambda p+\mu n-x f_\epsilon(n)+o(p)}{p+\mu n+x f_\epsilon(n)}\le\frac{Q_n}{P_n}\le\frac{\lambda p+\mu n+x f_\epsilon(n)+o(p)}{p+\mu n-x f_\epsilon(n)}.
\end{equation}
Denote $\lambda_n:=Q_n/P_n$. Then combining the previous equation with \eqref{eq:def c_m}, we also obtain that for $P_0=p$ and $Q_0=q$ large enough,
\newcommand*{\cmore}[1][p]{ \frac{c_m(\lambda_n) +\epsilon}{#1+\mu n - \barrier(n)} }
\newcommand*{\cless}[1][p]{ \frac{c_m(\lambda_n) -\epsilon}{#1+\mu n + \barrier(n)} }
\begin{equation*}
\cless \id_{\nom_n} \ \le\ \id_{\nom_n} \Prob_{P_n,Q_n}(T_m=1) \ \le\ \cmore	\id_{\nom_n} \,.
\end{equation*}
By Markov property, $\Prob_{p,q}(\nom_n\setminus\{T_m=n+1\})
= \Prob_{p,q}(\nom_n) - \EE_{p,q}\m[{ \id_{\nom_n} \Prob_{p,q}(T_m=1) }$. Therefore
\begin{align*}
			\m({1-\cmore} \Prob_{p,q}(\nom_n) &\ \le\ \Prob_{p,q}(\nom_n\setminus\{T_m=n+1\})
\\&\ \le\	\m({1-\cless} \Prob_{p,q}(\nom_n)\,.
\end{align*}
Combining these estimates with the two inclusions in \eqref{eq:anomaly inclusion}, we obtain the upper bounds
\begin{equation*}
	\Prob_{p,q}(\nom_{n+1})\ \le\ \m({1-\cless} \Prob_{p,q}(\nom_n) \,,
\end{equation*}
and the lower bounds
\begin{align*}
\Prob_{p,q}(\nom_{n+1}\cup\anom)
  &\ \ge\ \Prob_{p,q}\m({ (\nom_n\setminus\{T_m=n+1\}) \cup\anom }
\\&\ \ge\ \Prob_{p,q}(\nom_n\setminus\{T_m=n+1\}) + \Prob_{p,q}(\anom\setminus\nom_n)
\\&\ \ge\ \m({1-\cmore} \Prob_{p,q}(\nom_n) + \Prob_{p,q}(\anom\setminus \nom_n)
\\&\ \ge\ \m({1-\cmore} \Prob_{p,q}(\nom_n\cup \anom) \,.
\end{align*}
Then, by iterating the two bounds, we get
\begin{equation*}
\Prob_{p,q}(\nom_N) \le \prod_{n=0}^{N-1} \m({1-\cless} \quad\text{and}\quad
\Prob_{p,q}(\nom_N\cup \anom) \ge \prod_{n=0}^{N-1} \m({1-\cmore}
\end{equation*} for any $N\ge 1$.
Since $\nom_n\subset \{T_m>n\}\subset \nom_n \cup\anom$ up to a $\Prob_{p,q}$-negligible set, the above estimates imply that
\begin{equation*}
		\prod_{n=0}^{N-1} \m({1-\cmore} - \Prob_{p,q}(\anom) \ \le\ \Prob_{p,q}(T_m>N)
\ \le\	\prod_{k=0}^{N-1} \m({1-\cless} + \Prob_{p,q}(\anom)\,.
\end{equation*}
From the Taylor series of the logarithm we see that $-x-x^2\le \log(1-x)\le -x$ for all $x\ge 0$. Therefore, for any positive sequence $\nseq x$, we have
\begin{equation*}
\exp\mB({ -\sum_{n=0}^{N-1} x_n -\sum_{n=0}^{N-1} x_n^2}
\ \le\ \prod_{n=0}^{N-1}(1-x_n)\ \le\ \exp\mB({ -\sum_{n=0}^{N-1} x_n }\,.
\end{equation*}

Now, we consider the sum \begin{equation*}\sum_{n=0}^{tp}\frac{c_m(\lambda_n)\pm\epsilon}{p+\mu n\mp x f_\epsilon(n)}.
\end{equation*}
First, by \eqref{eq:lambdaestimate}, we see that $\lambda_n=\frac{\lambda p+\mu n}{p+\mu n}\left(1+o(1)\right)$ where $o(1)$ is uniform over all $n\in[0,tp]$ as $p\to\infty$. Namely,
\begin{equation*}
\abs{1-\frac{\lambda p+\mu n\pm xf_\epsilon(n)+o(p)}{p+\mu n\mp x f_\epsilon(n)}\cdot\frac{p+\mu n}{\lambda p+\mu n}}=\abs{\frac{x f_\epsilon(n)\left(\frac{p+\mu n}{\lambda p+\mu n}+1\right)+o(p)}{p+\mu n\mp x f_\epsilon(n)}},
\end{equation*}
where the right hand side tends to zero uniformly on $n\in[0,tp]$ as $p\to\infty$.
On the other hand, we also have
$\frac{c_m(\lambda_n)\pm \epsilon}{p+\mu n\mp \barrier(n)} = \frac{c_m(\lambda_n)\pm \epsilon}{p+\mu n} (1+o(1))$ uniformly on $[0,tp]$, for any fixed $t>0$. Hence,
\begin{equation*}
\sum_{n=0}^{tp} \frac{c_m(\lambda_n)\pm \epsilon}{p+\mu n\mp \barrier(n)}
\ \cv[]p  \int_0^{t} \frac{c_m\left(\frac{\lambda+\mu s}{1+\mu s}\right)\pm\epsilon}{1+\mu s}\dd s=\pm \frac{\epsilon}{\mu}\log(1+\mu t)+\int_0^{t} \frac{c_m\left(\frac{\lambda+\mu s}{1+\mu s}\right)}{1+\mu s}\dd s\,.
\end{equation*}
Above, we also used the fact that $c_m(\lambda)$ is continuous in $\lambda$, which follows directly from its definition and is also seen below via an explicit expression. We also have
\begin{equation*}
\sum_{n=0}^{tp} \left(\frac{c_m(\lambda_n)+\epsilon}{p+\mu n- \barrier(n)}\right)^2 \cv[]p 0
\end{equation*} for all $t>0$. Combining this with the last three displays, we conclude that
\begin{align}\label{eq:Tmscalingbounds}
&(1+\mu t)^{- \frac{\epsilon}\mu}\exp\left(-\int_0^{t} \frac{c_m\left(\frac{\lambda+\mu s}{1+\mu s}\right)}{1+\mu s}\dd s\right) - \limsupp \Prob_{p,q}(\anom)
	\ \le\	\liminf_{p,q\to\infty} \Prob_{p,q}(T_m>tp) \notag
\\&	\ \le\	\limsup_{p,q\to\infty} \Prob_{p,q}(T_m>tp)
	\ \le\	(1+\mu t)^{\frac{\epsilon}\mu}\exp\left(-\int_0^{t} \frac{c_m\left(\frac{\lambda+\mu s}{1+\mu s}\right)}{1+\mu s}\dd s\right)
			+ \limsupp \Prob_{p,q}(\anom) \,.
\end{align}

Now take the limit $m,x\to\infty$.
First, using the data of Table~\ref{tab:prob(p,q)diag}, we observe that the sequence $(c_m(\lambda))_{m\ge0}$ is increasing with a finite limit: \begin{align}\label{eq:defcinfty}
c_m(\lambda)\, =&\
\lim_{p,q\to\infty} \m({ p\cdot \Prob_{p,q}(P_1\wedge Q_1\leq m)}
\ =\ \lim_{p,q\to\infty} p\cdot \sum_{k=1}^m
   \m({ \Prob_{p,q}(\rp[p-k+1],\rn[p-k])+\Prob_{p,q}(\lp[q-k-1],\lm[q-k]) }
\notag \\
=&\ -\frac43 \frac{t_c}{b \cdot \lambda^{7/3} c(\lambda)}
            \sum_{k=1}^m(1+\nu_c)\m({ \frac{a_0}{u_c}+a_1 } a_k u_c^k
\notag \\ \cv[]m &\
-\frac43 \frac{t_c}{b \cdot \lambda^{7/3} c(\lambda)}
         (1+\nu_c) \m({ \frac{a_0}{u_c}+a_1 } (A(u_c)-a_0)
\ =:\ c_\infty(\lambda).
\end{align}
Furthermore, we notice that $(1+\nu_c)\left(\frac{a_0}{u_c}+a_1\right)(A(u_c)-a_0)=-b\mu$, a computation already done in the proof of \cite[Proposition~11]{CT20}. This gives $c_\infty(\lambda)=\frac{4}{3}\frac{\mu}{c(\lambda)\lambda^{7/3}}.$
Moreover, in the limit $m,x\to\infty$, the error term $\limsup\Prob_{p,q} (\anom)$ tends to zero due to Lemma~\ref{lem:one jump diag}. The middle terms $\liminf_{p,q\to\infty}\Prob_{p,q}(T_m>tp)$ and $\limsup_{p,q\to\infty}\Prob_{p,q}(T_m >tp)$ do not depend on $m$ due to the convergence $\Prob_{p,q}(T_0-T_m>\epsilon p) \cv[]{p,q} 0$ seen at the beginning of the proof. Thus by sending $m\to\infty$ and $\epsilon\to 0$, the monotone convergence theorem finally yields \begin{equation*}
\lim_{p,q\to\infty} \Prob_{p,q}\m({T_m>tp} = \exp\left(-\int_0^t c_\infty\left(\frac{\lambda+\mu s}{1+\mu s}\right)\frac{\dd s}{1+\mu s}\right).
\end{equation*}
Now recall that
\begin{equation*}c(\lambda)=\frac{4}{3}\int_0^\infty(1+s)^{-7/3}(\lambda+s)^{-7/3}ds.
\end{equation*}
We note first that
\begin{equation*}
c\left(\frac{\lambda+x}{1+x}\right)=\frac{4}{3}(1+x)^{11/3}\int_x^\infty(1+s)^{-7/3}(\lambda+s)^{-7/3} \dd s.
\end{equation*}
This yields
\begin{align*}
& \od{}{x} \m({ \int_0^{\mu^{-1}x} c_\infty \m({ \frac{\lambda+\mu s}{1+\mu s} } \frac{\dd s}{1+\mu s} }
=\frac1\mu \cdot c_\infty\m({ \frac{\lambda+x}{1+x} } \cdot\frac{1}{1+x}
\\ =\ &(1+x)^{-7/3}(\lambda+x)^{-7/3}\left(\int_x^\infty(1+s)^{-7/3}(\lambda+s)^{-7/3}ds\right)^{-1}
=-\od{}{x} \log\int_x^\infty(1+s)^{-7/3}(\lambda+s)^{-7/3}\dd s.
\end{align*}
Finally, integrating this equation on each of the sides gives the claim.
\end{proof}

In order to prove the diagonal local convergence in its full generality as Theorem~\ref{thm:cv} suggests, we also show the following generalized bounds:

\begin{prop}\label{prop:scalingTm}
For all $m\in\natural$, the scaling limit of the jump time $T_m$ has the following bounds:
\begin{equation*}
\forall t>0\,,\qquad	\liminf_{p,q\to\infty} \Prob_{p,q}\m({T_m>tp} \ge \exp\left(-\int_0^t\max_{\ell\in[\lambda',\lambda]}c_\infty\left(\frac{\ell+\mu s}{1+\mu s}\right)\cdot\frac{ds}{1+\mu s}\right)
\end{equation*} and
\begin{equation*}
\limsup_{p,q\to\infty} \Prob_{p,q}\m({T_m>tp} \le \exp\left(-\int_0^t\min_{\ell\in[\lambda',\lambda]}c_\infty\left(\frac{\ell+\mu s}{1+\mu s}\right)\cdot\frac{ds}{1+\mu s}\right)
\end{equation*}
where $c_\infty$ is defined as in \eqref{eq:defcinfty} and the limit is taken such that $q/p\in [\lambda',\lambda]$.
\end{prop}

\begin{proof}
We modify the above proof as follows: First, \eqref{eq:lambdaestimate} translates to
\begin{equation*}
\frac{\lambda' p+\mu n-x f_\epsilon(n)}{p+\mu n+x f_\epsilon(n)}\le\frac{Q_n}{P_n}\le\frac{\lambda p+\mu n+x f_\epsilon(n)}{p+\mu n-x f_\epsilon(n)}
\end{equation*}
conditional on $\nom_n$. Then, the identity $Q_n/P_n:=\lambda_n=\frac{\lambda p+\mu n}{p+\mu n}\left(1+o(1)\right)$ is to be replaced by the bounds
\begin{equation*}
\frac{\lambda' p+\mu n}{p+\mu n}\left(1+o(1)\right)\le\lambda_n\le \frac{\lambda p+\mu n}{p+\mu n}\left(1+o(1)\right).
\end{equation*}
Finally, we notice that $\lambda\mapsto c_m(\lambda)$ is a continuous function for every $m\ge 0$ on any compact strictly positive interval, having the limit $c_\infty(\lambda)$ as $m\to\infty$ with the same property. Therefore, we can replace $c_m\left(\frac{\lambda+\mu s}{1+\mu s}\right)$ in \eqref{eq:Tmscalingbounds} by its minimum or maximum over the interval $[\lambda',\lambda]$, respectively, and finally take the limit $m\to\infty$.
\end{proof}

The limit law
\begin{equation}\label{scalinglimitlaw}
\Prob(L>t):=\int_t^\infty(1+x)^{-7/3}(\lambda+x)^{-7/3}dx
\end{equation}
can be interpreted as the law of the quantum length of an interface resulted from the conformal welding of two quantum disks in the Liouville Quantum Gravity of parameter $\gamma=\sqrt{3}$, introduced in the context of the mating of the trees theory in \cite{matingoftrees} and studied in \cite{AHS20}. More precisely, this measure results from a welding of two independent quantum disks of parameter $\gamma=\sqrt{3}$ and weight $2$ along a boundary segment of length $L$. See \cite{matingoftrees,AG19,AHS20} for precise definitions of such quantum disks. In particular, a quantum disk conditioned to have a fixed boundary length is well-defined.

As defined in \cite{AG19}, an $(R,L)$-length quantum disk $(D,x,y)$ is a quantum disk decorated with two marked boundary points $x,y$, which is sampled in the following way: First, a quantum disk $D$ of a fixed boundary length $R+L$ is sampled. Then, conditional on $D$, the boundary point $x$ is sampled from the quantum boundary length measure. Finally, define $y$ to be the boundary point of $D$ such that the counterclockwise boundary arc from $x$ to $y$ has length $R$. By giving the quantum disk an additional weight parameter and setting its value to $2$, the points $x$ and $y$ can in fact be sampled independently from the LQG boundary length measure, as explained in \cite{AHS20}.

For two independent $\sqrt{3}$-quantum disks, there is a natural perimeter measure on $(0,\infty)^2$, given by
\begin{equation}\label{perimeterlaw}
dm(u,v)=u^{-7/3}v^{-7/3}du dv.
\end{equation}
This measure is the Lévy measure of a pair of independent spectrally positive $4/3$-stable Lévy processes, which has a direct connection to the jumps of the boundary length processes of SLE($16/3$). On the other hand, it is known that the typical disks swallowed by the SLE($16/3$) are $\sqrt{3}$-quantum disks; see \cite{MSW20}. This perimeter measure allows us to randomize the boundary arc lengths of the quantum disks as follows.

Due to the convergence $q/p\to\lambda\in(0,\infty)$ (and $p/p\to 1$) in our discrete picture, we consider the measure \eqref{perimeterlaw} conditional on the set $\{(u,v): u=1+L,\ v=\lambda+L,\ L>0\}$, such that the two independent quantum disks have perimeters $(\lambda,L)$ and $(L,1)$, respectively. This gives rise to the law of the segment $L$ as
\begin{equation*}
\Prob(L\in dx)=\mathcal{N}^{-1}(1+x)^{-7/3}(\lambda+x)^{-7/3}dx
\end{equation*}
where $\mathcal{N}$ is a normalizing constant in order to yield a probability distribution. Gluing the two quantum disks along the boundary segment of length $L$ such that the marked boundary points of the two disks coincide to the points $\rho$ and $\rho^\dagger$, respectively, finally yields \eqref{scalinglimitlaw} as the law of the interface length.

The same law of $L$ has been recently derived in \cite[Remark~2.7]{AHS20} as a special case of the general conformal welding of quantum disks. Since the parameters there also match with the expected ones for the universality class of the critical Ising model, this gives some hints that the Ising interfaces should indeed converge towards an SLE$(3)$-curve on a LQG surface, as predicted in the literature. This convergence remains as an important open problem.

\newcommand{\ribbon}{\mathcal{R}}
\newcommand*{\rmap}[1][m]{\mathcal{R}_{#1}}
\newcommand*{\uleft}[1][T_m]{\umap_{#1}}
\newcommand*{\uright}[1][T_m]{\umap^*_{#1}}
\newcommand{\kk}{\mathcal{K}_m}
\newcommand{\pjump}{\mathcal{P}}
\newcommand{\qjump}{\mathcal{Q}}
\newcommand{\pleft}{\mathcal{P}}
\newcommand{\qleft}{\mathcal{Q}}
\newcommand{\pright}{\mathcal{P}^*}
\newcommand{\qright}{\mathcal{Q}^*}

\begin{figure}
\begin{center}
\includegraphics[scale=0.8]{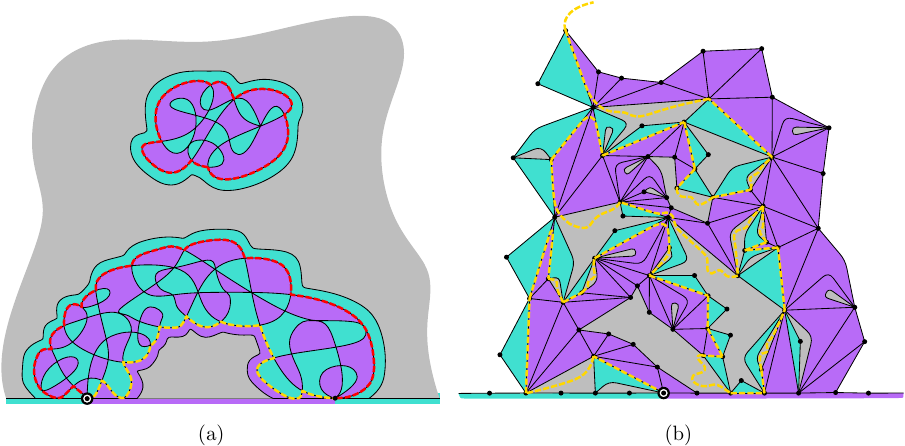}
\end{center}\vspace{-1em}
\caption{Left: a glimpse of a realization of the local limit at $\nu=\nu_c$ after just $q\to\infty$, which is considered in \cite{CT20}. The interface is still finite. In this work, we bypass this intermediate limit by letting $p,q\to\infty$ in a diagonal regime. Right: a realization of the ribbon containing the interface in the local limit at $\nu=\nu_c$.}
\label{fig:ribbon}
\end{figure}

\subsection{The local convergence in the diagonal regime}

We recall first the definition of the local limit $\prob_\infty=\prob_\infty^{\nu_c}$ (see \cite[Section 5.3]{CT20}). The probability measure $\prob_\infty$ is defined as the law of a random triangulation of the half-plane which is obtained as a gluing of three infinite, mutually independent, one-ended triangulations $\law_\infty\umap_\infty$, $\law_\infty\ribbon_\infty$ and $\law_\infty\umap_\infty^*$ along their boundaries, which satisfy the following properties: $\law_\infty\umap_\infty$ has the law $\prob_0$, $\law_\infty\umap_\infty^*$ has the law $\prob_0$ under the inversion of the spins, and $\law_\infty\ribbon_\infty$ is defined as the law of the increasing sequence $(\lim_{n\to\infty}\law_\infty[\emapo_n]_r)_{r\ge 0}$ under $\Prob_\infty$. 
We call $\ribbon_\infty$ the \emph{ribbon}. See \cite{CT20} for a more detailed study and Figure~\ref{fig:ribbon} for an illustration.
The fact that the ball $[\emapo_n]_r$ stabilizes in a finite time, and thus the limit $n\to\infty$ is well-defined, follows from the positive drift of the perimeter processes. Observe that the boundary of $\law_\infty\ribbon_\infty$ consists of three arcs: a finite one consisting of edges of $\partial\emapo_0$ only, and two infinite arcs of spins \< and \+, respectively. The gluing is performed along the infinite boundary arcs such that the spins match with the boundary spins of $\umap_\infty$ and $\umap^*_\infty$, respectively.

Then, fix $m\ge 0$, and define $\ribbon_m$ as the union of the explored map $\emapo_{T_m-1}$ and the triangle explored at $T_m$. Now the triple $(\umap_{T_m},\ribbon_m,\umap_{T_m}^*)$ partitions a triangulation under $\prob_{p,q}$, such that $\umap_{T_m}$ and $\umap_{T_m}^*$ correspond to the two parts separated by the triangle at $T_m$. We will reroot the unexplored maps $\uleft$ and $\uright$ at the vertices $\rho_\umap$ and $\rho_{\umap^*}$, which are the unique vertices shared by $\umap_{T_m}$ and $\ribbon_m$, and $\umap_{T_m}^*$ and $\ribbon_m$, respectively. Now the boundary condition of $\uleft$ is denoted by $(\pleft,(\qleft_1,\qleft_2))$. This notation is in line with \cite[Theorem~4]{CT20}. Similarly, the boundary condition of $\uright$ is $((\pright_1,\pright_2),\qright)$. Observe that the condition
\begin{equation*}
\Step_{T_m} \in \{ \rp[P_{T_m-1}+\kk], \rn[P_{T_m-1} +\kk] \}
\end{equation*}
uniquely defines an integer $\kk$, which represents the position relative to $\rho^\dagger$ of the vertex where the triangle revealed at time $T_m$ touches the boundary. Here, we make the convention that $\rp[p+k]=\lp[q-k-1]$ and $\rn[p+k]=\lm[q-k-1]$ for $k\ge 0$. In particular, $|\kk|\leq m$. See also \cite[Figure 12]{CT20} for a similar setting when $q=\infty$.

\newcommand{\PrE}[4]{\prob#1 \mb({ ([#2]_r,[#3]_r,[#4]_r) \in \mathcal{E} }}

\begin{lemma}[Joint convergence before gluing]\label{lem:loc cv on big jump}
Fix $\epsilon,x,m>0$, and let $\mathcal{J} \equiv \mathcal{J}^\epsilon_{x,m} := \{\tauxy = T_m \ge \epsilon p\}$. Then for any $r\ge 0$,
\begin{equation*}\label{eq:loc cv on big jump}
\begin{aligned}
& \limsup_{p,q\to\infty} \abs{ \PrE\pqy{\rmap}{\uleft}{\uright}
               - \PrE\yy{\rmap[\infty]}{\uleft[\infty]}{\uright[\infty]} }
\\   \le \ &
\limsup_{p,q\to\infty} \prob_{p,q} (\mathcal{J}^c) + \prob_\infty(\tauxy<\infty)
\end{aligned}
\end{equation*}
where $\mathcal{E}$ is any set of triples of balls.
\end{lemma}

\begin{proof}
The proof applies the idea of the proof of \cite[Lemma~14]{CT20}. Assuming that known, the only thing one needs to take care of is the fact that the random numbers $\pright_1$, $\pright_2$, $\qleft_1$ and $\qleft_2$ tend to $\infty$ uniformly, and that $\pleft$ and $\qright$ stay bounded, conditional on $\mathcal{J}$. Observe also that the random number $\kk$ is automatically bounded in this setting, so we do not need any condition for $\kk$ on the event $\mathcal{J}$.

Similarly as in \cite{CT20}, we have the lower bounds $\pright_1\geq\mu(\epsilon p-1)-xf_\epsilon(\epsilon p-1)=:\underline{\pright_1}$ and $\pright_2\ge p+\min_{n\ge 0}(\mu n-x f_\epsilon(n))-1-m=:\underline{\pright_2}$  as well as the upper bound $\qright\le m+1$ for the boundary condition of $\uright$.
For completeness and convenience, let us show similar bounds for the boundary condition of $\uleft$. Let $S^\+$ and $S^\<$ be the distances from $\rho$ to $\rho_{\umap^*}$ and $\rho_\umap$ along the boundary, respectively. First, expressing the total perimeter of $\uleft$, the number of edges between $\rho$ and $\rho^\dagger$ clockwise and the number of \+ edges on the boundary of $\uleft$, respectively, we find the equations
\begin{equation*}
\begin{cases}
\qleft_1+\qleft_2+\pleft=Q_{T_m-1}-\kk \\
S^{\<}+\qleft_2+\max\{0,\kk\}=q \\
\pleft=\delta-\min\{0,\kk\}
\end{cases}
\end{equation*}
where $\delta=1$ if $\Step_{T_m} =\rp[P_{T_m-1}+\kk]$, and otherwise $\delta=0$, as well as $S^{\<}$ is the number of $\<$ edges on $\ribbon_m\cap\partial\emapo_0$. The solution of this system of equations is
\begin{equation*}
\begin{cases}
\qleft_1=Y_{T_m-1}+S^{\<}-\delta \\
\qleft_2=q-S^{\<}-\max\{0,\kk\} \\
\pleft=\delta-\min\{0,\kk\}
\end{cases}.
\end{equation*}
We have $S^{\<}\in [0,1-\min_{n\ge 0}(\mu n-x f_\epsilon(n))]$, and the function $n\mapsto\mu n-x f_\epsilon(n)$ is increasing. Therefore, we deduce $\qleft_1\ge\mu(\epsilon p-1)-x f_\epsilon(\epsilon p-1)-2=:\underline{\qleft_1}$ and $\qleft_2\ge q+\min_{n\ge 0}(\mu n-x f_\epsilon(n))-1-m=:\underline{\qleft_2}$, together with $\pleft\le m+1$. The claim follows.
\end{proof}

\newcommand{\pglued}[1]{\prob#1 \mb({ [\mop]_r \in \mathcal{E} }}
\newcommand{\mop}{\tmap \oplus \tmap'}

\begin{proof}[Proof of the convergence $\prob\pqy^{\nu_c} \to \prob\yy^{\nu_c}$.]
The triangulation $\law\pqy \bt$ (respectively, $\law\yy \bt$) can be represented as the gluing the triple $\law\pqy (\rmap,\uleft,\uright)$ (respectively, $\law\yy (\rmap[\infty],\uleft[\infty],\uright[\infty])$) along their boundaries. This is done pairwise between the three components, taking into account that the location of the root vertex changes during this procedure. Given a triangulation $\tmap$ with a simple boundary, and an integer $S$, let us denote by $\overrightarrow \tmap^{S}$ (resp. $\overleftarrow \tmap^{S}$) the map obtained by translating the root vertex of $\tmap$ by a distance $S$ to the right (resp. to the left) along the boundary. Denote by $\rho$ and $\rho'$ the root vertices of two triangulations $\tmap$ and $\tmap'$, respectively, and let $L$ be the number of edges in $\tmap$ and $\tmap'$ which are admissible for the gluing. More precisely, we assume that $L$ is a random variable taking positive integer or infinite values, such that
\begin{equation}\label{eq:gluing length}
\law\pqy L \cv[]{p,q} \infty \text{ in distribution and }\law\yy L = \infty \text{ almost surely.}
\end{equation}
Finally, let $\mop$ be the triangulation obtained by gluing the $L$ boundary edges of $\tmap$ on the right of $\rho$ to the $L$ boundary edges of $\tmap'$ on the left of $\rho'$. The dependence on $L$ is omitted from this notation because the local limit of $\mop$ is not affected by the precise value of $L$, provided that \eqref{eq:gluing length} holds.

Now under $\prob\pqy$, we have
\begin{equation}\label{eq:gluing-of-3}
\bt = \overrightarrow{(\umap\rmap)}^{S^\+ + S^\<} \oplus \uright \qtq{where} \umap\rmap = \uleft \oplus \overleftarrow{(\rmap)}^{S^\<}
\end{equation}
where we recall that $S^\+$ and $S^\<$ are the distances from $\rho$ to $\rho_{\umap^*}$ and $\rho_\umap$ along the boundary, respectively. Similarly, $\law\yy \bt$ can be expressed in terms of $\uleft[\infty]$, $\rmap[\infty]$, $\uright[\infty]$ and $S^\jj$ using gluing and root translation.

On the event $\mathcal{J}$, the perimeter processes $\nseq X$ and $\nseq Y$ stay above $\mu n-xf_\epsilon(n)$ up to the time $\tauxy$. Thus their minima over $[0,\tauxy)$ are reached before the deterministic time $N_{\min} = \sup\Set{n\ge 0}{\mu n-xf_\epsilon(n)\le 0}$, and $S^\+$ and $S^\<$ are measurable functions of the explored map $\emapo_{N_{\min}}$. It follows that $\law\pqy S^\jj$ converges in distribution to $\law\yy S^\jj$ on the event $\mathcal{J}$. Using the relation \eqref{eq:gluing-of-3} together with \cite[Lemmas~15-16]{CT20}, we deduce from Lemma~\ref{lem:loc cv on big jump} that for any $x,m,\epsilon>0$, and for any $r\ge 0$ and any set $\mathcal{E}$ of balls, we have
\begin{equation*}
\limsupp \mb|{\, \prob\pqy ( \btsq_r \in \mathcal{E} )
               - \prob\yy( \btsq_r \in \mathcal{E} ) \, }    \ \le\
\limsupp \prob\pqy (\mathcal{J}^c) + \prob\yy (\tauxy<\infty)    \,.
\end{equation*}

The left hand side does not depend on the parameters $x,m$ and $\epsilon$. Therefore to conclude that $\prob\pqy$ converges locally to $\prob\yy$, it suffices to prove that
$\displaystyle \limsupp \prob\pqy (\mathcal{J}^c) + \prob\yy (\tauxy<\infty)$\, converges to zero when $x,m\to\infty$ and $\epsilon\to 0$. The latter term converges to zero, since if $x\to\infty$, we have $\tauxy \to\infty$ almost surely under $\prob\yy$. For the first term, a union bound gives
\begin{equation*}
\prob\pqy(\mathcal{J}^c) \ \le\ \prob\pqy (\tauxy< T_m) + \prob\pqy(T_m<\epsilon p) \,,
\end{equation*}
where the first term on the right can be bounded using Lemma~\ref{lem:one jump diag}:
\begin{equation*}
\lim_{m,x \to\infty}  \limsupp \prob\pqy(\tauxy<T_m)    \ =\ 0 .
\end{equation*}
For the last term, we use the lower bound of Proposition~\ref{prop:scalingTm}:
\begin{equation*}
\lim_{\epsilon\to 0}\  \limsupp \prob\pqy(T_m<\epsilon p)\ \leq\
1-\lim_{\epsilon\to 0}\exp\left(-\int_0^\epsilon \max_{\ell\in[\lambda',\lambda]}c_\infty\left(\frac{\ell+\mu s}{1+\mu s}\right)\frac{ds}{1+\mu s}\right) \ =\ 0    \,.\qedhere
\end{equation*}
\end{proof}

\appendix

\section{A one-jump lemma for the process $\law_{p,q}\nseq{X_n,Y}$ at $\nu=\nu_c$}\label{sec:bigjumplemma proof}

\newcommand*{\ea}{\asymp}
\newcommand{\cst}{\mathrm{cst}}
\newcommand{\pp}[2]{\mathfrak{p}_{#1,#2}}
\newcommand{\pxx}[1][k]{\mathfrak{p}^x_{#1}}
\newcommand{\pyy}[1][k]{\mathfrak{p}^y_{#1}}
\newcommand{\Py}{\Prob\py{}}
\newcommand{\PY}{\Prob\yy{}}
\newcommand{\Ey}{\EE\py{}}
\newcommand{\EY}{\EE\yy{}}

The proof is mostly a modification of a similar proof \cite[Appendix C]{CT20}. Here, we need to take care that both $X_n$ and $Y_n$ stay close to their asymptotic mean for $n<T_m$ with high probability, as $p,q\to\infty$ with $q/p\in[\lambda',\lambda]$, where $0<\lambda'\leq 1\leq\lambda<\infty$. We follow the exposition and the notation of \cite{CT20}.

For starters, we write
\begin{equation*}
\pp k{k'} = \Prob_\infty ( -(X_1,Y_1) = (k,k') )    \qtq{and}
\pxx = \Prob_\infty (-X_1=k)    \ ,\quad
\pyy = \Prob_\infty (-Y_1=k)  .
\end{equation*}
Then, for $k\leq p-2$ and $k'\leq q-1$, the basic relation for the comparison of the laws of the perimeter processes reads
\begin{equation}\label{eq:Doob}
\Prob_{p,q}(-(X_1,Y_1)=(k,k'))=\frac{z_{p-k,q-k'}}{z_{p,q}u_c^{k+k'}}\pp k{k'}=\frac{z_{p-k,q-k'}u_c^{(p+q)-(k+k')}}{z_{p,q}u_c^{p+q}}\pp k{k'},
\end{equation}
as easily verified using the data of Table~\ref{tab:prob(p,q)diag}. Observe that this condition is reminiscent of the Doob $h$-transform of a random walk, ceased to satisfy it since the condition \eqref{eq:Doob} breaks down for $k>p-2$ or $k'>q-1$. We also introduce the following notation: If $A$ and $B$ are two \emph{positive} functions defined on some set $\Lambda$, we say that
\begin{itemize}
\item $A(y) \la B(y)$ for $y \in \Lambda$, \emph{if}
there exists $C>0$ such that $A(y)\le CB(y)$ for all $y \in\Lambda$;
\item $A(y) \ea B(y)$ for $y \in \Lambda$, \emph{if}
$A(y) \la B(y)$	and $B(y) \la A(y)$.
\end{itemize}

We fix a cutoff $\theta \in (0,1)$ and let $p_\theta := \frac2{1-\theta}$ so that $\theta p \le p-2$ and $\theta q \le q-1$ for all $p,q\ge p_\theta$. The following lemma gives estimates for the jump probabilities of the perimeter processes in a single peeling step:

\newcommand{\Xp}{\{-X_1\le p-2\}}
\newcommand{\Yp}{\{-Y_1\le q-1\}}
\newcommand{\ykXp}[1][=k]{\{-Y_1#1\} \cap \Xp}
\newcommand{\xkYp}[1][=k]{\{-X_1#1\} \cap \Yp}
\newcommand{\zfrac}[1][k]{\frac{z_{p-#1,q-k'} u_c^{p+q-(#1+k')}}{z_{p,q} u_c^{p+q}}}

\newcommand{\Ap}{\mathcal{A}_x}
\newcommand{\zg}[1][h]{\{W \ge #1\} \cap \Ap}
\newcommand{\zl}[1][h]{\{W \le #1\} \cap \Ap}
\newcommand{\idzg}[1][h]{\id_{\zg[#1]}}
\newcommand{\idzl}[1][h]{\id_{\zl[#1]}}

\begin{lemma}\label{lem:estimates}Assume throughout the lemma that $q/p$ lies in a fixed compact interval $I\subset\R_+$ such that $0\notin I$. Then the perimeter increments during the first peeling step satisfy the following probability estimates:
\begin{enumerate}[label=(\roman*)]
\item\label{item:estimate infty}
$\pxx \ea \pyy \ea k^{-7/3}$ for $k\ge 1$.
\item\label{item:estimate X}
$\Prob_{p,q}(\xkYp) \ea k^{-7/3}$ and $\Prob_{p,q}(-X_1=p-k) \ea p^{-1} k^{-4/3}$ for all $p,q\ge p_\theta$ and $1\le k\le \theta p$.
\item\label{item:estimate Y}
$\Prob_{p,q}(\ykXp) \ea k^{-7/3}$ and $\Prob_{p,q}(-Y_1=q-k) \ea p^{-1} k^{-4/3}$ for $p,q\ge p_\theta$ and $1\le k\le \theta q$.
\item\label{item:estimate zfrac}
$\abs{\zfrac -1} \la p^{-1} \abs{k} + p^{-1/3}$ and $\abs{\zfrac -1} \la p^{-1} \abs{k'} + p^{-1/3}$ for any $(k,p)$, $(k',q)$ such that $-2\le k \le \theta p$ and $-2\le k' \le \theta q$.

\item\label{item:estimate large jump}
For $p,q\ge p_\theta$, $x\in [1,\theta (p\wedge q)]$ and $m\ge 1$,
\begin{equation*}
\Prob_{p,q}(-X_1 < p-m, \ -Y_1< q-m \text{ and } (-X_1)\vee (-Y_1)\ge x)
\ \la \   x^{-4/3} + p^{-1} m^{-1/3}   \,.
\end{equation*}

In the following, let $\Ap=\{ (-X_1) \vee (-Y_1) \le x\}$ and $W$ be either $\mu-X_1$ or $\mu-Y_1$.

\item\label{item:estimate S}
$\Prob_{p,q}(\zg) \la h^{-4/3}$, $\EE_{p,q}[W \idzg] \la h^{-1/3}$ and $\EE_{p,q}[W^2 \idzl] \la h^{2/3}$ for $p,q\ge p_\theta$, $x\in [1,\theta (p\wedge q)]$ and $h\in [1,x]$.

\item\label{item:estimate E}
$|\EE_{p,q}[W  \id_{\Ap}]| \la x^{-1/3}$ for $p,q \ge p_\theta$ and $x\in [1,\theta (p\wedge q)]$.

\item\label{item:estimate small jump}
For $p,q\ge p_\theta$, $x\in [1,\theta (p\wedge q)]$ and $\xi \in [2x^{-1},1]$,
\begin{equation*}
\log \m({ \EE_{p,q}[e^{\pm \xi W} \id_{\Ap}] }
\ \la\  x^{-4/3} e^{\xi x}.
\end{equation*}
\end{enumerate}
\end{lemma}

\newcommand{\varsubset}{\subset}
\renewcommand{\subset}{\subseteq}

\newcommand{\absb}[1]{\mb|{#1}}
\newcommand{\absB}[1]{\mB|{#1}}
\newcommand{\absh}[1]{\mh|{#1}}
\newcommand{\absH}[1]{\mH|{#1}}

\begin{proof}
\begin{enumerate}[label=\textbf{(\roman*)},wide=0pt,listparindent=\parindent,]
\item
Proven in \cite{CT20}.

\item
First, since $\Prob_{p,q}(\{-X_1=1\}\cap\Yp)$ has a finite limit as $p,q\to\infty$ while $q/p\in I$, we have $\Prob_{p,q}(\{-X_1=1\}\cap\Yp)\ea 1$.
Then for $2\leq k\leq\theta p$, we write
\begin{equation*}
\Prob_{p,q}(\xkYp)=\frac{z_{p-k,q}}{z_{p,q}u_c^k}\pp k{0}+\frac{z_{p-k,q-1}}{z_{p,q}u_c^k}\pp k{1}.
\end{equation*}
Since $k\leq\theta p$, the asymptotics of Equation~\eqref{eq:diagasympt} yield $\frac{z_{p-k,q}}{z_{p,q}u_c^k}\ea 1$ and $\frac{z_{p-k,q-1}}{z_{p,q}u_c^k}\ea 1$. The first estimate follows then by \textbf{(i)}. For the second estimate, we note that
\begin{equation*}
\Prob_{p,q}(-X_1=p-k)=t_c\frac{z_{k,q-1}z_{p-k+2,0}}{z_{p,q}}+t_c\nu_c\frac{z_{k,q}z_{p-k,1}}{z_{p,q}}\sim C(p,q)\cdot a_k u_c^k\cdot\left((p-k+2)^{-\frac{7}{3}}+(p-k)^{-\frac{7}{3}}\right)\cdot p^{\frac{4}{3}}
\end{equation*}
where $C(p,q)$ is a bounded constant depending on $p,q$ and bounded away from zero. Since $a_k u_c^k\ea k^{-4/3}$ as well as $(p-k+2)^{-7/3}\ea p^{-7/3}\ea (p-k)^{-7/3}$, the desired result follows.

\item
Since $\Prob_{p,q}(\{-Y_1=1\}\cap\Xp)$ has a finite limit, $\Prob_{p,q}(\{-Y_1=1\}\cap\Xp)\ea 1$. Then, assume $2\leq k\leq\theta q$, in which case
\begin{equation*}
\Prob_{p,q}(\ykXp)=\frac{z_{p,q-k}}{z_{p,q}u_c^k}\pp 0{k}+\frac{z_{p+1,q-k}}{z_{p,q}u_c^k}\pp {-1}{k}
\end{equation*}
Since $k\leq\theta q$, the asymptotics of Equation~\eqref{eq:diagasympt} yield $\frac{z_{p,q-k}}{z_{p,q}u_c^k}\ea 1$ and $\frac{z_{p+1,q-k}}{z_{p,q}u_c^k}\ea 1$. The first estimate follows. Secondly, for $k=2,\dots,\theta q$,
\begin{equation*}
\Prob_{p,q}(-Y_1=q-k)=t_c\frac{z_{p+1,k}z_{1,q-k-1}}{z_{p,q}}+t_c\nu_c\frac{z_{p,k}z_{0,q-k+1}}{z_{p,q}}\sim \tilde{C}(p,q)\cdot a_k u_c^k\cdot\left((q-k+1)^{-\frac{7}{3}}+(q-k-1)^{-\frac{7}{3}}\right)\cdot p^{\frac{4}{3}}
\end{equation*}
where $\tilde{C}(p,q)$ is bounded and bounded away from zero. The result follows since $(q-k+1)^{-7/3}\ea q^{-7/3}\ea (q-k-1)^{-7/3}$ and $q\ea p$.

\item
From the asymptotic expansion $z_{p,q}u_c^{p+q}=\frac{b\cdot c(q/p)}{\Gamma(-4/3)\Gamma(-1/3)}p^{-11/3}\left(1+O(p^{-1/3})\right)$, we see that there exist constants $C=C(\theta)$ and $p_0=p_0(\theta)$ such that for all $p,q\geq p_0$, $-2\leq k\leq\theta p$ and $-2\leq k'\leq\theta q$,
\begin{equation*}
\frac{(p-k)^{-11/3}}{p^{-11/3}}\left(1-C p^{-1/3}\right)\leq\frac{z_{p-k,q-k'}}{z_{p,q}u_c^{k+k'}}\leq\frac{(p-k)^{-11/3}}{p^{-11/3}}\left(1+C p^{-1/3}\right).
\end{equation*}
After writing down the Taylor expansions of each of the sides of the inequality, the rest of the proof of the first estimate goes similarly as the proof of a corresponding claim in \cite{CT20}. The second estimate follows after swapping the roles of $p$ and $q$, and $k$ and $k'$, respectively, and noting that $q^{-1} \abs{k'} + q^{-1/3}\ea p^{-1} \abs{k'} + p^{-1/3}.$

\item
We estimate
\begin{align*}
&\ \Prob_{p,q}(-X_1 < p-m, \ -Y_1<q-m \text{ and } (-X_1)\vee (-Y_1)\ge x) \\ \le&\ \Prob_{p,q}\left(\{-X_1\le p-2\}\cap\{\theta q>-Y_1\geq x\}\right)+\Prob_{p,q}\left(\{-Y_1\le q-1\}\cap\{\theta p>-X_1\geq x\}\right) \\ +&\ \Prob_{p,q}(-X_1\in [\theta p,p-m])+\Prob_{p,q}(-Y_1\in[\theta q,q-m]) \\ \la&\ \sum_{k=x}^{\theta q}k^{-7/3}+\sum_{k=x}^{\theta p}k^{-7/3}+\sum_{k=m}^{(1-\theta) p}p^{-1}k^{-4/3}+\sum_{k=m}^{(1-\theta) q}p^{-1}k^{-4/3}\la x^{-4/3}+p^{-1}m^{-1/3},
\end{align*}
where we used the results \textbf{(ii)}-\textbf{(iii)}.

\item
This goes similarly as the proof of a corresponding claim in \cite{CT20}, after one notices that the conditions $x\leq \theta q$ and $h\geq 1$ imply $\{h-\mu\leq -Y_1\leq x\}\subset\{1\leq -Y_1\leq \theta q\}$.

\item
For $k\leq p-2$ and $k'\leq q-2$, Equation~\eqref{eq:Doob} gives the estimate \begin{equation*}
\abs{\Prob_{p,q}\left((-X_1,-Y_1)=(k,k')\right)-\pp k{k'}}\leq\pp k{k'}\abs{\frac{z_{p-k,q-k'}}{z_{p,q}u_c^{k+k'}}-1}.
\end{equation*}
If $W=\mu-X_1$, the equation above then yields
\begin{align*}
    \abs{\EE_{p,q}[W  \id_{\Ap}]-\EE_\infty[W  \id_{\Ap}]}
=&\ \abs{\sum_{k=-2}^x (\mu+k) \m({
      \Prob_{p,q} \m({ -X_1=k, -Y_1\in [-1,x] }
    - \Prob_\infty\m({ -X_1=k, -Y_1\in [-1,x] }
    } }
\\ =&\
    \abs{\sum_{k=-2}^x \sum_{k'=-1}^x(\mu+k) \m({
      \Prob_{p,q} \m({-X_1=k, -Y_1=k'}
    - \Prob_\infty\m({-X_1=k, -Y_1=k'}
    } }
\\ \le&\
\sum_{k=-2}^x \sum_{k'=-1}^x |\mu+k|\abs{\Prob_{p,q}(-X_1=k, -Y_1=k')-\pp k{k'}}
\\ \le&\
\sum_{k=-2}^x\sum_{k'=-1}^x|\mu+k|\pp k{k'}\abs{\frac{z_{p-k,q-k'}}{z_{p,q}u_c^{k+k'}}-1}
\\ \la&\ \sum_{k=-2}^x|\mu+k|(k+3)^{-7/3}\left(p^{-1}|k|+p^{-1/3}\right)\la p^{-1/3}
\end{align*}
where we used the estimates \textbf{(i)} and \textbf{(iv)}. If $W=\mu-Y_1$, symmetry and the second estimate in \textbf{(iv)}
yield the same asymptotic upper bound. The rest of the proof goes like the proof of an analogous claim in \cite[Appendix C]{CT20}.

\item This is proven, \emph{mutatis mutandis}, in \cite[Appendix C]{CT20}. \qedhere
\end{enumerate}
\end{proof}

\newcommand{\tauxx}[1][x]{\tau_{#1}}

Let $\tauxx = \inf \Set{n\ge 0}{|X_n-\mu n| \vee |Y_n-\mu n| > x}$. Then, using the Markov property, we find the following analog of \cite[Lemma~24]{CT20}.

\begin{lemma}\label{lem:cst barrier}
Fix some $\epsilon>0$ and let $x =\chi \mb({ N (\log N)^{1+\epsilon} }^{3/4}$.
Then for any $0<\lambda_{\min}\leq 1\leq\lambda_{\max}<\infty$, $\theta\in(0,\lambda_{\min})$, $p,q\ge \tilde p_\theta:=p_\theta/(1-\theta)$ such that $\frac{q}{p}\in[\lambda_{\min},\lambda_{\max}]$ as well as $m\ge 1$ and $\chi,N \ge 2$ such that $x\in [1,\frac{\theta}{1+\theta} (p\wedge q)]$, we have
\begin{equation*}
\Prob_{p,q} (\tauxx \le N, \tauxx < T_m) \la \frac1{(\log \chi + \log N)^{1+\epsilon/2}} + \frac Np m^{-1/3}\,.
\end{equation*}
\end{lemma}

\begin{proof}

For $n\ge 1$, let $\Delta X_n = X_n-X_{n-1}$ and $\Delta Y_n = Y_n-Y_{n-1}$, and
\begin{equation*}
J_x = \inf\Set{n\ge 1}{ (-\Delta X_n) \vee (-\Delta Y_n) \ge x}.
\end{equation*}  Following the corresponding proof in \cite{CT20}, we bound the probability of the event $\{\tauxx \le N, \tauxx < T_m\}$ both in the cases $\{J_x\le \tauxx \}$ (large jump estimate) and $\{\tauxx < J_x\}$ (small jump estimate).

\medskip
\noindent\textbf{Large jump estimate: union bound.} We have
\begin{equation*}
\Prob_{p,q}( \tauxx \le N, \tauxx < T_m \text{ and } J_x \le \tauxx )\le \sum_{n=1}^N \Prob_{p,q}( n \le \tauxx \text{ and } J_x=n < T_m )    \,.
\end{equation*}
If $n\leq\tauxx$, in particular $P_{n-1}\ge p-x$ and $Q_{n-1}\ge q-x$. Let
\begin{equation}\label{eq:defD}
\mathcal{D}:=\bigg\{(p',q') : p'\ge p-x,\ q'\ge q-x,\ \frac{\lambda_{\min}-\theta}{1+\theta}\leq\frac{q'}{p'}\leq(1+\theta)\lambda_{\max}+\theta \bigg\}.
\end{equation}
Then for $i<\tauxx$, we have the estimates
\begin{equation*}
\frac{Q_i}{P_i}-\lambda_{\max}\leq\frac{q+\mu i+x}{p+\mu i-x}-\lambda_{\max}\leq\frac{(q-\lambda_{\max} p)+\mu(1-\lambda_{\max})i+(1+\lambda_{\max})x}{p-x}\leq \frac{(1+\lambda_{\max})x}{p-x}\leq\theta(1+\lambda_{\max})
\end{equation*}
and
\begin{equation*}
\frac{Q_i}{P_i}-\lambda_{\min}\ge\frac{q+\mu i-x}{p+\mu i+x}-\lambda_{\min}\ge\frac{(q-\lambda_{\min} p)+\mu(1-\lambda_{\min})i-(1+\lambda_{\min})x}{p}\ge-\frac{(1+\lambda_{\min})x}{p}\ge-(1+\lambda_{\min})\frac{\theta}{1+\theta},
\end{equation*}
since by assumption, $\lambda_{\min}p\leq q\leq\lambda_{\max} p$ and $1\leq x\leq \frac{\theta p}{1+\theta}$. In particular, this holds for $i=n-1$, and we conclude that $(P_{n-1},Q_{n-1})\in\mathcal{D}$.

On the other hand, $J_x=n<T_m$ immediately implies $P_n>m$, $Q_n>m$ and $ (-\Delta X_n) \vee (-\Delta Y_n) \ge x$. Thus, the Markov property of $\nseq{P_n,Q}$ gives the upper bounds \begin{align*}
 \Prob_{p,q}( n \le \tauxx \text{ and } J_x=n < T_m ) &\leq \EE_{p,q}\left(\Prob_{P_{n-1}, Q_{n-1}}\left(P_1>m, \ Q_1>m, \ (-X_1) \vee (-Y_1)\ge x\right)\id_{\{P_{n-1}\geq p-x, Q_{n-1}\geq q-x\}}\right) \\ &\leq \sup_{(p',q')\in\mathcal{D}}\Prob_{p',q'}\left(-X_1<p-m, \ -Y_1<q-m, \ (-X_1) \vee (-Y_1)\ge x\right).
\end{align*}
The assumptions $p,q\geq\tilde{p}_\theta=\frac{p_\theta}{1-\theta}$ and $1\leq x\leq\frac{\theta}{1+\theta}(p\wedge q)$ ensure that, for $p'\geq p-x$ and $q'\geq q-x$, the condition $p'\wedge q'\geq p_\theta$ with $1\leq x\leq \theta(p'\wedge q')$ is satisfied. Hence, by Lemma~\ref{lem:estimates} \textbf{(v)},
\begin{equation}\label{eq:large jump estimate}
\Prob_{p,q}(\tauxx \le N, \tauxx < T_, \text{ and } J_x \le \tauxx )\la N\left(x^{-4/3}+p^{-1}m^{-1/3}\right) = \frac{\chi^{-4/3}}{ (\log N)^{1+\epsilon} } + \frac{N}{p} m^{-1/3}.
\end{equation}

\newcommand{\unit}{\mathbf{e}}
\newcommand{\tauu}{\tau^{\unit}_x}

\noindent\textbf{Small jump estimate: Chernoff bound.} For each of the four unit vectors $\unit \in \integer^2$, define
\begin{equation*}
\tauu = \inf\Set{n\ge 0}{ (\mu n-X_n, \mu n-Y_n)\cdot \unit \ge x } \,,
\end{equation*}
so that $\tauxx = \min_{\unit} \tauu$. We start by estimating
\begin{align}
    \Prob_{p,q}( \tauxx \le N, \tauxx < T_m \text{ and } J_x > \tauxx )
\ &\le\ \Prob_{p,q}( \tauxx \le N, \tauxx < J_x ) \notag \\
&\le \sum_{\unit}\Prob_{p,q}\left(\tauxx=\tauu\le N, \tauu < J_x\right) \,.    \label{eq:pre Chernoff}
\end{align}

If $\tauu = n < J_x$, then $(\mu n-X_n, \mu n-Y_n) \cdot \unit = \sum_{i=1}^n (\mu -\Delta X_i, \mu -\Delta Y_i) \cdot \unit \ge x$, and $(-\Delta X_i) \vee (-\Delta Y_i) \le x$ for all $i=1,\ldots,n$. Therefore, applying the Chernoff bound,
\begin{align}
      \Prob_{p,q}(\tauxx=\tauu \le N, \tauu < J_x) &\leq e^{-\xi x} \sum_{n=1}^N \EE_{p,q} \m[{ \idd{\tauxx=\tauu=n} \prod_{i=1}^n
                        e^{\xi (\mu-\Delta X_i, \mu-\Delta Y_i) \cdot \unit}
                        \idd{(-\Delta X_i) \vee (-\Delta Y_i) \le x }          }\notag \\
                        &\leq e^{-\xi x} \sum_{n=1}^N \EE_{p,q} \m[{ \idd{\tauxx=n} \prod_{i=1}^n
                        e^{\xi (\mu-\Delta X_i, \mu-\Delta Y_i) \cdot \unit}
                        \idd{(-\Delta X_i) \vee (-\Delta Y_i) \le x }          }
\label{eq:stopped Chernoff}
\end{align}
for all $\xi\ge0$.
\newcommand{\phix}[1][p,q]{\varphi^{x,\unit}_{#1}(\xi)}
\newcommand{\phixs}[1][p_*,q_*]{\varphi^{x,\unit}_{#1}(\xi)}
\newcommand{\phixp}[1][p',q']{\varphi^{x,\unit}_{#1}(\xi)}
\newcommand{\phixi}[1][P_i,Q_i]{\varphi^{x,\unit}_{#1}(\xi)}

For $p,q\in \natural \cup \{\infty\}$, let $\phix = \EE_{p,q}[ e^{\xi (\mu-X_1, \mu-Y_1) \cdot \unit} \id_{\Ap}]$, where $\Ap = \{(-X_1) \vee (-Y_1) \le x \}$ was already encountered in Lemma~\ref{lem:estimates}. Since the pair $(X_1,Y_1)$ takes only finitely many values on the event $\Ap$ and $\law_{p,q}(X_1,Y_1) \to \law_\infty(X_1,Y_1)$ in distribution, we have $\phix \to \phix[\infty]$ as $p,q \to\infty$ and $q/p\in I$ for any compact interval $I\subset\R_+$ such that $0\notin I$. Recall the set $\mathcal{D}$ defined by \eqref{eq:defD}. Since the peeling process converges in this set when $p',q'\to\infty$, the function $\phixp$ is continuous in its one-point compactification $\mathcal{D}\cup\{(\infty,\infty)\}$. Hence, there exists a pair $(p_*,q_*)=(p_*(x,\unit,\xi), q_*(x,\unit,\xi))\in \mathcal{D}\cup\{(\infty,\infty)\}$ such that
\begin{equation*}
\phixs = \sup_{(p',q')\in \mathcal{D}\cup\{(\infty,\infty)\}} \phixp.
\end{equation*}
Let $\nseq[1]{\Delta X^*_n, \Delta X^*}$ be a sequence of i.i.d.\ random variables independent of $\nseq{X_n,Y}$ and with the same distribution as $\law_{p_*,q_*}(X_1,Y_1)$. Define
\begin{equation*}
(U_i,V_i) = \begin{cases}
			-(\Delta X_i  , \Delta Y_i)   & \text{if }i \le \tauxx
		\\ 	-(\Delta X_i^*, \Delta Y_i^*) & \text{if }i  >  \tauxx.
\end{cases}
\end{equation*}

On the event $\{\tauxx=n\}$, the future $(U_i, V_i)_{i>n}$ of the process is an i.i.d.\ sequence independent of the past such that
$\EE_{p,q}[e^{\xi (\mu+U_i, \mu+V_i) \cdot \unit} \idd{U_i \vee V_i \le x }] = \phixs$. Therefore we can continue the bound \eqref{eq:stopped Chernoff} with
\begin{align}
&\ e^{-\xi x} \sum_{n=1}^N \EE_{p,q} \m[{ \idd{\tauxx=n} \prod_{i=1}^n
    e^{\xi (\mu+U_i, \mu+V_i) \cdot \unit} \idd{U_i \vee V_i \le x }      }
\notag \\ =\ &
e^{-\xi x} \sum_{n=1}^N \mb({\phixs}^{-(N-n)} \EE_{p,q} \m[{
    \idd{\tauxx=n} \prod_{i=1}^N
    e^{\xi (\mu+U_i, \mu+V_i) \cdot \unit} \idd{U_i \vee V_i \le x }      }
\notag \\ \le \ &
e^{-\xi x} \cdot (1\vee \phixs^{-N}) \cdot \EE_{p,q} \m[{ \prod_{i=1}^N
    e^{\xi (\mu+U_i, \mu+V_i) \cdot \unit} \idd{U_i \vee V_i \le x }  }.
\label{eq:completed Chernoff}
\end{align}
Now $\tauxx$ is a stopping time with respect to the natural filtration $\nseq \filtr$ of the process $\nseq{U_n,V}$. Therefore for all $i\ge 0$,
\begin{equation*}
	 \EE_{p,q}\Econd{ e^{\xi (\mu+U_{i+1}, \mu+V_{i+1}) \cdot \unit} \idd{U_{i+1} \vee V_{i+1} \le x} }{\filtr_i}
\ =\		\idd{i< \tauxx} \cdot \phixi + \idd{i\ge\tauxx} \cdot \phixs \le \	\phixs    \,,
\end{equation*}
where we have the last inequality due to the fact that $(P_i,Q_i)\in \mathcal{D}$ on the event $\{i< \tauxx\}$. By expanding the expectation in \eqref{eq:completed Chernoff} with $N$ successive conditioning, we see that it is bounded by $\phixs^N$. Then, combining \eqref{eq:stopped Chernoff} and \eqref{eq:completed Chernoff} yields
\begin{equation*}
\Prob_{p,q}(\tauxx=\tauu \le N, \tauu < J_x)
\ \le\   e^{-\xi x} (\phixs^N \vee 1) \,.
\end{equation*}
By Lemma~\ref{lem:estimates}\ref{item:estimate small jump}, there exists a constant $C$ such that $\phix \le \exp(C x^{-4/3} e^{\xi x})$ for all $p,q\ge p_\theta$, $x\in [1,\theta (p\wedge q)]$, $\xi \in [2 x^{-1},1]$ and unit vector $\unit \in \integer^2$. Note that we already have seen in the derivation of the large jump estimate that the conditions $p'\ge p-x$ and $q'\ge q-x$ imply $p'\wedge q'\geq p_\theta$ and $1\leq x\leq \theta(p'\wedge q')$. Therefore,  we also have $\phixs \le \exp(C x^{-4/3} e^{\xi x})$ by the definition of $\phixs$.
Hence,
\begin{equation*}
\Prob_{p,q}(\tauxx=\tauu \le N, \tauu < J_x)
\ \le\   \exp(-\xi x + C\cdot N x^{-4/3} e^{\xi x}) \,.
\end{equation*}
Plugging this into \eqref{eq:pre Chernoff} and taking $\xi x = c\log \log x$ with $c=1+\epsilon/2$ yields
\begin{equation*}
\Prob_{p,q}( \tauxx \le N, \tauxx < T_m \text{ and } J_x > \tauxx )
\ \le\  4 \exp(-c\log\log x + C N x^{-4/3} (\log x)^c).
\end{equation*}
Thanks to the relation between $x$ and $N$ given in the assumptions, we have $N x^{-4/3} (\log x)^c \ea \chi^{-4/3} \frac{(\log \chi + \log N)^c}{(\log N)^{1+\epsilon}}$, which is bounded by a constant for $\chi,N\ge 2$. It follows that
\begin{equation*}
\Prob_{p,q}( \tauxx \le N, \tauxx < T_m \text{ and } J_x > \tauxx )
\ \la\  \exp(-c\log\log x) \ea (\log \chi + \log N)^{-c}.
\end{equation*}
By adding the large jump estimate \eqref{eq:large jump estimate} to the above small jump estimate, we conclude that $\Prob_{p,q}( \tauxx \le N, \tauxx < T_m ) \la (\log \chi + \log N)^{-c} + N p^{-1} m^{-1/3}$, where we again use the boundedness of $\chi^{-4/3} \frac{(\log \chi + \log N)^c}{(\log N)^{1+\epsilon}}$.
\end{proof}

\begin{proof}[Proof of Lemma~\ref{lem:one jump diag}]
Let $\Delta_n = |X_n-\mu n| \vee |Y_n-\mu n|$. Recall that $\tauxy = \inf \Set{n\ge 0}{\Delta_n > \barrier(n)}$ where $\barrier[](n) = \mb({ (n+2)(\log(n+2))^{1+\epsilon} }^{3/4}$, and we want to prove that
\begin{equation*}
\lim_{x,m\to\infty} \limsupp \Prob_{p,q}(\tauxy<T_m) =0 \quad\text{while}\quad\frac{q}{p}\in[\lambda',\lambda]\,.
\end{equation*}

Consider the sequences $(N_k)_{k\ge 0}$ and $(x_k)_{k\ge 0}$ defined by $N_0=x_0=0$,
\begin{equation*}
\Delta N_k := N_k - N_{k-1} = 2^k    \qtq{and}
\Delta x_k := x_k-x_{k-1}
= \frac{x}3 \mb({ \Delta N_k \m({ \log \Delta N_k }^{1+\epsilon} }^{3/4}.
\end{equation*}
Then we have $N_k = 2^{k+1}-2$ and
\begin{align*}
x_k = \frac{x}3 \sum_{i=1}^k 2^{\frac34 i} \cdot (i \log 2)^{\frac34(1+\epsilon)}
\le \frac{x}3 \cdot \frac{2^{\frac34 (k+1)}}{2^{3/4}-1} (k \log 2)^{\frac34 (1+\epsilon)}
\le x \mb({ 2^k (\log 2^k)^{1+\epsilon} }^{3/4}.
\end{align*}
In other words, $x_k \le \barrier(N_{k-1})$.

\newcommand*{\kex}{K^\epsilon_{x,m}}%
Consider the sequence of horizontal segments $I_k = \Set{(n,x_k)}{n\in (N_{k-1}, N_k]}$. Due to the previous inequality, all of these segments are below the curve $\Delta_n = \barrier(n)$. Let $\kex$ be the index $k$ where $\Delta_n$ goes above $I_k$ for the first time up to $T_m$, that is,\
\begin{equation*}
\kex = \inf \Set{k\ge 1}{\exists n \in (N_{k-1},N_k] \text{ s.t.\ } \Delta_n>x_k \text{ and } n<T_m} \,.
\end{equation*}
Then we have $\{\tauxy < T_m\}\subset \{\kex < \infty\}$, and an union bound would allow us to restrict our consideration to the set on the right hand side of the inclusion.
Observe that, for any $n\ge 1$, the conditions $\Delta_{N_{k-1}}\le x_{k-1}$ and $\Delta_{n+N_{k-1}}>x_k$ imply $\tilde \Delta_n := |X_{n+N_{k-1}} - X_{N_{k-1}} -\mu n| \vee |Y_{n+N_{k-1}} - Y_{N_{k-1}} -\mu n| > \Delta x_k$.
Therefore by Markov property of $\law_{p,q} \nseq{X_n,Y}$,
\begin{equation*}
\Prob_{p,q}(\kex=k) \le
\EE_{p,q}\m[{ \Prob_{P_{N_{k-1}},Q_{N_{k-1}}}\mb({ \exists n \in (0,\Delta N_k] \text{ s.t.\ } \Delta_n >\Delta x_k \text{ and } n<T_m } \id_{\{\Delta_{N_{k-1}} \le x_{k-1} \}} }.
\end{equation*}
On the other hand, $\Delta_{N_{k-1}}\le x_{k-1}$ also implies
\begin{equation*}
\frac{Q_{N_{k-1}}}{P_{N_{k-1}}}-\lambda\leq\frac{(q-\lambda p)+\mu(1-\lambda)N_{k-1}+(1+\lambda)x_{k-1}}{p+\mu N_{k-1}-x_{k-1}}\leq\frac{(1+\lambda)x_{k-1}}{p-x_{k-1}}
\end{equation*}
and
\begin{equation*}
\frac{Q_{N_{k-1}}}{P_{N_{k-1}}}-\lambda'\geq\frac{(q-\lambda' p)+\mu(1-\lambda')N_{k-1}-(1+\lambda')x_{k-1}}{p+\mu N_{k-1}+x_{k-1}}\geq -\frac{(1+\lambda')x_{k-1}}{p+x_{k-1}}.
\end{equation*}
We note that $x_{k-1}\le x f_\epsilon(\Lambda p)$, which is of smaller order than $p$.
Let $\lambda_{\min}$ and $\lambda_{\max}$ be positive constants such that $\lambda_{\min}<\lambda'\leq 1\leq\lambda<\lambda_{\max}$. Then for $p$ large enough, $\frac{Q_{N_{k-1}}}{P_{N_{k-1}}}\in[\lambda_{\min},\lambda_{\max}]$. In this case, we obtain
\begin{align*}
&\ \EE_{p,q}\m[{ \Prob_{P_{N_{k-1}},Q_{N_{k-1}}}\mb({ \exists n \in (0,\Delta N_k] \text{ s.t.\ } \Delta_n >\Delta x_k \text{ and } n<T_m } \id_{\{\Delta_{N_{k-1}} \le x_{k-1} \}} }\\ &\le
\sup_{p'\ge p-x_{k-1}, \ q'\ge q-x_{k-1}, \ \frac{q'}{p'}\in[\lambda_{\min},\lambda_{\max}]} \Prob_{p',q'}(\tauxx[\Delta x_k] \le \Delta N_k, \tauxx[\Delta x_k] < T_m)\,.
\end{align*}

Let $k_0=k_0(p,q)$ be the largest $k$ such that $N_k \le \Lambda (p\wedge q)$, where $\Lambda \ge 1$ is some cut-off value that will be sent to infinity after $p$, $x$ and $m$. Explicitly, $k_0=\floor{\log_2\left(\frac{\Lambda}{2}(p\wedge q)+1\right)}$, and $\Delta N_{k_0}=\frac{N_{k_0}}{2}+1=O(p)$.
Then, for any fixed $x$, $m$ and in the limit $p,q\to\infty$, we have $\Delta x_k \le \Delta x_{k_0}\leq\frac{\theta}{1+\theta} (p\wedge q)$ and $p-x_{k-1} \ge p- \barrier(\Lambda (p\wedge q)) >\tilde p_\theta$ as well as $q-x_{k-1} \ge q- \barrier(\Lambda (p\wedge q)) >\tilde p_\theta$ for all $k\le k_0$.
Therefore we can apply Lemma~\ref{lem:cst barrier} to bound the above supremum, and obtain for large enough $p,q$ and $k_0$ that
\begin{align*}
\Prob_{p,q}(\kex \le k_0) & \la \sum_{k=1}^{k_0} \m({
\frac1{ (\log (x/3) + \log(\Delta N_k))^{1+\epsilon/2}} + \frac{\Delta N_k}p m^{-1/3} }
\\ & = \sum_{k=1}^{k_0} \frac1{ (\log(x/3)+k \log 2)^{1+\epsilon/2} } + \frac{N_{k_0}}p m^{-1/3}
\la \frac1{ (\log x)^{\epsilon/2} } + \Lambda m^{-1/3}\,.
\end{align*}
On the other hand, $k_0<\kex <\infty$ implies $T_m > N_{k_0}$. Therefore 
\begin{align*}
\Prob_{p,q}(k_0 < \kex < \infty) \ &\le\
\Prob_{p,q}(T_0 > N_{k_0},\ k_0 < \kex < \infty) \\ &=\
\EE_{p,q}\left(\Prob_{P_{N_{k_0}-1},Q_{N_{k_0}-1}}(T_0\neq 1)\id_{(T_0>N_{k_0}-1)}\id_{(k_0 < \kex < \infty)}\right).
\end{align*}
Now $T_0>N_{k_0}-1$ implies $P_{N_{k_0}-1}\ge 1$ and $Q_{N_{k_0}-1}\ge 1$, and together with $k_0<\kex$ also $\Delta_{N_{k_0}-1}\le x_{k_0}$. This yields the estimate
\begin{equation*}
\lambda'-\frac{(1+\lambda')x_{k_0}}{p+x_{k_0}}\leq\frac{Q_{N_{k_0}-1}}{P_{N_{k_0}-1}}\leq\lambda+\frac{(1+\lambda)x_{k_0}}{p-x_{k_0}}.
\end{equation*}
Therefore, for $p$ large enough, we have $0<\lambda_{\min}<\lambda'-\frac{(1+\lambda')x_{k_0}}{p+x_{k_0}}<1<\lambda+\frac{(1+\lambda)x_{k_0}}{p-x_{k_0}}<\lambda_{\max}<\infty$. On the other hand, for $p', q'>0$ such that $q'/p'\in[\lambda_{\min},\lambda_{\max}]$, we have $\Prob_{p',q'}(T_0=1)\sim -\nu_c t_c\frac{4}{3}\frac{a_0 a_1}{b c(q'/p')}(p')^{4/3}(q')^{-7/3}$. Thus, there exist a constant $\delta=\delta(\lambda_{\min},\lambda_{\max})>0$ such that
\begin{equation*}
\Prob_{p',q'}(T_0\ne 1)\leq 1-\frac{\delta}{p'}.
\end{equation*}
We also have the trivial estimate $P_n\le p+2n$. In the end, we conclude \begin{align*}
\EE_{p,q}\left(\Prob_{P_{N_{k_0}-1},Q_{N_{k_0}-1}}(T_0\neq 1)\id_{(T_0>N_{k_0}-1)}\id_{(k_0 < \kex < \infty)}\right)&\leq \prod_{n=0}^{N_{k_0}-1}\left(1-\frac{\delta}{p+2n}\right)\leq\exp\left(-\sum_{n=0}^{N_{k_0}-1}\frac{\delta}{p+2n}\right)\\ &\leq
\exp\left(-\int_0^{N_{k_0}/p}\frac{\delta dx}{1+2x}\right)=\left(1+2\frac{N_{k_0}}{p}\right)^{-\frac{\delta}{2}}\\ &\leq 2^{-\frac{\delta}{2}}\left(\frac{N_{k_0}}{p}\right)^{-\frac{\delta}{2}}.
\end{align*}

Since $N_{k_0}\geq 2\left(2^{k_0-1}-1\right)\geq\frac{\Lambda}{2}(p\wedge q)-1$, it follows that $\left(\frac{N_{k_0}}{p}\right)^{-\frac{\delta}{2}}\la\Lambda^{-\frac{\delta}{2}}$.
We conclude that for every fixed $\Lambda>0$, and uniformly for $x>0$ and $m\ge 1$,
\begin{equation*}
\limsupp \Prob_{p,q}(\tauxy < T_m) \le
\limsupp \Prob_{p,q}(\kex < \infty) \la (\log x)^{-\epsilon/2} + \Lambda m^{-1/3} + \Lambda^{-\frac{\delta}{2}}.
\end{equation*}
Taking the limit $m,x\to\infty$ and then $\Lambda \to\infty$ finishes the proof.
\end{proof}

\paragraph{Acknowledgements.}
The authors would like to thank M. Albenque, J. Bouttier, S. Charbonnier, K. Izyurov, M. Khristoforov, A. Kupiainen and L. Ménard for enlightening discussions. They also thank the anonymous reviewers whose many suggestions helped to improve the presentation of the paper.

\paragraph{Funding.}
Both authors have been primarily supported by the Academy of Finland via the Centre of Excellence in Analysis and Dynamics Research (project No. 271983), as well as the ERC Advanced Grant 741487 (QFPROBA). The first author has been supported by Swiss National Science Foundation (SNF) Grant 175505 during the revision of this paper. The second author also acknowledges support from the Icelandic Research Fund (Grant Number: 185233-051) as well as from the LABEX MILYON (ANR-10-LABX-0070) of Université de Lyon, within the program «~Investissements d’Avenir~» (ANR-11-IDEX-0007) operated by the French National Research Agency (ANR).

\paragraph{Competing interests.}
The authors have no relevant financial or non-financial interests to disclose.


\bibliographystyle{abbrv}
\bibliography{database-linxiao_JT_5}

\end{document}